\documentclass{comnet}
\usepackage{geometry}                % See geometry.pdf to learn the layout options. There are lots.
\usepackage[utf8]{inputenc}
\usepackage{graphicx}
\usepackage{amssymb}
\usepackage{epstopdf}
\usepackage{placeins}
\usepackage{amsmath}
\usepackage{amsthm}
\usepackage{amsfonts}
\usepackage{amssymb}
\usepackage{color}
\usepackage{times}
\usepackage{booktabs}
\usepackage{multirow}
\usepackage{url}
\usepackage{setspace}

\newtheorem{theoremSI}{Theorem}
\newtheorem{corollarySI}{Corollary}[theoremSI]

\begin{document}

\title{The distance backbone of complex networks}

\shorttitle{Distance backbone of complex networks}

\shortauthorlist{Simas, Correia \& Rocha}

\author{%
\name{Tiago Simas,${}^{1}$ Rion Brattig Correia${}^{2,3,4}$ \& Luis M. Rocha${}^{2,3,5*}$}
\address{${}^{1}$Universidade Lusófona, Lisboa, Portugal\\
         ${}^{2}$Center for Social and Biomedical Complexity, Luddy School of Informatics, Computing \& Engineering, Indiana University, Bloomington IN, USA \\
         ${}^{3}$Instituto Gulbenkian de Ci\^{e}ncia, Oeiras, Portugal \\
         ${}^{4}$CAPES Foundation, Ministry of Education of Brazil, Bras\'{i}lia, DF, Brazil\\
         ${}^{5}$Department of Systems Science and Industrial Engineering, Binghamton University, Binghamton, NY 13902\email{$^*$Corresponding author: rocha@binghamton.edu}}}

\maketitle

\begin{abstract}
%\begin{sciabstract}
{
%Redundancy is known to be a major factor in the evolution and robustness of the networks of multivariate interactions that define complex systems, but its precise characterization is still lacking.  We present a general methodology to infer transitivity in the connectivity of such interactions that includes all possible measures of path length for weighted graphs. We show that it induces a distance backbone subgraph that is sufficient to compute all shortest paths, without breaking the graph into smaller components.  This backbone is very small in networks across domains ranging from air traffic to the human brain connectome---a redundancy that makes transport and communication on them robust to attacks. In addition to yielding a principled graph reduction technique, the general methodology provides a finer characterization of the triangular geometry of shortest paths, which is important to understand the dynamics of spreading and communication phenomena on real-world networks.
%
Redundancy needs more precise characterization as it is a major factor in 
%complex systems defined by 
the evolution and robustness of networks of multivariate interactions. We investigate the complexity of such interactions by inferring a connection transitivity that includes all possible measures of path length for weighted graphs. The result, without breaking the graph into smaller components, is a distance backbone subgraph sufficient to compute all shortest paths. This is important for understanding the dynamics of spread and communication phenomena in real-world networks. The general methodology we formally derive yields a principled graph reduction technique and provides a finer characterization of the triangular geometry of all edges---those that contribute to shortest paths and those that do not but are involved in other network phenomena. We demonstrate that the distance backbone is very small in large networks across domains ranging from air traffic to the human brain connectome, revealing that network robustness to attacks and failures seems to stem from surprisingly vast amounts of redundancy.}
{Complex Systems, Network Backbones, Redundancy, Shortest Paths}
\end{abstract}

%Teaser Sentence: This one-sentence summary (125 characters with period) provides a snapshot of your research for non-specialist readers and should complement rather than repeat the title.

%TEASER SENTENCE: Large real-world networks are very robustly redundant with small distance backbones sufficient to compute all shortest paths.

%\end{sciabstract}

\section{Introduction: Redundancy in the triangular organization of complex networks}
\label{section_intro}

%\subsection{Redundancy in the triangular organization of complex networks}
%\label{section_intro_motivation}

Humans excel at the activity of associating objects and concepts.
Indeed, significant scientific advances have come from characterizing multivariate associations as \emph{complex networks} \cite{vespignani2018twenty}.
Examples include: interactions between suppliers and consumers in an electrical grid, friendships and trust relationships among people, correlations in gene regulation data, connections among neurons, and many others \cite{barabasi2016network}.
Several sophisticated mathematical methods have been used to model multivariate associations, including hypergraphs \cite{Klamt:2009kx}, relations \cite{Klir1995,Mordeson2000}, simplicial complexes \cite{salnikov2018simplicial}, and dynamical systems theory \cite{strogatz2014nonlinear}.
In network science, complex networks have been studied mostly using graph theory \cite{newman2011complex,vespignani2018twenty}.
Graphs are intuitive and algorithmically simpler than the alternatives and have been used to model the
Internet \cite{pastor_vesp}, the World Wide Web \cite{albert-2002-74},
collaboration networks \cite{barrat-2004-101,Newman2001},
biological networks \cite{Oltvai2002}, the human brain \cite{sporns2011networks}, and many other types of multivariate associations and interactions \cite{barabasi2016network}.

Most advances in network science have come from the study of patterns of connectivity (network structure) that has provided many insights into the organization of complex systems. Much remains to be understood, however, about how the structure of networks affects the dynamics and robustness of complex systems \cite{vespignani2018twenty,newman2011complex}. For instance, in human brain networks we do not know how synaptic connectivity leads to the dynamical patterns of functional connectivity responsible for behavior \cite{sporns2011networks}. In systems biology we know much about the connectivity patterns of gene and protein regulatory networks from existing models \cite{helikar2012cell}, however, we also know that the structure of interactions from these models is not sufficient to predict regulatory dynamics or derive control strategies that allow us to, for instance, revert a diseased cell to a healthy state \cite{gates2016control}. Similar issues arise with the large-scale collection of social behavior data from social media and mobile devices, which has sparked much additional interest in network science \cite{correia2020mining,pescosolido2016linking}. The structure of social interactions can help us understand aspects of health and disease, such as the spread of pandemics \cite{kraemer2020effect,wang2020response} and detection of drug interactions \cite{correia2016monitoring}, but understanding the dynamical processes of these networks is required for us to be able to predict and control biomedical phenomena.

We address the link between structure and dynamics by exploring important \emph{patterns of redundancy} that contribute to how structure affects dynamics in networks.
Redundancy is thought to be a major factor in the evolution of complex systems \cite{conrad1990geometry} but a precise characterization of how it affects complex network dynamics is still lacking.
A full understanding of the interplay between network structure and dynamics requires a study of multivariate dynamics \cite{gates2016control} and its redundancy \cite{gates2021EfGraph}.
However, most often we do not possess enough time-resolved data or computational power to precisely characterize the multivariate dynamics of large networks.
In these cases, network structure is still very useful to understand the dynamics of spread and communication phenomena, which can be inferred from shortest paths between variables. Therefore, we focus on transitivity in network connectivity, which leads to much redundancy in the computation of shortest paths.

We formally develop a general methodology to infer transitivity that includes all possible measures of distance and path length for \emph{weighted graphs} and show using real world examples that their \emph{triangular organization} induces substantial redundancy in shortest-path computation in networks across many domains.
Such redundancy is very large in networks ranging from air traffic to the human brain connectome, which makes transport and communication on these networks robust to attacks and failures.
The general methodology provides a finer characterization of the behavior of shortest paths on networks than existing measures in network science, such as betweeness centrality, do.
The methodology also reveals that the edges that do not participate in the shortest paths of real-world networks tend to vary widely in how much they distort a given triangular geometry, which suggest a complex robustness mechanism.

Weighted graphs, where every edge is denoted with a positive real number, are often used to capture distance associations between linked nodes within a set of node variables. These networks are useful as knowledge graphs for big data inference to, for example, infer drug interactions from social media and electronic health records using shortest-path calculations \cite{correia2020mining,correia2019city,correia2016monitoring} or automate fact-checking using the Wikipedia knowledge graph \cite{ciampaglia2015computational}.
We have shown that such distance graphs obtained from real-world data are typically not metric, but rather \emph{semi-metric} \cite{simas_rocha_2014_MWS}, in that the triangle inequality of metric spaces is not observed for every edge in the graph. That is, the shortest distance between at least two nodes in the graph is not the direct edge between them but rather an indirect path via other nodes. 

Mathematically there are infinite ways to compute shortest-paths on these distance graphs, each isomorphic to a particular transitive closure \cite{simas_rocha_2014_MWS}. For instance, computing the most typical shortest path measure (e.g., via Dijkstra's algorithm), where the path length is the sum of the constituent edge weights (distances), is isomorphic to the particular transitive closure we refer to as the \emph{metric closure}. It enforces the triangle inequality on the closed graph thus: if an edge in the original graph is semi-metric, its weight is replaced by the length of the shortest indirect path between the nodes it links \cite{simas_rocha_2014_MWS}.
This generalizes to all forms of computing path length, whereby some distance edges obey a generalized triangle inequality---those we refer to as \emph{triangular} edges---but  many others do not---those we refer to as \emph{semi-triangular} edges.
Interestingly, the triangular edges constitute an \emph{invariant sub-graph} of the original graph that does not change with the computation of a given transitive closure computation and is sufficient to compute all shortest paths. 

We refer to this subgraph as the \emph{distance backbone} of a complex network (conceptualized as a weighted distance graph). The amount of redundancy in the network is defined by the size of the backbone subgraph in relation to the size of the original graph. Edges not on the distance backbone are superfluous in the computation of shortest paths, as well as in all network measures derived from shortest paths (e.g., efficiency, and betweenness centrality).
Moreover, because distance backbones preserve all shortest-path connectivity, we show that they contain all network bridges and do not break networks into constituent components as other network reduction and backbone techniques (e.g., thresholding) do.

We show that there is typically massive redundancy in graphs obtained from various types of data ranging from topical spaces of large document corpora to brain networks. For instance, the knowledge graph of more than 3 million concepts extracted from Wikipedia that has been used for automated fact-checking is 98\% semi-metric \cite{ciampaglia2015computational}. This means that its metric backbone contains only 2\% of the original edges and those are sufficient to compute all shortest-paths of the original graph used to infer factual associations. In social contact networks, the metric backbone is typically between 8\% and 30\%. Moreover, the relative size of the backbone in human brain (fMRI) networks has been shown useful in distinguishing healthy cohorts from autistic, depressive, and psychotic ones \cite{simas2015semi,simas2016commentary}.

Finally, we show that while semi-triangular edges are redundant in computing shortest paths, they distort the triangular geometry of the resulting distance backbone to varying degrees. Therefore, the amount of redundancy in a graph \textit{together} with the distribution of this topological distortion provide a nuanced characterization of all edges and their importance for robustness of shortest paths in a network. Altogether, our distance backbone analysis contributes to a better understanding of information transmission in complex networks.
Following, we introduce the necessary mathematical background.

%\subsection*{Background}

\section{Background: Closure of Proximity and Distance Graphs}
\label{section_back}
%\label{section_prox_dist_graphs}

A graph $P(X)$ is a binary relation of set $X$ with itself that characterizes the network of interactions amongst its $n=|X|$ variables $x$.
The nodes (or vertices) of $P$ are each of the variables $x$, and the edges (or links) between two variables $x_i$ and $x_j$ are denoted by $p_{ij}$.
In the simplest case, $p \in \{ 0, 1 \}$, the variables are either related (1) or not (0).

To characterize network interactions in a more natural way, it is very often important to consider intensity \cite{barrat-2004-101,barrat_berth_vesp_book_2008,Newman2001_PREII,Barrat2004PRL,Wang2005a,Goh2005} by allowing edges to be weighted. Weights can denote a \emph{proximity} (also known as similarity or strength) or a \emph{distance} (also known as dissimilarity) between nodes.
Proximity is proportional and distance is inversely proportional to the intensity of the interaction.
In the case of proximity, without loss of generality, edge weights can be normalized to the unit interval:  $p
\in [0, 1]$ \cite{Klir1995,simas_rocha_2014_MWS}.
Thus \emph{proximity graphs} $P(X)$ can be represented by
\emph{adjacency matrices} $P$ of size $n \times n$,  where entries denote the edge weights $p_{ij} = P[i,j], \forall x_i,x_j \in X$.
%
%In many complex network applications it is also important to relate variables in terms of distance. This allows us to define concepts such as path length and shortest paths on weighted graphs, which are essential to uncover community structure, small-world topology, and the like \cite{fortunato_review}.
%
When $P(X)$ is an undirected graph, $P$ is a reflexive and symmetric matrix  ($p_{ii}=1 \, \wedge \, p_{ij}=p_{ji}, \forall x_i,x_j \in X$).
%; we show below that a nonlinear isomorphism $\varphi$ between proximity and distance graphs is possible in this case\footnote{We restrict our analysis to reflexive and symmetric binary relations because distance is necessarily symmetric (and anti-reflexive, $d_{ii}=0$) \cite{galvin_shore91}.
%and we explore the isomorphism with distance graphs. 
%The methodology, however, is extendable to non-symmetric graphs (see \cite{simas_rocha_2014_MWS} for details).}.
%
%The simplest such isomorphism is the familiar $\varphi: d_{ij} = 1/p_{ij} - 1$, though an infinite number of functions can satisfy the isomorphism constraints (see below).

\subsection{Transitive closure}
\label{section_TC}

% Transitivity

Transitivity is an important concept in complex networks because it allows for the inference of indirect associations from data \cite{louch2000personal, burda2004network}. The transitive closure in proximity graphs and (indirect) shortest paths in distance graphs are isomorphic means to infer the propensity of variables that do not interact directly to affect one another indirectly via network interactions.
Because many forms of transitivity can be defined, there are many distinct ways to compute such indirect interactions in networks.

Weighted graphs $P(X)$ obtained from real-world data using various measures of proximity (e.g., co-occurrence, correlation, mutual information) are typically not transitive. That is, at least one pair of nodes, $x_i$ and $x_j$, are linked more strongly via a third node, $x_k$, than they are directly.
%
%
% t-norms and t-conorms
%
To study the transitivity of a weighted graph, we need to compute the strength of interaction between any two nodes given all possible indirect paths between them.
There are, however, infinite ways to numerically integrate the weights on the indirect paths.
Menger \cite{Menger} first generalized \emph{transitivity criteria} in the context of probabilistic metric spaces by introducing \emph{triangular norm} (t-norm) binary operations.

Later, Zadeh used t-norms to generalize logical operations in multi-valued logics such as Fuzzy logic \cite{Zadeh1965}, as follows.
A \emph{t-norm} $\land:[0,1] \times [0,1] \rightarrow [0,1]$, is a binary operation with the properties of commutativity ($a \land b = b \land a$), associativity ($a \land (b \land c) = (a \land b) \land c$), monotonicity ($a \land b \leq c \land  d$ iff $a\leq c$ and $b\leq d$), and identity element $1$ ($ a \land 1 = a$).
In other words, the algebraic structure $([0,1], \land)$ is a monoid \cite{Gondran2007}, and $\land$ generalizes \emph{conjunction} in logic to deal with truth values in the unit interval ($a,b \in [0,1]$) (for a thorough review, see \cite{Klir1995}).
Similarly, a \emph{t-conorm} $\lor$ generalizes \emph{disjunction} and  has the same properties as a t-norm, but its identity element is $0$ ($a \lor 0 = a$) \cite{Klir1995}.
To obtain \emph{dual} t-norm/t-conorm pairs, we can derive a t-conorm from a t-norm via a generalization of De Morgan's laws: $a \lor b = 1 - ((1-a) \land (1-b))$.

%%% Composition

To compute the transitivity of a graph $P(X)$, we use the \emph{composition of binary relations} based on the algebraic structure $([0,1], \land, \lor )$
%, which is not necessarily a semiring \cite{simas_rocha_2014_MWS} 
via the logical composition of the graph's adjacency
matrix with itself ($P \circ P$). The calculation procedure is similar to a matrix product, except that summation and multiplication are replaced by the t-conorm and t-norm, respectively \cite{Klir1995,Klement2004}:

$$P \circ P = \bigvee_{k}
\bigwedge (p_{ik},p_{kj})=p'_{ij}.$$

The most commonly used generalized disjunction (t-conorm) and conjunction (t-norm) pair is:  $\langle \lor \equiv$ \texttt{maximum}$, \land \equiv$ \texttt{minimum} $\rangle$.
Thus, the standard graph composition is referred to as the
\emph{max-min composition}:

 $$P \circ P = \max_{k}
\min (p_{ik},p_{kj})=p'_{ij}.$$

%%% Transitive Closure

The \emph{transitive closure} $P^T(X)$ of a graph $P(X)$ can then be defined as:

\begin{equation}
\label{eq_TC1}
%R^{T}=R \cup R^2 \cup \dots \cup R^k
P^{T} = \bigcup_{\eta=1}^{\kappa} P^{\eta},
\end{equation}

\noindent where $P^{\eta} = P \circ P^{\eta - 1}$, for $\eta = 2, 3, ...$, and $P^1 = P$ \cite{Klir1995}.
Furthermore, $P \cup Q$ denotes the union of two graphs defined on same node set $X$ and is defined by the disjunction of their respective adjacency matrix entries: $p_{ij} \lor q_{ij},\,\, \forall_{x_i,x_j \in X}$, where $\lor$ denotes the same t-conorm used in the composition.
In the most general case, $\kappa \rightarrow \infty$ \cite{Gondran2007}, but
when $\vee = \texttt{maximum}$, with any t-norm $\wedge$,
%that obeys the Archimedean property (see appendix 1),
we have $\kappa \leq |X|-1$
%\cite{Pang2003, Klir1995}.
\cite{Klir1995}.
In other words, in the latter case the transitive closure converges in finite time and is easily computed via formula \ref{eq_TC1} \cite{simas_rocha_2014_MWS}.

Because each t-conorm/t-norm pair employed defines a different graph composition, different criteria for transitivity can be computed \cite{Menger,Klir1995,simas_rocha_2014_MWS}.
A graph $P(X)$ is transitive for a t-conorm/t-norm pair $\langle \lor, \land \rangle$ if and only if:

 $$p_{ij} \geq \bigvee_{k}
\bigwedge (p_{ik},p_{kj}), \forall_{x_i, x_j, x_k \in X.}$$

\subsection{Distance closure isomorphism}
\label{section_isomorphism}

%Isomorphism map
%distance closure plus specific cases to relate to APSP
%simplest function

%Transitive Closure is a well established algorithm in graph theory.
%
The study of complex networks, including such phenomena as community structure \cite{fortunato_review}, node and edge centrality \cite{zuo2012network}, and link prediction \cite{martinez2016survey}, is based heavily on a notion isomorphic to transitive closure of a proximity graph $P(X)$:
%
%graphs including weighted graphs \cite{Barrat2008}, and
%
computation of all shortest paths of a \textit{distance graph} $D(X)$, typically via the Dijkstra algorithm \cite{dijkstra}. In this type of weighted graph, the adjacency matrix is denoted by $D$, edge weights $d_{ij} \in [0,\infty]$ denote the intuitive notion of distance and are anti-reflexive ($d_{ii}=0$) and symmetric ($d_{ij}=d_{ji}$)\footnote{Here we only consider the standard definition of distance as a symmetric and anti-reflexive function \cite{galvin_shore91}. 
%and we explore the isomorphism with distance graphs. 
The methodology, however, is extendable to non-symmetric graphs (see \cite{simas_rocha_2014_MWS} for details).}.

The two types of graphs and their closures can be defined as isomorphic via a non-linear (strictly monotonic decreasing) map $\varphi$, since proximity edges are constrained to $[0,1]$ while distance edges are in $[0,+\infty]$ \cite{simas_rocha_2014_MWS}.
The isomorphism between proximity graphs $P(X)$ and distance graphs $D(X)$ must observe the following constraint $\forall x_i,x_j,x_k \in X$:

\begin{equation}
\label{eq_isomorphism}
\mathop{f}\limits_{k} \{g(\varphi (p_{ik}),\varphi (p_{kj})\}=\varphi(\mathop{\lor}\limits_{k}\{ \land (p_{ik},p_{kj})\}),
\end{equation}

%\[\forall x_i,x_j \in X:\mathop{f}\limits_{k} \{g(\varphi (p_{ik}),\varphi (p_{kj})\}\]\[=\varphi(\mathop{\lor}\limits_{k}\{ \land (p_{ik},p_{kj})\})\]

\noindent where $\langle f,g \rangle$ is a pair of binary operations such that $f,g:[0,+\infty]\times [0,+\infty]\rightarrow [0,+\infty]$.
These operations are the t-conorm and t-norm isomorphic counterparts in distance graphs, so they also form monoids $([0,+\infty],f)$ and $([0,+\infty],g)$, each with properties of commutativity, associativity, monotonicity, and identity element ($+\infty$ for $f$ and 0 for $g$). Thus, $f$ and $g$ have been named a td-conorm and td-norm, respectively \cite{simas_rocha_2014_MWS}, where td is ``triangular distance.''
Though an infinite number of maps satisfy the isomorphism the simplest, which we use here unless otherwise noted, is the familiar distance function:

\begin{equation}
\label{eq_iso_map}
\varphi: d_{ij} = 1/p_{ij} - 1.
\end{equation}

Constraint \ref{eq_isomorphism} also leads to the equations that define each operation in terms of its isomorphic counterpart, where $\varphi^{-1}$ is the inverse function of $\varphi$:

\begin{eqnarray}
\label{eq_isomorphism_formulae}
g(d_{ik},d_{kj}) &=& \varphi(\land(\varphi^{-1}(d_{ik}),\varphi^{-1}(d_{kj}))) \nonumber \\
f(d_{ik},d_{kj}) &=& \varphi (\lor( \varphi^{-1}(d_{ik}),\varphi^{-1}(d_{kj})))\nonumber \\
\lor(p_{ik},p_{kj}) &=& \varphi^{-1}(f(\varphi(p_{ik}),\varphi(p_{kj}))) \nonumber \\
\land(p_{ik},p_{kj}) &=& \varphi^{-1}(g(\varphi(p_{ik}),\varphi(p_{kj}))).
\end{eqnarray}

%%% Transitive to general distance closure

As depicted in Figure \ref{fig:isomorphism-shortest-path}A, given a (strictly monotonic decreasing) map $\varphi$, the transitive closure $P^T(X)$ of graph $P(X)$ possesses an isomorphic \emph{distance closure} $D^T(X)$ of graph $D(X)$, and vice-versa.
This means that every transitive closure criterion established by a $\langle \land , \lor \rangle$-composition for proximity graphs yields an isomorphic distance closure criterion for distance graphs established by an $\langle f,g \rangle$-composition \cite{simas_rocha_2014_MWS}.
Therefore, there are infinite ways to algebraically compute indirect associations between nodes in a distance graph and shortest path computation is just one possibility.
In the range of all possibilities, many are meaningful and useful for complex network applications, including diffusion distances and the class of shortest-path distance closures we focus on below \cite{simas_rocha_2014_MWS}.

\begin{figure}[h]
    \centering
    \includegraphics[width=\textwidth]{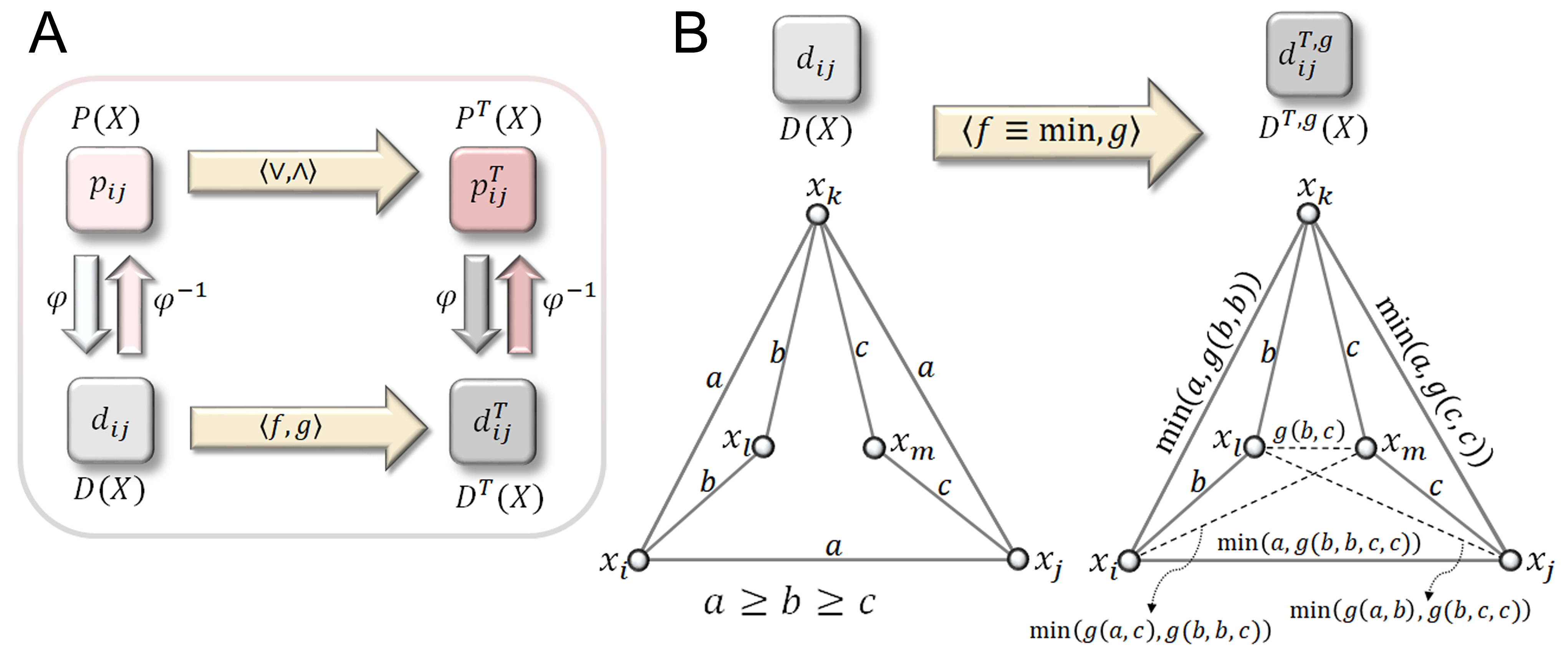}
    \caption{
        \textbf{General Isomorphism and Shortest-Path Distance Closure.} 
        (\textbf{A})
        Proximity and distance graphs with their respective transitive and distance closures.
        (\textbf{B})
        Schematic of the transformation of a distance graph $D(X)$ and its adjacency matrix with (symmetric) entries $d_{ij}$ into a distance closure graph $D^{T,g}(X)$ and its adjacency matrix with (symmetric) entries $d^{T,g}_{ij}$. This class of closure fixes $f \equiv \min$ so that the path with shortest length between each pair of nodes $x_i$ and $x_j$ is selected, while $g$ can be any td-norm and defines how path length is computed (eq. \ref{eq_length}).
        On the left, an example distance graph with 5 nodes and edge weights given by constraint $a\geq b \geq c$. On the right, its shortest-path distance closure;
        edges that do not exist in $D(X)$ are shown as dashed in $D^{T,g}(X)$.
    }
    \label{fig:isomorphism-shortest-path}
\end{figure}

\subsection{Shortest-path distance closures}
\label{section_shortest_path_DC}

%%% introduce the  special cases

%1. Max/Min, ultra-metric distance closure
%2. Metric Closure/APSP
%3. Dombi/Hamacehr, Diffusion distance closure

%

%When $\lor \equiv \max $, for any strictly monotonic decreasing isomorphism map $\varphi$ (such as eq. \ref{eq_iso_map}),

The most common distance closure is the \emph{metric closure}, where in our formulation,
%
%$\langle f , g \rangle \equiv \langle \min , + \rangle$.
%
$f(z,w)=\min(z,w)$ and $g(z,w)=z+w$, for $z,w \in [0, +\infty]$.
%LMR: x is already used for nodes of graph
%
This type of closure computes the \emph{shortest path} between all nodes $x$ in a distance graph $D(X)$: operation $g \equiv +$ is used to compute path length by summing the edge distance weights of each path, and operation $f\equiv \min $ is used to select the shortest path length\footnote{Using the simple map of eq. \ref{eq_iso_map}, the isomorphic counterpart of metric closure in proximity graphs becomes the transitive closure based on
$\lor(a,b) = \max(a,b)$ and $\land(a,b) = ab/(a+b-ab)$ for $a,b \in [0,1]$ %a t-conorm/t-norm pair which does not satisfy De Morgan's laws
\cite{simas_rocha_2014_MWS}.}.
The metric closure is equivalent to the \emph{All Pairs Shortest Paths} (APSP) problem \cite{Zwick}. It is typically computed via the  Dijkstra algorithm \cite{Brandes} but it can also be computed algebraically with the matrix composition of distance closure or the isomorphic transitive closure (formula \ref{eq_TC1}), also known as the \emph{distance product} \cite{Zwick}.
Once the metric closure of distance graph $D(X)$ is computed the resulting graph $D^{T,m}(X)$ is guaranteed to be metric: Every edge weight obeys the triangle inequality or $d^{T,m}_{ij} \leq d^{T,m}_{ik} + d^{T,m}_{kj}, \forall x_i,x_j,x_k \in X$. This means that the shortest distance between any two nodes in $D^{T,m}(X)$ is the direct edge weight that links them and no indirect path adds up to a shorter distance.

Another noteworthy distance closure is the isomorphic counterpart of the max-min transitive closure of proximity graphs, $\langle \lor, \land \rangle \equiv \langle \max, \min \rangle$.
For any strictly monotonic decreasing map $\varphi$, this transitive closure  is equivalent to the \emph{ultra-metric closure}, $D^{T,u}(X)$, of distance graph $D(X)$, defined by operations $\langle f, g \rangle \equiv \langle \min, \max \rangle$.
Instead of the triangle inequality, this closure enforces a stronger inequality: $d^{T,u}_{ij} \leq \max (d^{T,u}_{ik}, d^{T,u}_{kj}), \forall x_i,x_j,x_k \in X$ \cite{simas_rocha_2014_MWS}.
In this case, instead of computing path length by summing the edges in a path (as in the metric closure), path length is the ``weakest link'' in the path, either the largest distance edge-weight or the smallest proximity edge-weight.
In other words, the shortest path between any two nodes in the original distance graph is computed as the \emph{minimax path} \cite{Camerini197810}, which leads to an ultra-metric distortion of the original graph topology \cite{Rammal_ultrametricity}.

The metric and ultra-metric closures are special cases of the class of \emph{shortest-path distance closures} \cite{simas_rocha_2014_MWS}
%; see Figure \ref{fig:generalized_shortest_path_closure} for an example.
%
in which the t-conorm is fixed to $\lor \equiv \max$ in the transitive closure of proximity graphs, or equivalently, the td-conorm is fixed to $f \equiv \min$ in the distance closure of distance graphs.
This ensures that the closure converges in finite time and is easily computed via eq. \ref{eq_TC1} \cite{Klir1995,simas_rocha_2014_MWS}.
While there is an infinite set of such closures, because $f \equiv \min$, every possible closure results from choosing the shortest path between every pair of nodes. What changes is how one computes the \emph{length} of a path between $x_i$ and $x_j$ via $m$ indirect (non-repeating) nodes:

\begin{equation}
\label{eq_length}
\ell_{ij}= g(d_{ik_1},d_{k_1k_2},\ldots,d_{k_m j}).
\end{equation}

\noindent The canonical metric closure computes path length as the sum ($g \equiv +$) of all edges in the path.
The ultra-metric closure computes path length by selecting the weakest link  ($g \equiv \max $).

As an illustration, imagine we are interested in computing influence on a social network where nodes are people and edge weights denote social distance: Small edge-weights indicate people are socially near (thus, they can influence each other a lot), and large edge weights indicate they are socially distant. We can compute how indirectly-connected people may influence one another by computing a distance closure. The metric closure, in effect, accounts for each edge on an indirect path between two people in the network as an additive cost: It sums the distances edge-weights on the path to measure indirect social influence.
The ultra-metric closure, on the other hand, assumes that indirect social influence depends only on the weakest link on the path between two people: The largest distance edge-weight on the path is the measure of indirect social influence.

Beyond these two well-known closures, we can compute path length in infinite other ways by setting $g$ in eq. \ref{eq_length} to any td-norm that obeys the isomorphism constraint of eq. \ref{eq_isomorphism}.
For instance, path length could be computed by the Euclidean distance, $g \equiv \left(d_{ik}^2 + d_{kj}^2 \right )^{1/2}$, a more general Minkowski metric, $g \equiv \left({d_{ik}^r + d_{kj}^r} \right )^{1/r}$ \cite{schvaneveldt1990pathfinder} or even an operation that does not sum edge weight contributions, like a product, $g \equiv (d_{ik} + 1) \cdot (d_{kj} +1) -1$. Indeed, we can use the isomorphic counterpart of any of the many families of known t-norms \cite{Klir1995}.
%(see also \S %\ref{section_other_backbones})
%`Other backbones' below).
%
%or an averaging operator ($g \equiv \textrm{AVG} $) such as the arithmetic or geometric mean.
%
Thus, there are also infinite ways to compute indirect associations in complex networks via the \emph{general shortest-path distance closure}, which we denote by $D^{T,g}(X)$; see Figure \ref{fig:isomorphism-shortest-path}B for an example.
For a given td-norm $g$, the \emph{generalized average shortest-path length} (as defined by eq. \ref{eq_length}) is denoted by $\langle d^{T,g}_{ij}\rangle \forall x_i,x_i \in X$, that is, by the mean value of all entries of the adjacency  matrix of distance closure graph $D^{T,g}$\footnote{This assumes a connected graph; otherwise, one ignores the infinite entries of the adjacency matrix.}.

The study of desirable axiomatic characteristics of distance closures---including several specific cases in addition to metric and ultra-metric closures, such as diffusion distance closures---has been pursued by authors Simas and Rocha\cite{simas_rocha_2014_MWS}. Here we focus on an invariant subgraph of all shortest-path distance closures (those with $f \equiv \min$), which allows us to uncover structural \emph{redundancy} and shortest-path \emph{robustness} in graph models of real-world complex networks.

\begin{figure*}[ht!]
    \centering
    \includegraphics[width=\textwidth]{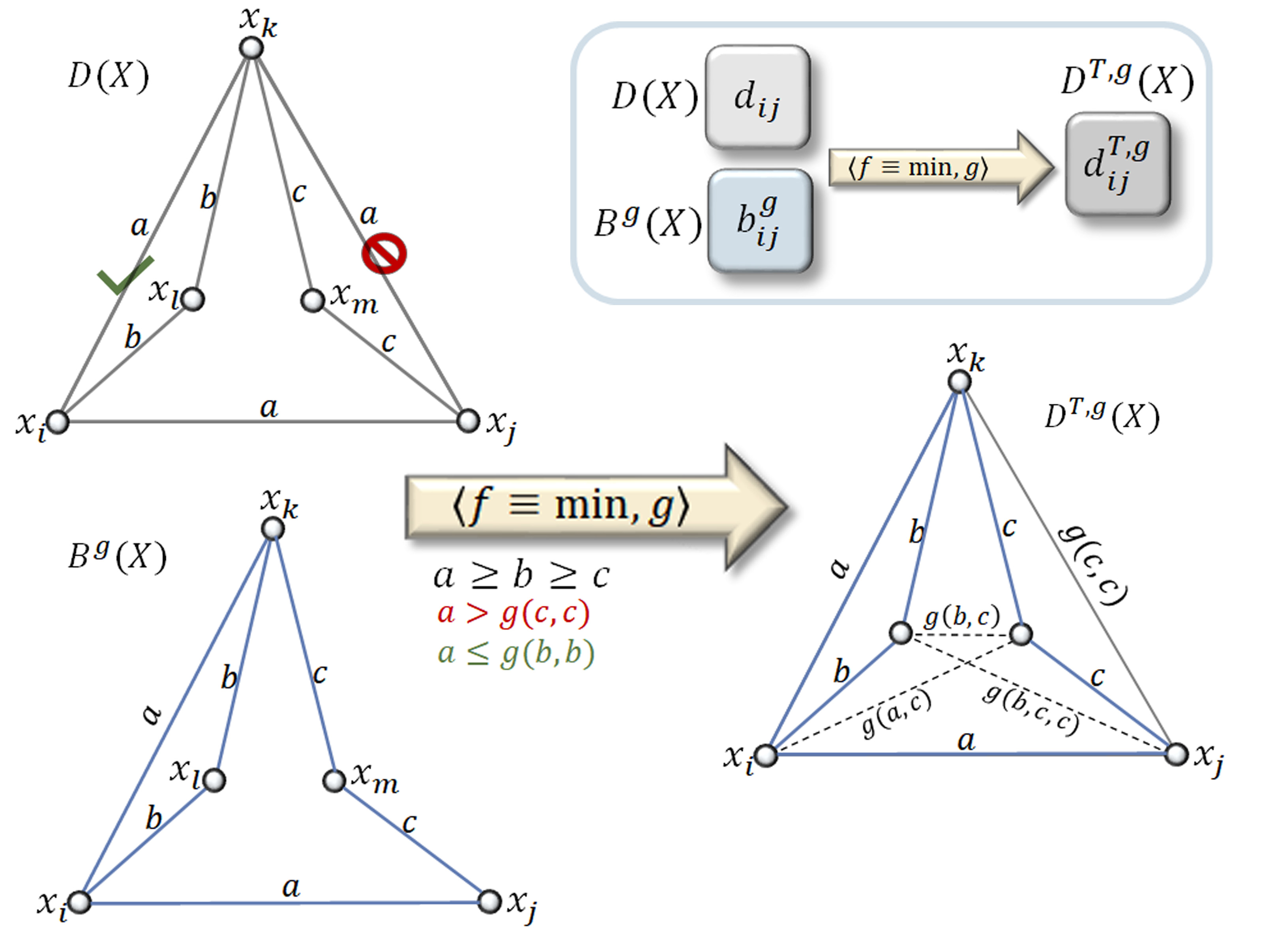}
    \caption{
        \textbf{Distance Backbone.} Top, right: Schematic of the shortest-path distance closure $D^{T,g}(X)$ obtained from either the original distance graph $D(X)$ or its distance backbone $B^g(X)$. Top, left: Example distance graph of 5 nodes with edge distance weights constrained by $a\geq b \geq c$ and by $a > g(c,c)$ and $a \leq g(b,b)$ which, respectively, break and enforce the generalized triangle inequality (eq. \ref{eq_gen_triang}) for nodes $\{x_j, x_k, x_m\}$ and $\{x_i, x_k, x_l\}$. Bottom: The distance backbone graph $B^g(X)$ (left) and the distance closure graph $D^{T,g}(X)$ (right) for any td-norm $g$ given the edge weight constraints considered and td-norm properties; backbone (triangular) edges in blue, semi-triangular edges in gray, and (indirect) edges that do not exist in $D(X)$ appear in dashed gray in $D^{T,g}(X)$.
    }
    \label{fig:distance_backbone}
\end{figure*}

\section{Results}
\label{section_results}

\subsection{Distance backbone: the general case}
\label{section_distance_backbone}

%\subsection*{The general case}
%\label{section_def_backbone}

The \emph{distance backbone} of a distance graph $D(X)$ is defined as its invariant subgraph
%
% $D^{b,g}(X)$
%
%LMR: changed the notation of backbone subgraph to be simpler by using letter B
%
$B^g(X)$
in the computation of a (shortest-path) distance closure $D^{T,g}(X)$ using the binary operation pair $\langle f \equiv \min, g \rangle$.
The edge weights of the distance backbone graph are given by:

\begin{equation}
\label{eq_backbone}
b^g_{ij} = \begin{cases} d_{ij},& \mbox{if } d_{ij}=d^{T,g}_{ij} \\ +\infty,& \mbox{if } d_{ij} > d^{T,g}_{ij} \end{cases}\, , \forall x_i,x_j \in X,
\end{equation}

\noindent where $b^g_{ij} = +\infty $ means that there is no direct edge between $x_i$ and $x_j$ in the distance backbone graph---the direct distance is infinite.
The edge weights of $D(X)$ that do not change after computation of the distance closure $D^{T,g}(X)$ are those that obey the \emph{generalized triangle inequality} imposed by $\langle f \equiv \min, g \rangle$:

\begin{equation}
\label{eq_gen_triang}
d_{ij} \leq g (d^T_{ik}, d^T_{kj}), \forall x_i,x_j,x_k \in X.
\end{equation}

\noindent The edges weights that do become smaller with the distance closure, and only these, break this inequality in $D(X)$ and are not included in the distance backbone.
When an edge $d_{ij}$ of $D(X)$ breaks the generalized triangle inequality, it means that the \emph{length} of at least one indirect path between $x_i$ and $x_j$, as computed by eq. \ref{eq_length}, is shorter than the direct distance: $\ell_{ij} = d^{T,g}_{ij} < d_{ij}$.
%In this general case, the \emph{indirect length} between $x_i$ and $x_j$ is computed as $g(d_{ik}, d_{kj})$, where $g$ is the td-norm of the isomorphism constraint  \ref{eq_isomorphism} (\S \ref{section_isomorphism}).
%
Figure \ref{fig:distance_backbone} depicts the process of computing the distance backbone for any td-norm $g$ and an example (two additional examples are shown in Figures \ref{fig:distance_backbone_2} and  \ref{fig:distance_backbone_3} in SI).

It is only in weighted distance graphs---where weights discriminate and characterize degree of association between nodes as distance---that the concept of distance backbone is meaningful and useful, as summarized by the following theorem (proof in SI):

\begin{theorem} [Backbone of non-weighted graphs]
\label{th_crisp_backbone}
If $D(X)$ is a standard, non-weighted graph, then its distance backbone for any td-norm  $g$ is the entire graph: $B^g(X) \equiv D(X).$
\end{theorem}

%LMR: Added proof to SI

%\begin{proof}
%The proof is straightforward from the definition of td-norm $g$ via its isomorphic t-norm $\wedge$ (\S \ref{section_back}). When graph $D(X)$ is non-weighted, it means that there is no distinguishing characteristic for the weights. In our framework this means that when an edge exists between two nodes $x_i$ and $x_j$, they are considered to be maximally associated: $d_{ij}=0$. Conversely, if there is no edge, the two variables are minimally associated: $d_{ij}=+\infty$\footnote{In the isomorphic space, all connected nodes have maximum proximity, $p_{ij}=1$, and when there is no edge between $x_i$ and $x_j$ we have minimum proximity, $p_{ij}=0$.}.
%%
%Because $0$ is the identity element of any td-norm $g$ (\S \ref{section_back}), the generalized triangle inequality (eq. \ref{eq_gen_triang}) cannot be broken for any edge with $d_{ij}=0$.
%\end{proof}

%
Edges of $D(X)$ that obey the generalized triangle inequality are called \emph{triangular}, and those that do not are called \emph{semi-triangular}, analogous to semimetrics that relax the standard triangle inequality \cite{galvin_shore91}. 
Interestingly, semi-triangular edges are not necessary to compute the distance closure $D^{T,g}(X)$ as they cannot appear in a (generalized) shortest path. Thus, triangular edges alone define the backbone and are sufficient to compute the closure, as the following theorem attests (proof in SI)\footnote{
Naturally, there is an isomorphic \emph{transitive backbone} of proximity graph $P(X)$ via eq.  \ref{eq_isomorphism}. Edges in the transitive backbone are transitive according to the criterion established by $\langle \vee \equiv \max, \wedge \rangle$. Edges not on the backbone break this transitivity criterion, and we can refer to them as \emph{semi-transitive}; see also transitive reduction in
\S \ref{section_discussion}.}:
%\S Discussion.}:

\begin{theorem}[Backbone Sufficiency]
\label{th_backbone_is_sufficient}
Given a distance graph $D(X)$ defined on (node) variable set $X$, its shortest-path distance closure defined by any td-norm $g$ is equivalent to the same closure of its distance backbone subgraph:
$D^{T,g}(X) \equiv B^{T,g} (X)$.
\end{theorem}

%LMR: Added proof to SI

%\begin{proof}
%The proof of this theorem is rather trivial. If an edge $d_{ij}$ of $D(X)$ breaks the generalized triangle inequality (eq. \ref{eq_gen_triang}), we have $d_{ij} > g (d^T_{ik}, d^T_{kj})$.
%%
%Therefore, $d_{ij}$ is not an edge of the distance backbone ($b_{ij} = +\infty$), and there must exist at least one indirect path between $x_i$ and $x_j$ via a set other nodes $x_k \in K \subset X$ such that $\ell_{ij} < d_{ij}$, where $\ell_{ij}$ is the length of the indirect path given by eq. \ref{eq_length}.
%%
%Since the closure computation (eq. \ref{eq_TC1} via isomorphism of eq. \ref{eq_isomorphism_formulae}) selects the shortest path between any pair of nodes ($f \equiv \min$), it cannot select the direct edge $d_{ij}$ for the shortest distance between $x_i$ and $x_j$, but rather the indirect path with smallest length: $d^T_{ij} = \min_K (\ell_{ij})$. Therefore,  $d_{ij}$ is not used to compute $d^T_{ij}$, nor the length of any shortest path that goes through $x_i$ and $x_j$---which must use $d^T_{ij}$ rather than $d_{ij}$, since $d^T_{ij} < d_{ij}$.
%%
%Finally, if $d_{ij}$ does not break the generalized triangle inequality (eq. \ref{eq_gen_triang}), then it is an edge of the distance backbone and is sufficient to compute the smallest path length between $x_i$ and $x_j$, as there cannot be a shorter indirect path between them per eq. \ref{eq_gen_triang}: $b_{ij} = d_{ij} = d^T_{ij}$.
%\end{proof}

From theorem \ref{th_backbone_is_sufficient}, it follows that if the original graph $D(X)$ is connected, then the distance backbone graph is also connected for any td-norm $g$ and it must contain all bridge edges (proofs in SI).

\begin{corollary} [Backbone Connectivity]
\label{th_backbone_is_connected}
Given a connected distance graph $D(X)$, its distance backbone graph $B^g(X)$ is also a connected graph for any td-norm $g$.
\end{corollary}

%\begin{proof}
%If graph $D(X)$ is connected, then there is a path between every pair of nodes $x_i$ and $x_j$ in graph, and its shortest-path distance closure with any td-norm $g$,  $D^{T,g}(X)$, is a complete (fully connected) graph. Since, per Theorem \ref{th_backbone_is_sufficient}, the closure of the backbone graph $B^g$ is sufficient to compute the same (complete) $D^{T,g}(X)$, the backbone graph must be connected as well.
%\end{proof}

\begin{corollary} [Backbone Contains All Bridges]
\label{th_backbone_bridges}
Given a distance graph $D(X)$, all its bridge edges are included in its distance backbone graph $B^g(X)$ for any td-norm $g$.
\end{corollary}

%\begin{proof}
%A bridge is an edge whose deletion increases the graph's number of connected components. Therefore, if a bridge were not present on the backbone graph $B^g(X)$,  Theorem \ref{th_backbone_is_sufficient} and Corollary \ref{th_backbone_is_connected} would be false.
%\end{proof}

%For instance, removal of metric edges from the original distance graph, is more likely to break the graph into modules, since all bridges are in the metric backbone.

These are very useful results because many graph reduction techniques, such as thresholding and other backbones, do not necessarily preserve graph connectivity or bridge edges
(see \S \ref{section_discussion}).
%(\S Discussion).
%
Figure \ref{fig:distance_backbone} exemplifies the backbone sufficiency theorem and connectivity corollary for any td-norm $g$ (as do Figures \ref{fig:distance_backbone_2} and  \ref{fig:distance_backbone_3}  in SI).

\begin{figure*}[ht!]
    \centering
    \includegraphics[width=\textwidth]{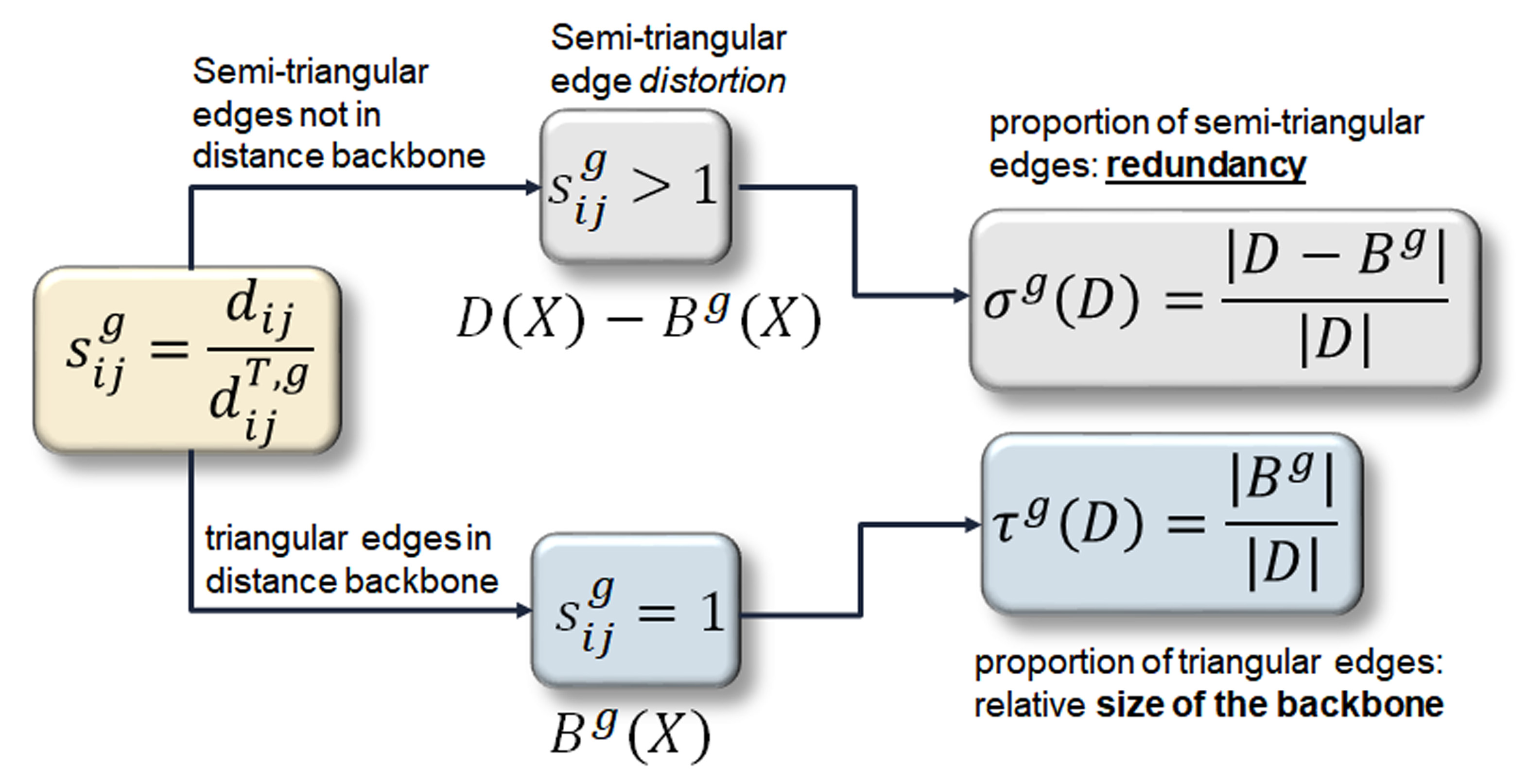}
    \caption{
        \textbf{Semi-triangular measures.}
        Parsing of the two types of distance graph edges, triangular ($s^g_{ij}=1$) and semi-triangular ($s^g_{ij} > 1$), from which graph-level measures of backbone size ($\tau^g(D)$) and redundancy ($\sigma^g(D)$) derive, respectively.
        The measures apply to all distance backbones that derive from a shortest-path distance closure with any td-norm $g$.
    }
    \label{fig:measures_triangular}
\end{figure*}

The \emph{proportion of semi-triangular edges} is therefore the proportion of edges in graph $D(X)$ that are not necessary to compute shortest-paths according to the distance closure defined by $\langle f \equiv \min, g \rangle$. This measure of \emph{edge redundancy} is given by:

\begin{equation}
    \label{eq_percentage_red}
    \sigma^g (D) = \frac{| \{ d_{ij}: d_{ij} > d^{T,g}_{ij} \} |}{| \{ d_{ij} \} |} , \forall_{x_i,x_j \in X: i > j}.
\end{equation}

Similarly, the \emph{proportion of triangular edges} in graph $D(X)$ is simply the \emph{relative size of its distance backbone} $B^g(X)$:

\begin{equation}
    \label{eq_percentage_backbone}
    %\begin{split}
    %\tau^g (D) & = \frac{| \{ d_{ij}: d_{ij} = d^{T,g}_{ij} \} |}{| \{ d_{ij} \} |} \\
    % & = \frac{| \{ b^g_{ij} \} |}{| \{ d_{ij} \} |} , \forall_{x_i,x_j \in X: i > j j}
    %\end{split}
    %
    \tau^g (D) = \frac{| \{ d_{ij}: d_{ij} = d^{T,g}_{ij} \} |}{| \{ d_{ij} \} |} = \frac{| \{ b^g_{ij} \} |}{| \{ d_{ij} \} |} , \forall_{x_i,x_j \in X: i > j}.
\end{equation}

\noindent It follows that $\tau^g = 1 - \sigma^g$. Since distance graphs are symmetric ($d_{ij} = d_{ji}$), and edges are nondirected, in formulae \ref{eq_percentage_red} and \ref{eq_percentage_backbone} we count each edge only once and do not tally reflexive edges $d_{ii}$. That is, we tally only the lower diagonal entries of the adjacency matrix:  $d_{i,j}: i > j$.

Operation $g$ instantiates a specific length measure for indirect paths on a distance graph $D(X)$ that results in a specific shortest-path distance closure $D^{T,g}(X)$.
Each such closure induces a topological \emph{distortion} \cite{simas_rocha_2014_MWS} of the original graph obtained from multivariate associations observed in real-world data, whereby semi-triangular edges are forced to conform to the respective triangle inequality given by eq. \ref{eq_gen_triang}.
However, only the semi-triangular edges get distorted; the triangular edges and the distance backbone they compose remain invariant.
Therefore, $\sigma^g$ also denotes the proportion of edges topologically distorted by a distance closure, whereas $\tau^g$ denotes the proportion of topologically invariant edges.

A measure of \emph{semi-triangular edge distortion} is easily obtained via a ratio of the direct distance over the shortest indirect path length between nodes  $x_i$ and $x_j$:

\begin{equation}
    \label{eq_edge_distortion}
    s^g_{ij} =\frac{d_{ij}}{d^{T,g}_{ij}}, \forall_{x_i,x_j \in X: i \neq j}.
\end{equation}

\noindent If an edge $d_{ij}$ is triangular, $s^g_{ij}=1$, meaning there is no distortion. If an edge is semi-triangular, $s^g_{ij}>1$, and the larger the ratio, the more the edge breaks the general triangle inequality and, thus, the more distorted it will be in the distance closure.
Figure \ref{fig:measures_triangular} depicts how measures \ref{eq_percentage_red}-\ref{eq_edge_distortion} relate to triangular and semi-triangular edges.

\begin{figure*}[ht!]
    \centering
    \includegraphics[width=\textwidth]{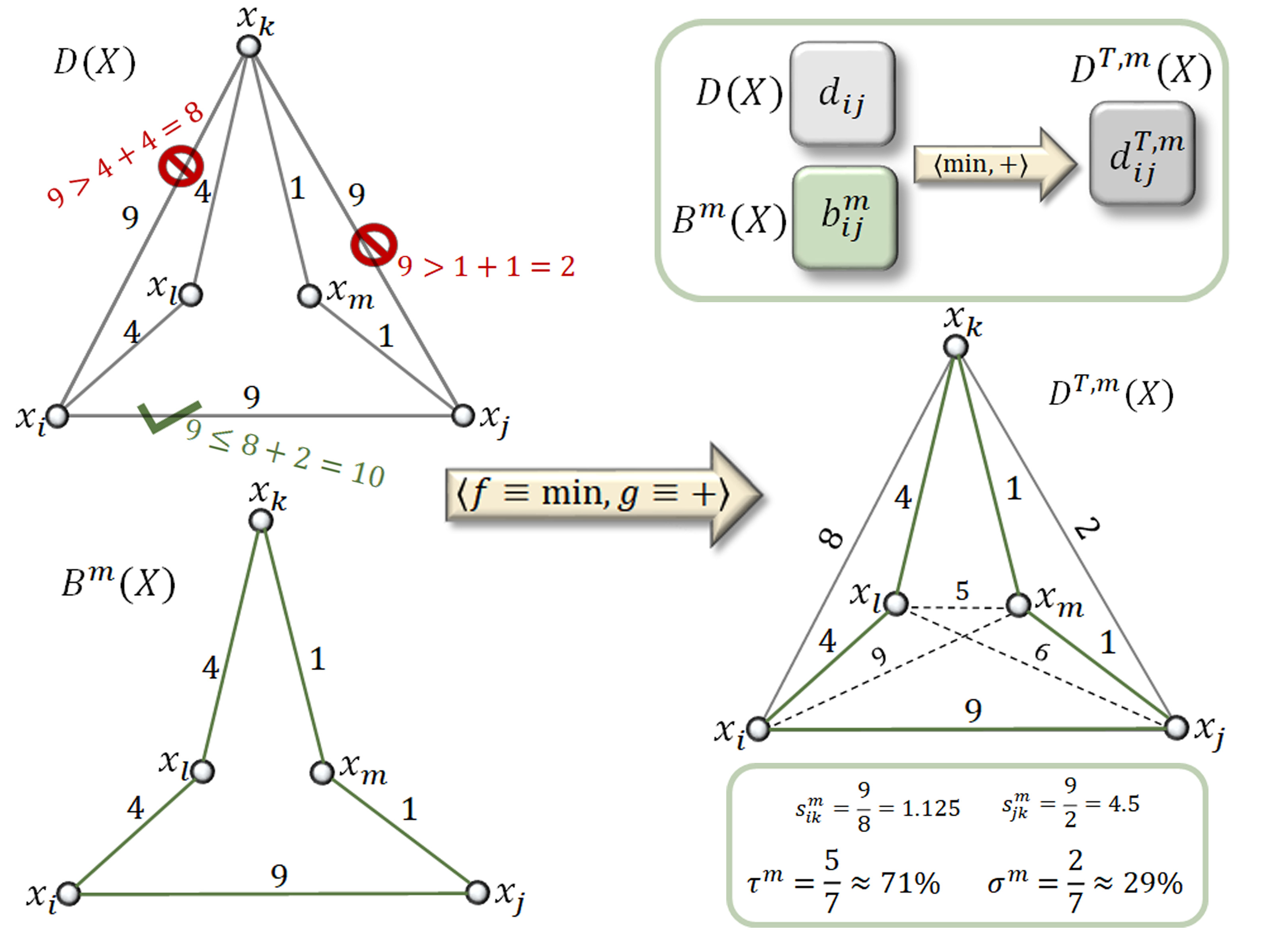}
    \caption{
        \textbf{Metric Backbone.}
        Top, right: Schematic showing metric closure $D^{T,m}(X)$ is obtained equivalently from either the original distance graph $D(X)$ or its metric backbone $B^m(X)$.
        Top, left: An example distance graph with 5 nodes and edge distance weights; edges $d_{ik}=9$ and $d_{jk}=9$ break the triangle inequality and are annotated with the computation of their indirect (shortest path) distances $d^{T,m}_{ik}=8$ and $d^{T,m}_{jk}=2$; edge $d_{ij}=9$ does not break the triangle inequality ($d^{T,m}_{ij}=10$).
        Bottom, left: The metric backbone graph $B^m(X)$.
        Bottom, right: The metric closure graph $D^{T,m}(X)$ with metric backbone edges in green, semi-metric edges in gray, and indirect edges (that do not exist in $D(X)$) in dashed gray. 
        Bottom, box: Measures of semi-metric edge distortion for the two edges that break the triangle inequality; backbone size and edge redundancy for graph $D(X)$.
    }
    \label{fig:metric_backbone}
\end{figure*}

\subsection{Metric backbone}
\label{section_metric_backbone}

%\subsection{A Standard Geometry}
%\label{section_def_metric}

The general distance backbone (eq. \ref{eq_backbone}) is based on the generalized triangle inequality (eq. \ref{eq_gen_triang}). However, most research on complex networks is based on the metric geometry given by the standard triangle inequality, obtained by setting $g \equiv +$ in (eq. \ref{eq_gen_triang}):
% $d_{ij} \geq d_{ik} + d_{kj}$.

\begin{equation}
    \label{eq_triang_metric}
    d_{ij} \leq d_{ik} + d_{kj}, \forall x_i,x_j,x_k \in X.
\end{equation}
In this case, the distance closure becomes the metric closure $D^{T,m}(X)$, used to compute standard shortest paths on distance graphs. Path length is computed by summing edge weights ($g \equiv +$ in eq. \ref{eq_length}), typically using Dijkstra's algorithm in the APSP 
%(\S Introduction).
(\S \ref{}).

Via the metric closure, we obtain a \emph{metric backbone} $B^{m}(X)$ by setting $g \equiv +$ in (eq. \ref{eq_backbone}).
This backbone contains all the edges of the original distance graph $D(X)$ that are metric, those that satisfy the standard triangle inequality (eq. \ref{eq_triang_metric});
edges not on the backbone are \emph{semi-metric} \cite{Rocha2002,Rocha2005,Simas2012}.
%, since they refer to distance functions that are reflexive and symmetric but do not satisfy the triangle inequality \cite{galvin_shore91}.
%
Figure \ref{fig:metric_backbone} shows the process of computing the metric backbone based on the general case of Figure \ref{fig:distance_backbone} and an example.
Notice that since the original example graph is connected (Top, left), its backbone is also a connected graph (Bottom, left) and the metric closure is a complete graph (Bottom, right), which exemplifies Corollary \ref{th_backbone_is_connected} for $g \equiv +$.

As shown above for the general distance backbone,
%(\S \ref{section_distance_backbone}),
semi-metric edges do not contribute to shortest-path computation and are thus redundant for that purpose (Theorem \ref{th_backbone_is_sufficient}).
This is seen in the example in Figure \ref{fig:metric_backbone} where edges $d_{ik}$ and $d_{jk}$ are semi-metric and do not contribute to any shortest-path computed for the metric closure $D^{T,m}(X)$.
Accordingly, by setting $g \equiv +$ in formulae \ref{eq_percentage_red} and \ref{eq_percentage_backbone} we obtain the \emph{semi-metric edge redundancy} and the \emph{relative size of the metric backbone} of graph $D(X)$, denoted by $\sigma^m (D)$ and $\tau^m (D)$---also known as the proportions of semi-metric and metric edges, respectively.
The distance graph in Figure \ref{fig:metric_backbone} is thus $\sigma^m (D) \approx 29\%$ redundant for shortest-path calculation with a metric backbone $B^m(X)$ that comprises $\tau^m (D) \approx 71\%$ of the original graph $D(X)$.
A measure of \emph{semi-metric edge distortion}, $s^m_{ij}$, is similarly obtained by setting $g \equiv +$ in eq. \ref{eq_edge_distortion}. If an edge $d_{ij}$ is metric, $s^m_{ij}=1$, meaning there is no distortion. If the edge is semi-metric, $s^m_{ij}>1$, and the larger the ratio, the more the edge breaks the triangle inequality (eq. \ref{eq_triang_metric}) and the more distorted it will be in the metric closure.
Figure \ref{fig:measures_metric} in SI depicts how $s^m_{ij}$ relates to the semi-metric redundancy and size of backbone measures.

\begin{figure*}[ht!]
    \centering
    \includegraphics[width=\textwidth]{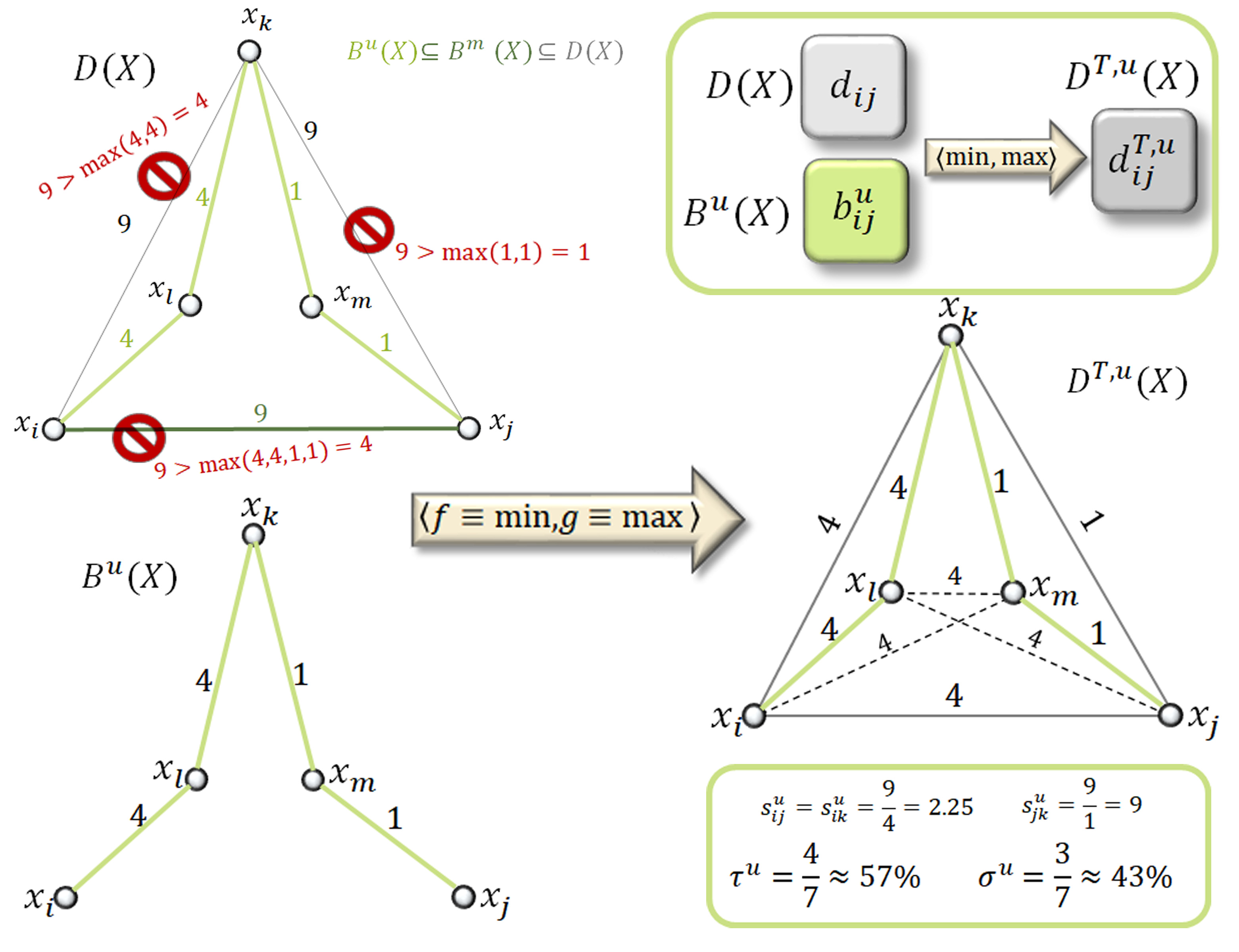}
    \caption{
        \textbf{Ultra-metric Backbone.}
        Top, right: Schematic showing ultra-metric closure $D^{T,u}(X)$ is obtained equivalently from either the original distance graph $D(X)$ or its ultra-metric backbone $B^u(X)$.
        Top, left: An example distance graph with 5 nodes and edge distance weights; edges $d_{ik}=d_{jk}=d_{ij}=9$ break the ultra-metric triangle inequality and are annotated with the computation of their indirect, shortest minimax path length ($g \equiv \max$): $d^{T,u}_{ik}=d^{T,m}_{ij}=4$ and $d^{T,m}_{jk}=1$; ultra-metric edges in light green, semi-ultra-metric edges in gray, and metric edge (in metric backbone but not in ultra-metric backbone) in dark green.
        Bottom, left: The ultra-metric backbone graph $B^u(X)$. Bottom, right: The ultra-metric closure graph $D^{T,u}(X)$ (right); indirect edges that do not exist in $D(X)$ appear in dashed gray. Bottom, box: Measures of semi-ultra-metric edge distortion for the three highlighted edges, as well as ultra-metric backbone size and edge redundancy of graph $D(X)$ shown at the bottom.
    }
    \label{fig:ultra_backbone}
\end{figure*}

\subsection{Ultra-metric backbone}
\label{section_ultrametric_backbone}

Given the family of shortest-path distance closures, other specific cases of distance backbones are meaningful.
For instance, with the ultra-metric closure $D^{T,u}(X)$, we compute path length as the weakest edge in the path by using $g \equiv \max$, also known as the \emph{minimax path} \cite{Camerini197810}
(\S \ref{section_isomorphism}).
%(\S Introduction).
%
This means that the closure now enforces the stronger \emph{ultra-metric triangle inequality}, $d_{ij} \leq \max (d_{ik},d_{kj})$. By setting $g \equiv \max$ in eqs. \ref{eq_backbone} and \ref{eq_gen_triang}, we obtain the \emph{ultra-metric backbone} $B^u(X)$, which contains all the edges of the original distance graph $D(X)$ that are ultra-metric.
Figure \ref{fig:ultra_backbone} shows the process of computing the ultra-metric backbone on the same example graph as Figure \ref{fig:metric_backbone}.

Edges that are not in the ultra-metric backbone may be referred to as \emph{semi-ultra-metric}, or simply semi-triangular as in the general case.
Via formulae \ref{eq_percentage_red} and \ref{eq_percentage_backbone} with $g \equiv \max$, we obtain the \emph{proportion of semi-ultra-metric edges} and the \emph{proportion of ultra-metric edges} in graph $D(X)$, denoted by $\sigma^u (D)$ and $\tau^u (D)$, respectively.
The distance graph in Figure \ref{fig:ultra_backbone} is thus $\sigma^u (D) \approx 43\%$ redundant, for the purpose of shortest minimax-path calculation, with an ultra-metric backbone $B^u(X)$ that comprises $\tau^u (D) \approx 57\%$ of the original graph $D(X)$.
A measure of \emph{semi-ultra-metric edge distortion}, $s^u_{ij}$, is similarly obtained via eq. \ref{eq_edge_distortion}.
%
%
%If an edge $d_{ij}$ is ultra-metric, $s^u_{ij}=1$, and there is no distortion. If the edge is semi-ultra-metric, $s^u_{ij}>1$; the larger the ratio, the more the edge breaks the ultra-metric triangle inequality  and the more distorted it will be in the ultra-metric closure.
%
Figure \ref{fig:measures_ultrametric} in SI depicts how $s^u_{ij}$ relates to ultra-metric redundancy and size of backbone measures.

Since the ultra-metric triangle inequality enforces a stronger transitivity criterion, edges not on the ultra-metric backbone may still be metric and thus included in the metric backbone  $B^m(X)$.
For instance, in Figure \ref{fig:ultra_backbone}, edge $d_{ij}$ is not on the ultra-metric backbone but is a metric edge.
The ultra-metric backbone is thus a subgraph of the metric backbone: $B^u (X) \subseteq B^m (X) \subseteq D(x)$.
Conversely, by Theorem \ref{th_backbone_is_sufficient}, the ultra-metric backbone is sufficient to compute the ultra-metric closure (all semi-ultra-metric edges are redundant for this purpose) but not to compute the metric closure.
%; to compute the metric closure the metric backbone
%(\S \ref{section_metric_backbone})
%is sufficient.
%
%
%\textcolor[rgb]{1.00,0.00,0.00}{The difference between the metric and ultra-metric backbone is a subgraph of $D(X)$ that contains all metric edges which are not ultra-metric: $D^{b,m-u}(X)$.}

\subsection{Other backbones}
\label{section_other_backbones}

The metric and ultra-metric backbones are based on well-known triangle inequalities, with the ultra-metric based on the strongest triangle inequality in the family of shortest-path distance closures \cite{simas_rocha_2014_MWS}.
Other criteria for triangle inequality can, however, be defined by setting $g$ to other binary operators in eq. \ref{eq_gen_triang} and in all formulae \ref{eq_length} to \ref{eq_edge_distortion}.
As discussed in the
\S \ref{section_isomorphism},
%\S Introduction,
$g$ can yield other well-known path length measures such as the Euclidean distance and the more general Minkowski metric, which lead to a \emph{Euclidean backbone} and the family of \emph{Minkowski backbones} associated with pathfinder networks \cite{schvaneveldt1990pathfinder}
(see \S \ref{section_related_concepts}).
%(\S Discussion).

We can go much beyond such familiar distances by exploring the space of t-norms $\land$ in the isomorphism of formulae \ref{eq_isomorphism} and \ref{eq_isomorphism_formulae}. These generalized logical conjunctions are well known and many families exist to be explored by the complex networks field \cite{Klir1995,simas_rocha_2014_MWS}.
This space, naturally, includes all the familiar distances. For instance,
%as shown in \S \ref{section_isomorphism},
the isomorphic counterpart $g_{+} \equiv d_{ik} + d_{kj}$ used to compute the metric backbone is the well-known Hamacher t-norm \cite{hamacher78}: $\land (p_{ik},p_{kj}) = p_{ik} \cdot p_{kj}/(p_{ik}+p_{kj}-p_{ik}\cdot p_{kj})$ for proximity weights $p_{ik},p_{kj} \in [0,1]$. It is a special case of the Dombi family of t-norms \cite{dombi1982general} that is isomorphic to the Minkowski metric family in our framework \cite{simas_rocha_2014_MWS}.

Interestingly, all of the above use distances that sum the contribution of each distance edge weight (or the powers of them) on a path, but we can consider others that do not\footnote{The ultra-metric $g_{\max} \equiv \max (d_{ik},d_{kj})$ is also an example of a distance that does not sum edge weights to compute the path length, even though it can be approximated by the Minkowski metric when $r \rightarrow +\infty$.}.
For example, under our isomorphism the product t-norm $\wedge_{\times} \equiv (p_{ik} \cdot p_{kj})$ leads to a td-norm also based on product: $g_{\times} \equiv (d_{ik}+1) \cdot (d_{kj}+1) -1$. Such an operation under shortest-path distance closure ($f \equiv \min$) yields a \emph{product backbone} that considers length to be proportional to the product of the edge weights on a path\footnote{Because $0$ is the identity element for td-norms, one cannot consider a more direct distance product such as $g \equiv (d_{ik} \cdot d_{kj})$ for path length. However, as defined, $g_{\times}$ preserves this desirable property for computing path length: $g_{\times} (d_{ik}, d_{kj}) = d_{ik}$ if $d_{kj} = 0$  and $g_{\times} (d_{ik}, d_{kj}) = d_{kj}$ if $d_{ik} = 0$. In other words, if an edge has $0$ distance, it does not affect the length of a path in which it appears.}.

Naturally, based on application, we can consider many other algebraic families of t-norms to define distance backbones.
For instance, a t-norm that computes the hyperbolic product of proximity weights yields isomorphic distance weights that become log-normalized after closure, a common technique in structural brain networks.
However, resulting distance weights (path length) are bound by the ultra-metric length function $g_{\max}$ and the \emph{drastic} td-norm:

\begin{equation*}
    \label{eq_drastic_TDNorm}
    g_{\textrm{drastic}} (d_{ik},d_{kj}) =
    \begin{cases}
        d_{ik} & \text{when $d_{kj}=0$}\\
        d_{kj} & \text{when $d_{ik}=0$}\\
        +\infty & \text{otherwise.}
    \end{cases}
\end{equation*}

\noindent That is, all possible td-norms $g$ in our distance backbone framework obey:

\begin{equation}
    \label{eq_TDNorm_order}
    g (d_{ik},d_{kj}) \in [g_{\max} (d_{ik},d_{kj}), g_{\textrm{drastic}} (d_{ik},d_{kj})], \forall d_{ik},d_{kj} \in [0, + \infty],
\end{equation}

\noindent as previously shown for all isomorphic t-norms \cite{Klir1995}. For instance, the following relationship is straightforward: $g_{\textrm{drastic}} \geq g_{\times} \geq g_{+} \geq g_{\max}$.

\subsection{Triangular geometry and the robustness of complex networks}
%\label{section_results}

Our formulation of distance backbones 
%(\S \ref{section_distance_backbone}), 
shows that running the APSP algorithm on a distance graph $D(X)$ (the metric closure, 
\S \ref{}),
%see \S Introduction),
as is commonly done in complex networks, enforces a topological distortion of the original graph whereby semi-metric edges are forced to satisfy the triangle inequality. The only edges invariant to this closure are on the metric backbone.
These observations of the common metric closure lead to two key conclusions that generalize to any distance backbone produced from the general shortest-path distance closure:

\begin{enumerate}
  \item \textbf{\emph{Semi-triangular distortion}}. Only the distance backbone edges of the original graph exist in the triangular space enforced by a given distance closure---the monoid $\langle d_{ij} \in [0, +\infty], g \rangle$ plus the generalized triangle inequality axiom of eq. \ref{eq_gen_triang}.
      %
      %An interesting consequence of this fact is that the metric backbone of a distance graph can be rendered more easily on a physical representation, preserving all edge distance weights proportionally. In contrast, semi-metric edges have to be more distorted to be rendered.
      %
      The semi-triangular edges break this geometry, and the \emph{semi-triangular distortion} $s^g_{ij}$ (eq. \ref{eq_edge_distortion}) measures, as a division factor, how much they must be distorted to fit the topology of the triangular space enforced by the distance closure.
      The distance backbone thus comprises triangular edges that function, metaphorically, as wormholes that minimize the (semi-triangular) distance edges via indirect shortest paths.
      This can be seen in the metric backbone example of Figure \ref{fig:metric_backbone}. Nodes $x_j$ and $x_k$ are $s^m_{jk}=4.5$ farther apart via their semi-metric direct edge than via an indirect path on the backbone that goes through $x_m$. In this sense, the shortest distance between $x_j$ and $x_k$ via $x_m$ exists on the metric backbone subgraph that obeys the triangle inequality; it exists ``on the metric geometry.'' In contrast, the direct edge $d_{jk}$ exists outside this metric geometry and is distorted (divided by $s^m_{jk}=4.5$) when the metric closure (APSP) is computed. 
  \item \textbf{\emph{Semi-triangular redundancy}}. Edges on a distance backbone are sufficient to compute all shortest paths according to a path length measure defined by the chosen distance closure (Theorem \ref{th_backbone_is_sufficient}).
      Thus, semi-triangular edges are redundant for shortest path computation and $\sigma^g (D)$ (eq. \ref{eq_percentage_red}) measures the amount of such \emph{semi-triangular redundancy} in $D(X)$. The value of this measure varies among large complex networks built from empirical data, but it is typically very large as shown below for metric and ultra-metric redundancy 
%(\S \ref{section_other_nets}) and 
(Table \ref{tableSM}).
\end{enumerate}

\begin{figure*}[ht!]
    \centering
    \includegraphics[width=\textwidth]{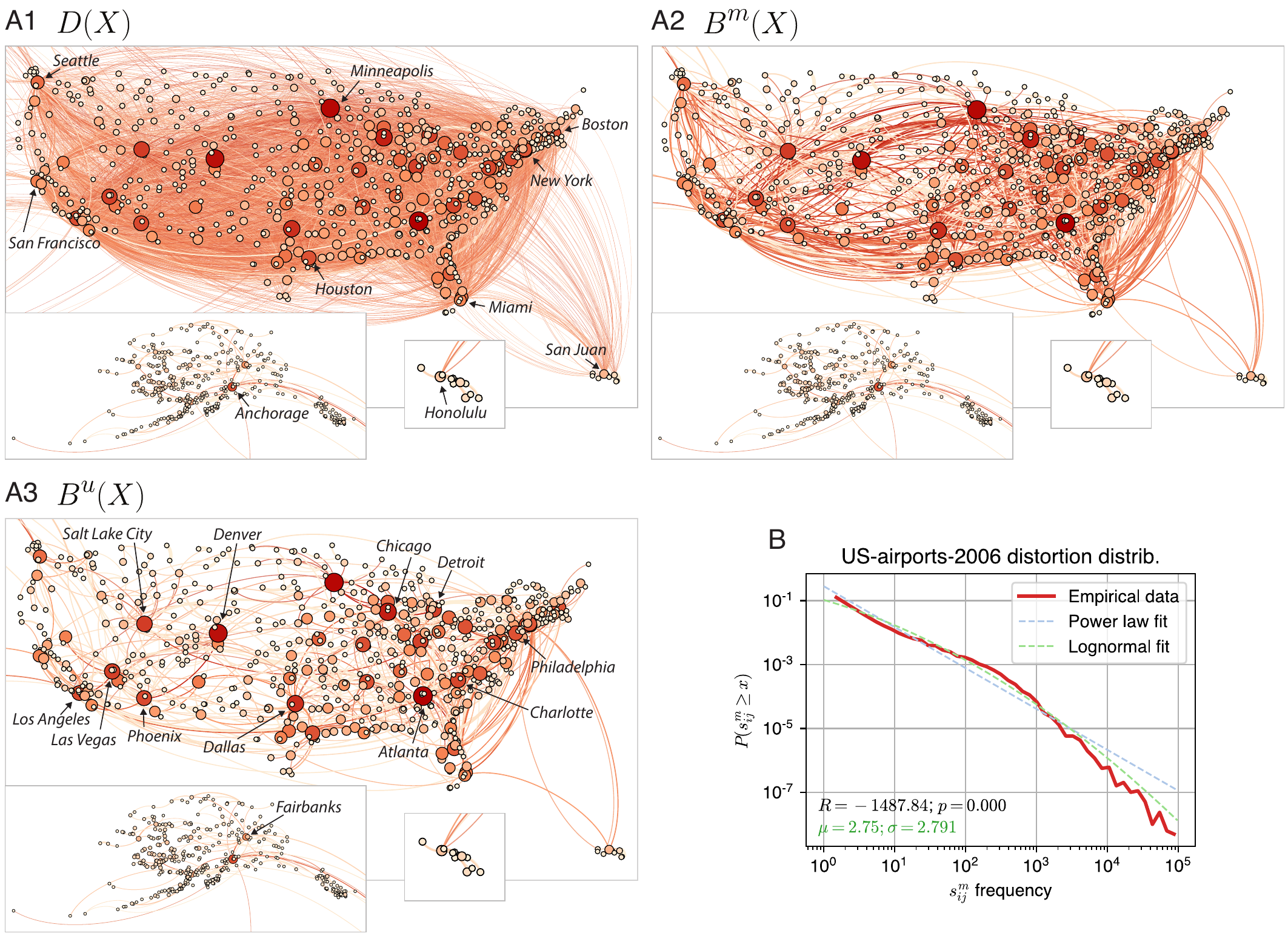}
    \caption{
        \textbf{The U.S. domestic nonstop airport transportation network and backbones for the year 2006.}
        \textbf{A1}. Distance network $D(X)$ with weights representing the average number of passengers between two airports. This is a reconstruction of the network used in Serrano et al. \cite{serrano} that keeps only the largest connected component and removes some U.S territory airports (e.g., Guam).
        \textbf{A2}. Metric backbone $D^m(X)$ with only $\tau^m = 16\%$ of the original edges.
        \textbf{A3}. Ultrametric backbone $D^u(X)$ with only $\tau^u = 9\%$ of the original edges.
        \textbf{B}. Log-binned distribution of semi-metric distortion $s^m_{ij}$ values for the $\sigma^m = 84\%$ of semi-metric edges in the network.
        Both a log-normal ($ \langle s_{ij}^m \rangle = 4.25$, SD $=1.916$) and a power law fit are shown; a comparison between the two favors the former as a better representation of the data. The best power law fit has an exponent $\gamma = 1.353$.
        Data fitted using the `powerlaw' python package \cite{Clauset2009powerlaw,Alstott2014powerlaw}.
    }
    \label{fig:US-Airports-2006}
\end{figure*}

%\begin{figure*}[ht!]
    %\centering
    %\includegraphics[height=2.5in]{new-images/us-airports-2006-s-values}
    %\caption{
        %\textbf{Semi-metric distortion distribution for the US domestic nonstop airport transportation network.}
        %
        %Log-binned distribution of semi-metric distortion $s^m_{ij}$ values for the $\sigma^m = 84\%$ of semi-metric edges in network.
        %
        %Both a log-normal ($ \langle s_{ij}^m \rangle = 4.25$; St.Dev $=1.916$) and a Powerlaw fit are shown, a comparison between the two favors the former as a better representation of the data; the best Powerlaw fit has an exponent $\gamma = 1.353$.
        %
        %Data fitted using the `powerlaw' python package \cite{Clauset2009powerlaw,Alstott2014powerlaw}.
    %}
    %\label{fig:US-Airports-2006-s-values}
%\end{figure*}

%\subsection{Metric behavior and the robustness of complex networks}
%\label{section_metric_complex_nets}

The concept of metric and semi-metric edges, as well as their proportion in a distance graph, relates directly to several key concepts in the study of complex networks.
\emph{Edge betweeness centrality} is defined as the number of shortest paths that pass through an edge in a graph \cite{Girvan2002}.
Since the only edges that contribute to shortest paths are on the metric backbone, it is clear that betweeness centrality is positive for metric edges and null for semi-metric edges.
The \emph{distortion} $s^m_{ij}$ of semi-metric edges (obtained by setting $g \equiv +$ in eq. \ref{eq_edge_distortion}), however, varies widely. In other words, among the edges that have null betweeness centrality, some are much more semi-metric than others depending on how strongly they break the triangle inequality.
For instance, Figure \ref{fig:US-Airports-2006}B shows that the $s^m_{ij}$ distribution for semi-metric edges is very heterogeneous with a wide variation of values in a network of air traffic between more than a thousand U.S. Airports.
%\footnote{The distribution fits a power law of exponent $\gamma=2.2$ quite well. Other distributions could be fit to the data, but the key point for our purposes is only to show that it follows a heterogenous distribution with wide variation.}.
%
%
The semi-metric distortion parameter $s^m_{ij}$ thus offers a finer characterization of edges not on the metric backbone than does betweeness centrality---those edges that do not contribute to shortest paths but can contribute to other phenomena on networks, including modularity and diffusion. This is meaningful as it impacts shortest path robustness, as discussed below.
%

%Figure \ref{fig:measures_metric} in SI adapts Figure \ref{fig:measures_triangular} to the metric backbone case to depict how $s^m_{ij}$ relates to metric redundancy and size of backbone measures.

The \emph{distribution of shortest path length} is also important for complex networks. We expect, for instance, a small mean value of such a distribution, $\langle d^{T,m}_{ij} \rangle$, in both the Erdos-Renyi model of random graphs and in small-world graphs \cite{Dorogovtsev2003}.
Interestingly, semi-metric edges do not affect this distribution at all since they do not contribute to the computation of shortest path length. Indeed, only the edges in the metric backbone contribute to this distribution.
Therefore, removing a semi-metric edge from a distance graph does not change its distribution of shortest path length, but removing a single metric edge (one on the metric backbone) may increase the average shortest path length, since every edge in the metric backbone participates in at least one of the shortest paths.
We can thus say that the shortest path length distribution of distance graphs is \emph{robust} to semi-metric edge removal, but is affected by metric edge removal. Therefore, the smaller the metric backbone (small $\tau^m$, large $\sigma^m$), the more robust the distribution of shortest path length of $D(X)$ is to random edge removal.

\begin{figure}[ht!]
    \centering
    \includegraphics[width=0.88\textwidth]{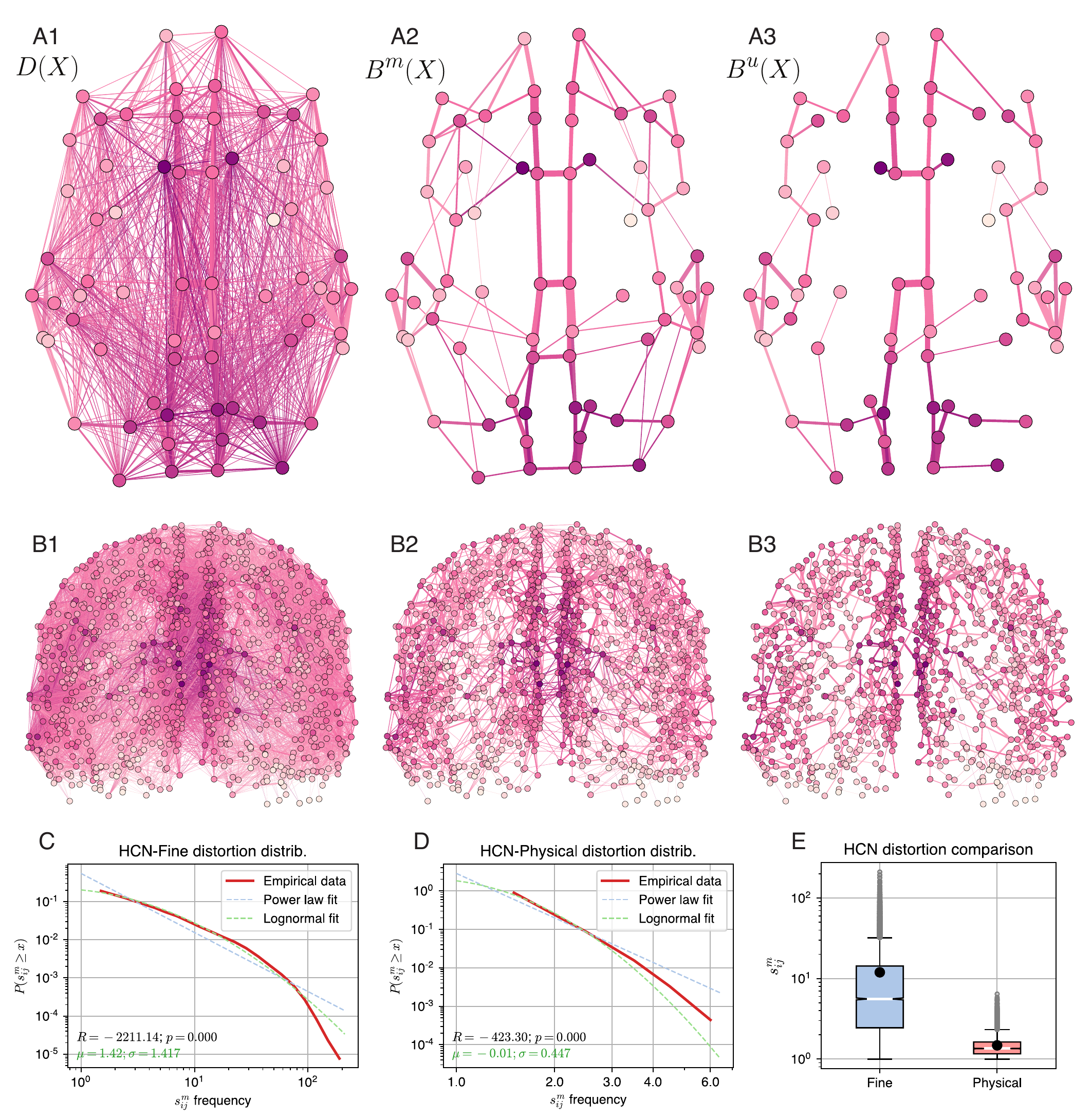}
    \caption{
        \textbf{Human Connectome Network and Backbones.}
        \textbf{A1-3}: \textit{HCN-Coarse}. \textbf{B1-3}: \textit{HCN-Fine}.
        \textbf{A1,B1}: Original distance Networks \cite{hagmann2008mapping}, whose distance weights are inversely proportional to the volume of cortico-cortical axonal pathways between brain regions (nodes), obtained via diffusion spectrum imaging 
        %(\S Materials, Methods, \& Data).
        (\S \ref{section_methods}).
        \textbf{A2,B2}. Metric backbone $D^m(X)$ with only $\tau^m = 9.23\%$ and $\tau^m = 17.57\%$ of original edges for  \textit{HCN-Coarse} and \textit{HCN-Fine}, respectively. 
        \textbf{A3,B3}. Ultrametric backbone $D^u(X)$ with only $\tau^u = 5.66\%$ and $\tau^u = 5.53\%$ of original edges for \textit{HCN-Coarse} and \textit{HCN-Fine}, respectively.
        \textbf{C}.
        Log-binned distribution of semi-metric distortion $s^m_{ij}$ values for the $\sigma^m = 82\%$ of semi-metric edges in \textit{HCN-Fine} network.
        Both a log-normal ($ \langle s_{ij}^m \rangle = 1.42$; SD $=1.417$) and a power law fit are shown; a comparison between the two favors the former as a better representation of the data. The best power law fit has an exponent $\gamma \approx 1.6$.
        \textbf{D}.
        Log-binned distribution of semi-metric distortion $s^m_{ij}$ values for the $\sigma^m = 51\%$ of semi-metric edges in \textit{HCN-Physical} network.
        Both a log-normal ($ \langle s_{ij}^m \rangle = -0.01$; SD $=0.447$) and a Powerlaw fit are shown; a comparison between the two favors the former as a better representation of the data.
        \textbf{E}. Box plot of semi-metric distortion $s^m_{ij}$ values comparing networks \textit{HCN-Fine} and \textit{HCN-Physical}.
        Data fitted using the `powerlaw' python package \cite{Clauset2009powerlaw,Alstott2014powerlaw}.
    }
    \label{fig:hagmann-66-998}
\end{figure}

For example, in the Airport traffic network depicted in Figure \ref{fig:US-Airports-2006}, it is desirable to have a robust shortest path length distribution so that removal of an edge does not significantly increase the distance between cities. Indeed, the metric backbone makes up only $\tau^m (D) \approx 16\%$ of the network (Table \ref{tableSM}). Random edge removal (meaning the interruption of all air traffic between two cities) will thus affect shortest paths on the network less than a sixth of the time, which denotes robustness to this type of disturbance.
Natural networks can be even more robust to this type of removal. For instance, the metric backbone of a Human Connectome Network (HCN) \cite{hagmann2008mapping} shown in Figure \ref{fig:hagmann-66-998} makes up only $\tau^m (D) \approx 9\%$ or $18\%$ of the network, depending on the size of the brain parcellation used, with corresponding redundancy of $\sigma^m (D) \approx 91\%$ or $82\%$, respectively. This means that the shortest paths on the network are very robust to random edge removal as they are affected only between a fifth and a tenth of the time.
Analysis of other networks below 
%(\S \ref{section_other_nets}) 
reveals a similar phenomenon in networks across biological, technological, and social realms.

While complex networks such as the HCN can display high robustness to random edge removal by being organized around very small metric backbones, the random removal of edges from the metric backbone itself can have varying impacts on the distribution of shortest path length.
The impact depends on the topology of the backbone itself as well as the shape of the distribution of semi-metric distortion values $s^m_{ij}$, which are thus additional robustness mechanisms available to complex systems.
Removal of edges from the backbone tends to increase some shortest paths and thus the average shortest path length, but the amount of increase depends on the available alternative paths.
If the backbone itself contains alternative paths of similar length, the impact of removal is minimal; this situation is more likely if the backbone preserves strong transitivity (or community structure) with a small number of bridges. Similarly, if there are many edges with very small semi-metric distortion $s^m_{ij} \approx 1$ that are not on the backbone, then an edge randomly removed from the backbone is likely to be replaced by one with very small distortion and thus small impact on the shortest path distribution, as is the case for the \textit{HCN-Physical} network as discussed below (see also Figure \ref{fig:hagmann-66-998}). Conversely, if there are many semi-metric edges with large distortion $s^m_{ij} \gg 1$, random removal of edges from the backbone is likely to have a big impact on the shortest path distribution and its average value.
This further highlights the importance of the finer characterization of semi-metric edges afforded by the distortion measure and its distribution, but not by betweeness centrality.

Consider the Figure \ref{fig:metric_backbone} example; removing edge $b^m_{ij}$ of the backbone of network $D(x)$ results in minimal change to shortest paths, only affecting the shortest path between $x_i$ and $x_j$: $d^{T,m}_{ij}: 9 \rightarrow 10$. The metric backbone does not require the addition of another edge and the average shortest path length changes very little $\langle d^{T,m} \rangle = 4.9 \rightarrow 5$ due to the short indirect distance between nodes $x_i$ and $x_j$, via nodes $x_l,x_k$ and $x_m$. This strong transitivity allows the backbone to lose edge $b^m_{ij}$ with minimal impact on the distance closure  $D^{T,m}$.
%
%}.
%
In contrast, removing edge $b^m_{jm}$ results in a big change to shortest paths: $d^{T,m}_{jm}: 1 \rightarrow 10$, $d^{T,m}_{jk}: 2 \rightarrow 9$, $d^{T,m}_{jl}: 6 \rightarrow 13$. The metric backbone requires the addition of the previously semi-metric edge $d_{jk} \rightarrow b^m_{jk}=9$ and the average shortest path length changes considerably as well $\langle d^{T,m} \rangle: 4.9 \rightarrow 7.2$.
Another case is removing edge $b^m_{il}$, which also requires adding previously semi-metric edge $d_{ik} \rightarrow b^m_{ik}=9$ to the backbone, but has less impact on the shortest path distribution: $d^{T,m}_{il}: 4 \rightarrow 13$, $d^{T,m}_{ik}: 8 \rightarrow 9$, $d^{T,m}_{im}: 9 \rightarrow 10$. The average shortest path length changes less than in previous case  $\langle d^{T,m} \rangle: 4.9 \rightarrow 6$.

The difference between the last two cases depends on the semi-metric distortion of the previously semi-metric edges that need to be added to the backbone after edge removal. While the first is a case of adding an edge with high semi-metric distortion $s^m_{jk}=4.5$ to the backbone, the second is the opposite $s^m_{ik} = 1.125$  (see Figure \ref{fig:metric_backbone}, bottom, right).
Notably, edges $d_{jk}$ and $d_{ik}$ both have null betweeness centrality in the original graph despite such distinct impacts on shortest paths after removal of metric edges from the backbone. This illustrates how, compared to betweeness centrality, the semi-metric distortion measure and its distribution may be more characteristic of impact on and robustness of the shortest path length distribution.

%LMR: Stuff from ultrametric below. May need different section or intro

The ultra-metric backbone is also very useful for characterizing robustness of shortest-paths to attacks. It derives from the strongest td-norm $g$  (eq. \ref{eq_TDNorm_order}), meaning that the shortest possible path length in a distance graph $D(X)$ is given by the ultra-metric closure as captured by eq. \ref{eq_TDNorm_order}.
Therefore, edges in the ultra-metric backbone include the strongest associations between node-variable pairs in the original data used to produce $D(X)$ and constitute a subgraph of the metric backbone, 
%(\S \ref{section_ultrametric_backbone} and Figure \ref{fig:ultra_backbone}).
as exemplified in Figure \ref{fig:ultra_backbone}.
Because ultra-metric edges are defined by distance weights that are much smaller than any indirect path connecting their respective nodes (with length computed by any $g$), they are likely to be included in many standard shortest paths and thus observe high edge betweeness.
In other words, the ultra-metric backbone contains the strongest pairwise associations that are most likely links in many paths associating multiple variables. Therefore, attacks on ultra-metric edges are likely to strongly impact the distribution of shortest path lengths and its average value.
%
%LMR: we could test this...
%
The relative size of the ultra-metric to the metric backbone ($\tau^u/\tau^m$) is thus an indication of how robust to attack the metric backbone and the distribution of shortest-path length in a graph are (see values for various networks in Table \ref{tableSM}).

In addition to betweeness centrality and robustness of shortest path length distribution, the metric backbone affects all measures associated with shortest path length, such as efficiency, reachability, and modularity.
For instance, removal of metric edges from the original distance graph is likely to break the graph into separate components since all bridges are included in the metric backbone (Corollary \ref{th_backbone_bridges}). The likelihood is particularly high when using thresholding or other reduction techniques that do not consider the metric backbone 
(\S \ref{section_discussion}).
%(\S Discussion).

Semi-metric edges, on the other hand, do not affect shortest paths and do not form bridges yet they fill-up the graph and impact measures that are related to local connectivity such as clustering coefficient, node degree, and transitivity ratio (different from transitive closure).
Modularity in particular depends on both types of edges since metric edges include all bridges and semi-metric edges affect local connectivity.
Similarly, the small-world phenomenon on weighted graphs depends on both types of edges as average shortest path length depends only on the metric backbone and the clustering coefficient depends on both types of edges.
Table \ref{tableClustering} in SI shows how the clustering coefficient tends to decrease when semi-metric and semi-ultra-metric edges are removed from various networks. 
A study of the preservation of community structure in the metric backbone is forthcoming \cite{epibackbone_working}.

\subsection{Backbones of networks across domains}
\label{section_other_nets}

\begin{table*}
    \footnotesize
    \centering
    %\begin{tabular}{cl|rrr|rr|rr|r|l}
    \begin{tabular}{cl|rrr|r|r|r|l}
        \toprule
        & $D(X)$ & $|X|$ & $|d_{i>j}|$ & $\delta $ & $\tau^m$ 
        %& $\sigma^m$ 
        & $\tau^u$ 
        %& $\sigma^u$ 
        & $\tau^u/\tau^m$ & ref. \\
        \midrule
        %
        % Infrastructure Networks
        \parbox[t]{2mm}{\multirow{2}{*}{\rotatebox[origin=c]{90}{Techn.}}}
        & U.S.-airports-2006 & 1,075 & 11,973 & 2.07 & 16.14 
        %& 83.86 
        & 8.98 
        %& 91.02 
        & 55.64 & \cite{serrano} \\
        & U.S.-airports-500 & 500 & 2,980 & 2.39 & 37.15 
        %& 62.85 
        & 16.75 
        %& 83.26 
        & 45.08 & \cite{colizza2007reaction} \\
        \midrule
        %
        % Biological Networks
        \parbox[t]{2mm}{\multirow{6}{*}{\rotatebox[origin=c]{90}{Biological}}}~
        & Enterocyte GRN & 8,058 & 1,689,653 & 5.21 & 1.75 
        %& 98.25 
        & 0.83 
        %& 99.17 
        & 47.51 & \cite{enterocyte_working} \\
        & HCN-fMRI & 998 & 497,503 & 100 & 5.5 
        %& 94.5 
        & 0.2 
        %& 99.8 
        & 3.64 & \cite{hagmann2008mapping} \\
        & HCN-Coarse & 66 & 1,148 & 53.52 & 9.23 
        %& 90.77 
        & 5.66 
        %& 94.34 
        & 61.32 & \cite{hagmann2008mapping} \\
        & HCN-Fine & 989 & 17,865 & 3.66 & 17.57 
        %& 82.44 
        & 5.53 
        %& 94.47 
        & 31.49 & \cite{hagmann2008mapping} \\
        & C-elegans & 297 & 2,148 & 4.89 & 46.97 
        %& 53.03 
        & 13.97 
        %& 86.03 
        & 29.73 & \cite{WATTS1998} \\
        & HCN-Physical & 989 & 17,865 & 3.66 & 49.25 
        %& 50.75 
        & 5.53 
        %& 94.47 
        & 11.23 & \cite{hagmann2008mapping} \\
        \midrule
        % Social Networks
        \parbox[t]{2mm}{\multirow{6}{*}{\rotatebox[origin=c]{90}{Social}}} 
        & High-school & 788 & 118,291 & 38.15 & 7.84 
        %& 92.16 
        & 0.66 
        %& 99.34 
        & 8.49 & \cite{Salathe2010highresolution} \\
        & Primary-school & 242 & 8,317 & 28.52 & 9.5 
        %& 90.50 
        & 2.9 
        %& 97.10 
        & 30.51 & \cite{sociopatterns2011primaryschool} \\
        & Freeman & 32 & 266 & 53.63 & 31.96 
        %& 68.05 
        & 11.65 
        %& 88.35 
        & 36.47 & \cite{Freeman:1979} \\
        & Cond-mat-2003 & 27,519 & 116,181 & 0.03 & 77.27
        %&
        & 62.77
        %&
        & 81.23 & \cite{Newman2001} \\
        & Cond-mat & 13,861 & 44,619 & 0.05 & 81.13 
        %& 18.87 
        & 70.62 
        %& 29.38 
        & 87.05 & \cite{Newman2001} \\
        & Net-science & 379 & 914 & 1.28 & 83.59 
        %& 16.41 
        & 78.45 
        %& 21.55 
        & 93.85 & \cite{newman2006community} \\
        \midrule
        % Knowledge
        \parbox[t]{2mm}{\multirow{4}{*}{\rotatebox[origin=c]{90}{Knowledge}}}
        & Wikipedia-Fact & 3.4M & 23M & $\approx 0$ & 2 
        %& 98 
        & - 
        %& - 
        & - & \cite{ciampaglia2015computational} \\
        & MyLib-keywords & 500 & 115,737 & 92.78 & 4.36 
        %& 95.64 
        & 0.43 
        %& 99.57 
        & 9.9 & \cite{Rocha2005} \\
        & Instagram depression & 3,288 & 230,799 & 4.27 & 8.1 
        %& 91.90 
        & 1.47 
        %& 98.53 
        & 18.12 & \cite{correia2016monitoring,correia2019thesis} \\
        & MyLib-journals & 1,690 & 51,234 & 3.59 & 22.4 
        %& 77.60 
        & 7.59 
        %& 92.41 
        & 33.89 & \cite{Rocha2005} \\
        & MyLib-users & 381 & 6,525 & 9.01 & 27.49 
        %& 72.51 
        & 7.79 
        %& 92.22 
        & 28.32 & \cite{Rocha2005} \\
        \bottomrule
    \end{tabular}
    \caption{
        \textbf{Metric and ultra-metric backbones of networks across domains}.
        $|X|$: number of nodes;
        $|d_{i>j}|$: number of finite distance edges;
        $\delta $: density of distance graph $D(X)$;
        $\tau^m$, $\tau^u$: relative size of metric and ultra-metric backbones. 
        Values of $\delta, \tau^m, \tau^u, $
        %\sigma^m, \sigma^u$, 
        and $\tau^u/\tau^m$ are shown as percentages ($\%$); the proportion of semi-metric ($\sigma^m$) and semi-ultrametric ($\sigma^u$) edges are obtained directly from size of backbones since $\sigma = 1-\tau$.
        Rows are ranked by increasing size of metric backbone ($\tau^m$) within each domain group.
        To facilitate comparison, analysis is restricted to the largest connected component and the reported number of nodes $|X|$ refers to that subgraph. See
        \S \ref{section_methods} 
        %\S Materials, Methods \& Data
        for additional details and a description of the networks.
    }
  \label{tableSM}
\end{table*}

Using our open-source Python package, we have computed the metric and ultra-metric backbones of various networks across domains ranging from biology to society and technology (\S \ref{section_methods}).
%(\S Materials, methods \& data). 
The results are summarized in Table \ref{tableSM}.
Figures \ref{fig:US-Airports-2006} and \ref{fig:hagmann-66-998} further depict the metric $B^m(X)$ and ultra-metric $B^u(X)$ backbones of the \textit{U.S.-airports-2006} traffic network \cite{serrano} and two distinct parcellations of a \textit{human connectome network} (HCN) built from a cohort of five participants \cite{hagmann2008mapping}.

%%%% General results of metric backbone

It is striking that more than half of the networks across all domains have metric backbones that contain $\tau^m \leq 20\%$ of the edges in the original network and almost three quarters have $\tau^m \leq 40\%$.
This denotes a great deal of redundancy in the computation of shortest-paths with $\sigma^m \geq 60\%$ for most networks studied.
The two smallest metric backbones found are for a large gene-regulatory network (GRN) of more than 8000 genes interacting in insect intestinal cells \cite{enterocyte_working}, and the very large knowledge network of 3.4 million Wikipedia concepts built for automatic fact-checking \cite{ciampaglia2015computational}, with $\tau^m = 1.75\%$ and  $\tau^m \approx 2\%$, respectively.
It is impressive that such small backbones are sufficient to compute all shortest paths.
For the \textit{Enterocyte GRN}, this suggests that $\sigma^m = 98.25\%$ of all gene interactions have stronger alternative pathways, which adds much robustness to gene regulation.
For the knowledge network, this means that automatic fact-checking inferences pursued via shortest-paths on Wikipedia (such as in \cite{ciampaglia2015computational}) can ignore $98\%$ of concept associations in the knowledge graph.

%% The Physical, Euclidean vs real multivariate networks

All metric backbones, except for scientific collaboration networks (discussed below), are composed of fewer than $\tau^m = 50\%$ of the original edges.
An interesting case of a fairly large backbone is the \textit{HCN-Physical} network. The graph is built from the finer parcellation of the human brain used to study the human connectome \cite{hagmann2008mapping} but the edges denote the physical length of each connection rather than the volume of cortico-cortical axonal pathways captured by the \textit{HCN-Fine} network 
%(\S Materials, Methods \& Data.)
(\S \ref{section_methods}.)
Since the edges of \textit{HCN-Physical} are constrained by a physical, 3D geometry, we expect a Euclidean distance graph except for minor variations due to connections being curved in the interior of the brain. Indeed this is what our analysis confirms. The metric backbone is about half of the original network ($\tau^m=49.25\%$) but the semi-metric edges are almost metric with very little semi-metric distortion as shown in Figures \ref{fig:hagmann-66-998}D and \ref{fig:hagmann-66-998}E. The largest value of $s_{ij}^m = 6.4$, with mean and median values of 1.5 and 1.4, respectively.

In contrast, the \textit{HCN-Fine} network has a small metric backbone ($\tau^m=17.57\%$) and its distribution of semi-metric distortion displays a wide variation as shown in Figures \ref{fig:hagmann-66-998}C and \ref{fig:hagmann-66-998}E. The largest value of $s_{ij}^m \approx 200$, with mean and median values of 12 and 5.6, respectively.
In other words, while the \textit{HCN-Fine} network behaves like most natural networks we have observed---a small metric backbone with the edges outside of the backbone characterized by a long-tailed distribution of semi-metric distortion---the \textit{HCN-Physical} network is essentially entirely metric with almost negligible semi-metric distortion to about $50\%$ of the edges, as the box plot comparison in Figure \ref{fig:hagmann-66-998}E emphasizes \footnote{The \textit{HCN-Physical} network can be seen as a real-world null model of a distance graph on Euclidean geometry with a minor uniformly random semi-metric distortion to edges (with a $50\%$ chance of breaking the triangle inequality), which also explains its very small ultra-metric backbone ($\tau^u =5.53\%$).}.

%%%% Robustness and ultra-metric backbone

The networks in Table \ref{tableSM} also reveal that the ultra-metric backbone is most often a small subgraph of the original distance graph: For $68\%$ of the networks, $\tau^u \leq 10\%$.
Therefore, the edges that most affect shortest paths are difficult to damage by random attack. 
Of the networks we studied, the shortest paths of the \textit{HCN-fMRI} network are the most robust to random attack on two counts. 
First, because the metric backbone is very small ($\tau^m = 5.5\%$), almost $\sigma^m \approx 95\%$ of random attacks would be on semi-metric edges that have no effect on shortest paths.
Second, the ultra-metric backbone is composed of only $\tau^u = 0.2\%$ of the edges in the original graph, which represents only $\tau^u/\tau^m = 3.64\%$ of the small metric backbone. Thus, the edges with most impact on shortest-paths would rarely be damaged under random attack.
Even though the meaning of shortest paths in networks built from fMRI correlation data is not obvious, 
this suggests that the functional activity of human brains is very robust to disturbances from the point of view of shortest-path communication between regions.

Very similar behavior is observed for the social network of \textit{High School} students \cite{Salathe2010highresolution}, the knowledge networks of keywords obtained from the user profiles of a digital library (\textit{MyLib-Keywords}), and a large set of Instagram user posts related to depression (\textit{Instagram-depression}).
The robustness of shortest-paths in the social network of High School students means, for instance, that in the presence of an infectious epidemic, the transmission speed and spread would be virtually unaffected by removal of a very large number the social connections\cite{epibackbone_working}. Similarly, in knowledge networks removal of connections between concepts or keywords hardly disturbs associations (inferences), such as recommendations \cite{Rocha2005} or adverse drug reactions \cite{correia2016monitoring}, made through indirect paths on the backbones of those networks.

Sometimes, however, the ultra-metric backbone is composed of a large proportion of the metric backbone: $\tau^u/\tau^m \gtrapprox 50\%$. 
This can happen in situations where the metric backbone is very small ($\tau^m \lessapprox 15\% $), such as the \textit{Enterocyte GRN}, \textit{HCN-Coarse}, or \textit{U.S.-airports-2006} networks. It can also happen in situations where the metric backbone is large ($\tau^m > 75\% $), such as the scientific co-authorship networks \textit{cond-mat-2003}, \textit{cond-mat}, and \textit{net-science}, shown in Figure \ref{fig:netscience}.
In the first case, because the metric backbone is small, random attacks are likely to hit semi-metric edges and thus have no effect on shortest-paths overall. However, directed attacks to the metric backbone are very likely to hit the ultra-metric backbone and thus result in a large impact to the shortest path distribution.
This is seen in the examples in Figures \ref{fig:metric_backbone} and \ref{fig:ultra_backbone}. As discussed above, the removal of any ultra-metric edge, $\{b^u_{jm}, b^u_{mk}, b^u_{il}, b^u_{lk} \}$ has a much stronger impact on shortest paths than the removal of the metric edge $b^m_{ij}$.
The \textit{U.S.-airports-2006} network shown in Figure \ref{fig:US-Airports-2006} provides a real-world example of this behavior. While its metric backbone is small, it is composed of $\tau^u/\tau^m \approx 56\%$ ultra-metric edges. This suggests that shortest-paths in the overall network would be very affected by interruptions to any one of more than half of the connections in the metric backbone. In other words, air traffic is fragile to targeted attacks on backbone connections.

A similar case is seen in 
%the coarser brain parcellation of 
the \textit{HCN-Coarse} network. Its ultra-metric backbone is $\tau^u/\tau^m \approx 61\%$ of a small metric backbone. This suggests that while there is much redundancy and robustness in calculation of shortest-paths, the few edges that contribute to shortest-paths are most important. 
As seen in Figure \ref{fig:hagmann-66-998}.A, most of these edges involve regions that bridge the hemispheres (near the cortical midline) and are located close to the insula.
These observations are coherent with the fact that these edges and regions channel much of the communication between hemispheres and into and out of the insula, a central cortical hub.  
Similar behavior is observed in the finer brain parcellation seen in Figure \ref{fig:hagmann-66-998}.B, even though in this case the metric backbone is a little more robust to attack, as $\tau^u/\tau^m \approx 32\%$.

\begin{figure}[ht!]
    \centering
    \includegraphics[width=\textwidth]{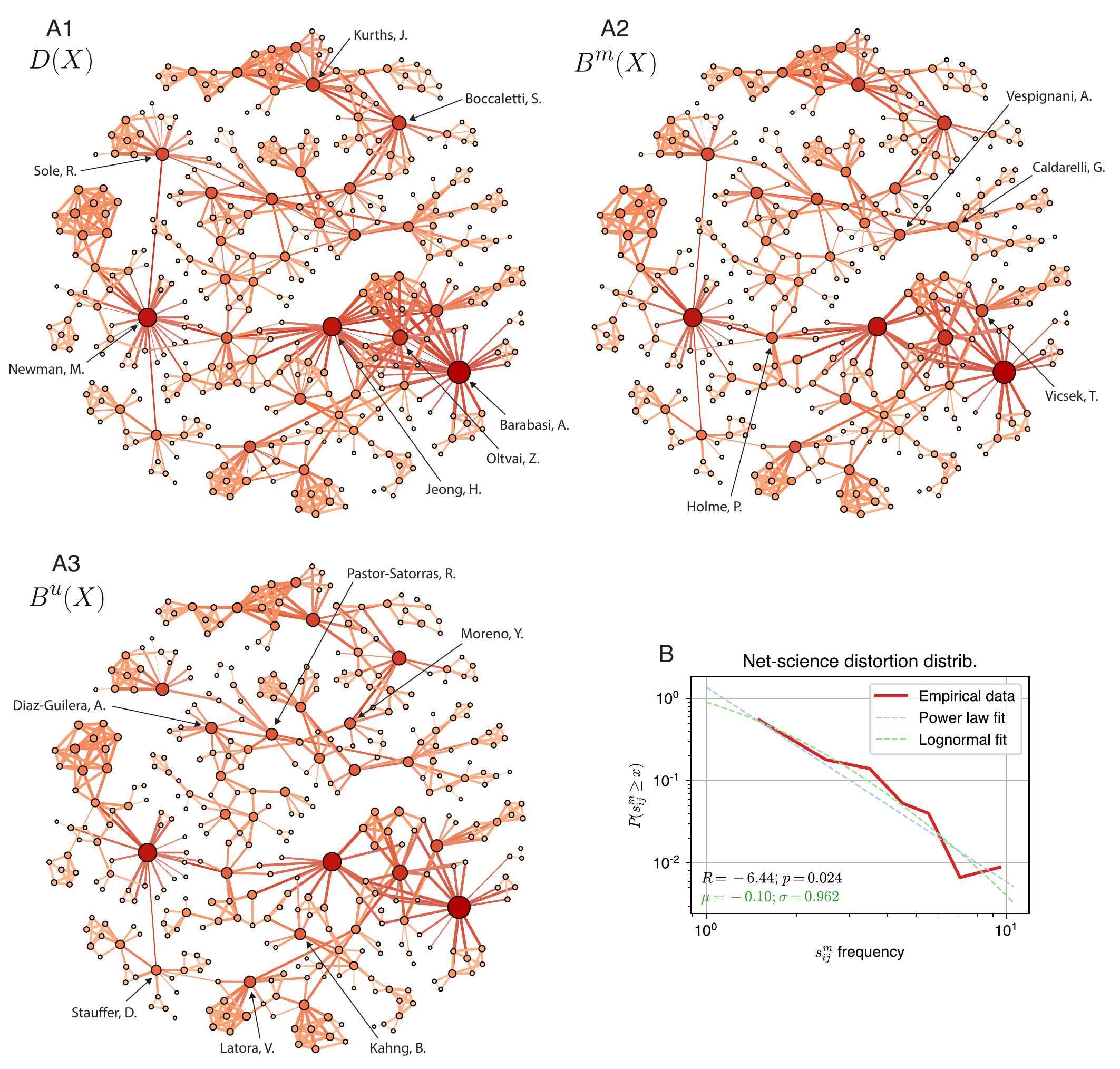}
    \caption{
        \textbf{Net-science network and backbones.}
        \textbf{A1}. Original Net-science distance network.
        \textbf{A2}. Metric backbone $D^m(X)$ with $\tau^m = 83.59\%$ of original edges.
        \textbf{A3}. Ultrametric backbone $D^u(X)$ with $\tau^u = 78.45\%$ of original edges.
        \textbf{C}. Log-binned distribution of semi-metric distortion $s^m_{ij}$ values for the $\sigma^m = 16.41\%$ of semi-metric edges in network.
        A log-normal ($ \langle s_{ij}^m \rangle = -0.10$; SD $=0.962$) and a power law fit are shown; comparison favors the former as a better representation of the data.
    }
    \label{fig:netscience}
\end{figure}

The case of a large metric backbone ($\tau^m > 75\% $) that is mostly ultra-metric ($\tau^u/\tau^m > 80\% $) is rare in the networks we analyzed. Indeed, it is only observed in scientific co-authorship networks, such as the one of Network Science scientists shown in Figure \ref{fig:netscience}, which are known to be particularly modular \cite{newman2006community}.
These networks contain many bridges between largely decoupled communities, and bridges must be on the backbone (Corollary \ref{th_backbone_bridges}). Our results reveal that communities tend to be not only metric but ultra-metric. This means that the shortest path between authors in the same community is by far the direct edge between them (they are direct co-authors).
%---communities that obey the metric and even the ultra-metric triangle inequality.
%
In contrast, the other social networks we analyzed (especially the primary and high school social contact networks) are much less metric. Even though these networks also have strong community structure \cite{Salathe2010highresolution,sociopatterns2011primaryschool,epibackbone_working} most social ties are stronger via indirect connections that break the triangle inequality.
The scientific collaboration networks studied are therefore much more fragile to random attack than the other social networks. Given their very large metric and ultra-metric backbones, removal of a random edge is likely to strongly affect the distribution of shortest paths. This is especially true in the collaboration network of Network Science scientists shown in Figure \ref{fig:netscience}, where $\tau^u/\tau^m \approx 94\%$, which is the most strongly triangular ultra-metric network we analyzed.

Finally, to emphasize that unlike shortest paths the clustering coefficient depends on both metric and semi-metric edges, Table \ref{tableClustering} in SI shows the values of this coefficient for several networks and their backbones. As expected, the removal of semi-metric edges tends to reduce the clustering coefficient, which is essentially null for the ultra-metric backbone. A more detailed study of the distance backbone and community structure is beyond the scope of this article, and is forthcoming \cite{epibackbone_working}.

%\begin{table*}[!th]
%    \centering
%    \scalebox{0.8}{
%    \begin{tabular}{c|llll}
%    \toprule
%    {} &  {US-airports-500} &  {HCN-Coarse} & {C-Elegans} & {MyLib-Keywords}  \\
%    \midrule
%    $D(X)$    & 0.6175  & 0.7165  &  0.2924  & 0.9301  \\
%    $D^m (X)$ & 0.1671  & 0.1318  &  0.0745  & 0.1919  \\
%    $D^u (X)$ & 0.0     & 0.0     &  0.0     & 0.0     \\
%    \bottomrule
%    \end{tabular}
%  }
%  \caption{Watts and Strogatz (average) clustering coefficient for some distance graphs and their metric and ultra-metric backbones. Note that the clustering coefficient treats graphs as unweighted, that is edge weights are assumed to be 1.}
%  \label{tableClustering}
%\end{table*}

\section{Discussion}
\label{section_discussion}

\subsection{Related Concepts}
\label{section_related_concepts}

%Transitive Reduction

The distance backbone is related to but distinct from the concept of \emph{transitive reduction} in graph theory \cite{aho1972transitive}. The latter is the smallest directed graph that has the same reachability, and thus transitive closure, as the original graph.
The transitive reduction has been expanded to weighted, directed graphs \cite{klamt2010transwesd} and general fuzzy graphs and relations \cite{chakraborty1985reduction}, but developed mostly for directed, unweighted graphs. 
Its algorithmic complexity is larger than that of the distance backbone (via Dijkstra's algorithm or distance product) precisely because it accounts for directed graphs that can be acyclic or cyclic.
Certainly, the distance closure framework can be expanded to deal with directed graphs by relaxing the symmetry axiom of distance functions, which become quasimetrics \cite{simas_rocha_2014_MWS}. Indeed, this has been done to compute shortest paths in the large-scale Wikipedia knowledge graph for automatic fact-checking \cite{ciampaglia2015computational}, even though in this case there was no performance advantage over using the undirected distance graph 
%(\S Results; 
\S \ref{section_results};
Table \ref{tableSM}). Here we stick to generalized metric spaces that retain the symmetry axiom, but exploration of quasimetric spaces is certainly a future development possibility.

A very important difference is that the distance backbone does not necessarily return the smallest graph whose closure is the same as that of the original graph. It includes \emph{all} triangular edges, some of which may be removed and still yield the same distance closure graph if there are alternative paths of the same length in the backbone.
For instance, in the metric backbone example of Figure \ref{fig:metric_backbone}, if the direct distance between $x_i$ and $x_j$ were $d_{ij}=10$ (rather than 9), edge $d_{ij}$ would still be in the metric backbone (because the triangle inequality would not be broken) but not in the transitive reduction.
Finally, the general distance closure framework for weighted graphs that we pursue allows us to explore (infinitely) many transitivity criteria based on alternative path length measures (via td-norm $g$ and eq. \ref{eq_length}) and distinct path aggregation measures (via td-conorm $f$). Here we explore only shortest-path distance closures and backbones ($f \equiv \min $), but other weighted graph criteria such as diffusion \cite{coifman2005geometric} and diffusion-like distances, such as resistance \cite{estrada2010resistance, bozzo2013resistance} and communicability \cite{estrada2012complex, silver2018tuned}, are also possible \cite{simas_rocha_2014_MWS}\footnote{To calculate diffusion distances, rather than assume the distance between a node pair is a single shortest path, we integrate the length of all possible paths between the pair with, for example, an averaging operator. This allows us to consider situations where information (or other phenomena) traverse networks by some form of random or stochastic walk on locally available paths. Our distance closure methodology includes such types of distances \cite{simas_rocha_2014_MWS}, however, diffusion-like distance backbones likely require additional convergence criteria because, for the general case, the transitive closure is reached only for $\kappa \rightarrow \infty$ in \ref{eq_TC1}. In contrast, for the shortest-path distance closures we pursue here, we have $\kappa \leq |X|-1$ (\S \ref{section_TC}), which leads each one to convergence to a (unique) distance backbone in a (small) finite number of matrix compositions. Therefore, extension of distance backbones to diffusion-like distances is left for future work.}.
%

%other reductions beyond transitive reduction

%Thresholding

In addition to transitive reduction, in network science there are several other graph reduction techniques available. The most obvious is to remove edges above a certain distance value $d_{ij} > \alpha$, known as \textit{thresholding} distance edge weights. Unfortunately, with this method bridges that link graph components can easily be removed, dramatically altering the distribution of shortest paths. In contrast, the distance backbone guarantees that the distribution of shortest paths is unaffected with the removal of semi-metric edges.
Figure \ref{fig:backbone_thresholds} in SI shows the destruction of connectivity and shortest paths by thresholding the simple example graph of Figures \ref{fig:metric_backbone} and \ref{fig:ultra_backbone}.
Furthermore, because real-world networks tend to be organized in a multiscale manner whereby distance weights are not uniformly and independently distributed but can be hierarchically organized and locally correlated, a global threshold can easily remove features and structures that are present only at a large distance/length scale (or low proximity/strength) \cite{serrano}. In contrast, the distance backbone preserves the original graph connectivity (and all shortest paths), including long-range path lengths. As discussed below, in forthcoming work we show that, for social contact networks at least, the metric backbone also preserves community structure \cite{epibackbone_working}.

%MST

Unlike edge thresholding, the \textit{minimum spanning tree} (MST) of a distance graph preserves the connectivity of the original graph. The MST is the acyclic subgraph that connects all edges of a connected graph
%\footnote{Or connects all edges in each component assembling a minimum spanning forest.} 
and minimizes the sum of all distance edges included.
Because as a tree it is acyclic, the MST graph reduction is very destructive to local community structure\cite{serrano}.
Importantly, unlike the distance backbone (Theorem \ref{th_backbone_is_sufficient}), the MST is not sufficient to compute the metric closure\footnote{Since the MST is defined for summing edges, the metric closure ($g=+$) is the appropriate comparison. However, even considering extensions of the MST concept to minimize any path length measure (with any $g$ in eq. \ref{eq_length}), the MST subgraph, by construction, is not sufficient to compute the respective distance closure.}. 
In other words, the MST reduction does not guarantee the preservation of the distribution of shortest paths in the original graph.
Figure \ref{fig:backbone_thresholds}B in SI shows the MST of the simple example graph of Figures \ref{fig:metric_backbone} and \ref{fig:ultra_backbone}. In this case, the metric edge $d_{ij}$ is removed from the MST resulting in an increase of the shortest path between nodes $x_i$ and $x_j$: $d^{T,m}_{ij}: 9 \rightarrow 10$.

%%% Disparity filter backbone

Several graph reduction techniques center on the extraction of a backbone subgraph. In particular, the \textit{disparity filter backbone} \cite{serrano} has been proposed precisely to deal with the shortcomings of thresholding and MST graph reductions to preserving the multiscale structure of complex networks---even for large distance weights (or low strength/proximity).
However, as a statistical technique, it requires a null model with a significance level parameter to maintain (not filter out) the structure that larger distance weights can bring.
In contrast, the distance backbone does not need a statistical null model as it is axiomatically defined by the chosen geometry monoid 
(\S \ref{section_distance_backbone}).
%(\S Results).
%
Moreover, the distance backbone guarantees the same connectivity and no change to the shortest path distribution, unlike the disparity filter which has no such guarantee and can even remove nodes. 
For instance, the disparity filter backbone of the \textit{U.S.-airports-2006} network of over 1000 airports presented in 
%\S Results 
\S \ref{section_results}
is composed of $24\%$ of the original edges (for a significance level $\alpha=0.2$, per Table 1 in \cite{serrano}). In addition to removing $76\%$ of all edges, it also removes $23\%$ of all nodes. Moreover, shortest paths are not preserved  by this reduction technique even for the $77\%$ of nodes that remain (Table 1 and Figure 1 in \cite{serrano}).
In contrast, the metric backbone of the same U.S. airport traffic network shown in Figure \ref{fig:US-Airports-2006} is almost half the (edge) size ($\tau^m = 16.14\%$) but keeps all node variables with the same connectivity and shortest path length distribution.
%
%\footnote{We computed the \textit{U.S.-airports-2006} from the same data reported in \cite{serrano}, \S Materials, methods \& data. However, there were minor differences to the number of nodes and edges, which nonetheless do not change our conclusions.}. 
%
Indeed, many metric backbones shown in Table \ref{tableSM} are substantially smaller and still preserve these important characteristics of the original distance graphs, and tend to preserve community structure in social networks \cite{epibackbone_working}.

%Pathfinder networks

Another related concept is that of \emph{pathfinder networks} \cite{schvaneveldt1990pathfinder}, which, similarly to distance backbones are obtained by removing graph edges that break the triangle inequality for different ways of computing path length. They are extracted from distance graphs by adjusting two parameters: $r$ sweeps a space of transitivity criteria by varying how path length is computed, and $q$ sets the maximum number of edges that are considered to compute indirect paths.
Typically, $r$ sweeps a space restricted to the Minkowski distance, of which it is a parameter. For $r=1$, path length is the sum of the edges, and for $r=+\infty$ it is the maximum edge in the path---equivalent to our formulations $g \equiv +$ and $g \equiv \max$, respectively.

The distance closure approach \cite{simas_rocha_2014_MWS} we pursue is instead based on the more general algebraic space formed by t-norms and t-conorms from probabilistic metric spaces and fuzzy set theory \cite{Klir1995} 
(see \S \ref{section_back}).
%(\S Introduction).
%
In this formalism, operation $g$ that defines path length (eq. \ref{eq_length}) is not circumscribed to the Minkowski distance or any specific metric, but rather expanded to any metric space monoid. Furthermore, while in the present work we set $f \equiv \min $ to focus on shortest-path distance closures, the framework allows for other path aggregations such as averaging, diffusion, or diffusion-like distances \cite{simas_rocha_2014_MWS} (see discussion above).
Additionally, since we always compute the full distance closure 
(\S \ref{section_isomorphism}), 
%(\S Introduction),
we consider all possible violations of the generalized triangle inequality for any path length, not just length $q$. In other words, in our formulation, $q$ is always maximized to the graph diameter---which is equivalent to $\kappa$ in eq. \ref{eq_TC1} and $\kappa \leq |X|-1$ 
(\S \ref{section_back}).
%(\S Introduction).

Importantly, while we remove semi-triangular edges to reveal the distance backbone, we do not consider them irrelevant nor do we throw them away.
Semi-triangular behavior and properties are very relevant to various network science concepts including clustering, expanding the notion of betweeness centrality, and understanding the robustness of shortest paths to random attack better, as shown in 
\S \ref{section_results}.
%\S Results.
%
Therefore, we characterize the semi-triangular distortion, overall proportion, and distribution in the graphs we analyze. %(\S \ref{section_measures_triang}), %
Indeed, precisely characterizing the distortion inherent in semi-triangular edges has led to the development of competitive recommender systems \cite{Rocha2002,Rocha2005,Simas2012,simas_rocha_2014_MWS}, link prediction tasks in computational biology \cite{abihaidar_GB08}, information extraction in social media \cite{correia2016monitoring}, and even to the ability to distinguish healthy from autistic, depressive, and psychotic human cohorts from brain (fMRI) networks \cite{simas2015semi}.
Thus, the distance backbone methodology provides a complete characterization of triangular organization that goes beyond the edge removal procedures of pathfinder networks.

\subsection{Future development: improving explanation in network phenomena}
\label{section_future}

%Towards dynamics, epidemics

Since the distance backbone is composed of the set of edges sufficient to compute all shortest paths for a given length measure, we expect it to include the preferred and most parsimonious communication paths in the network. Furthermore, because the distance backbone is typically very parsimonious, the paths it contains serve as preferred ``lines of argumentation'' to explain and visualize important paths in the network.  Following are current and future areas of potential development.

Epidemiology is one of the fields where network and data science have led to concrete advances \cite{kraemer2020effect, y2018charting,wang2020response} since the structure of the contact network in which a disease propagates plays a crucial role. 
Heterogeneous (scale-free) networks strongly favor spread \cite{y2018charting} so distance backbones and the heterogeneity of semi-triangular distortion are likely to be relevant to the study of social contact networks and their role in disease spread.
In forthcoming work we test the hypothesis that the metric backbone comprises the most relevant pathways for epidemic spread processes on social contact networks. We show that in social networks built from contact data, the metric backbone: (a) preserves the original community structure; (b) is a subgraph much smaller than the original distance graph (e.g. primary- and high-school social contact networks in Table \ref{tableSM} have  $\tau^m = 9.5$ and $7.84\%$, respectively); (c) is by far the preferred subgraph of the same size for epidemic transmission; and (d) preserves connectivity, though as discussed above,
%(\S \ref{section_related_concepts}),  
other reduction techniques do not guarantee \cite{epibackbone_working}.
Because deleting edges on the backbone, especially the ultra-metric backbone, is likely to result in a measurable impact on average shortest-paths in the network 
%(\S Results),
(\S \ref{section_results}),
it would be useful to highlight them as preferred disease spreading pathways in actionable epidemiological models. 
Going forward, we will test different distance backbones on the same networks as well as the robustness of the epidemic processes to deletions of specific social connections. Since the metric backbone is the preferred transmission subgraph, containment strategies to most curb epidemic transmission can be studied. Indeed, deletion of ultra-metric edges or edges that must be replaced by highly distorted semi-metric edges is most likely to substantially increase the average path length 
(\S \ref{section_results}) 
%(\S Results) 
and thus overall disease propagation speed.

% In brain information integration

The study of human brain structural connectivity has revealed characteristic network features. For instance, various studies suggest that inter-modular bridge edges and connector nodes are critical for information integration in the human brain \cite{sporns2011networks}. This analysis hinges on computing betweeness centrality of nodes and edges, as well as how much they contribute to modularity. Therefore, computing the metric backbone of structural connectivity in human brain networks is likely to provide a more nuanced characterization of important nodes and edges in information integration (especially for bridges that are all on the backbone) and those that only contribute to modularity and shortest-path robustness (edges with varying degrees of semi-metric distortion).
Analysis of different distance closures are expected to further add to the toolbox to study structural connectivity of the human brain, as they accommodate various normalization schemes for human brain connectivity measurements 
(\S \ref{section_other_backbones}).
%(\S Results).
%
Our analysis 
%(\S \ref{section_results}, 
of a human connectome network uncovered high topological redundancy: only $\tau^m=9\%$ or $17.57\%$ of edges contribute to the metric backbone, depending on coarser or finer brain parcellation, respectively %(\S Results,
(\S \ref{section_results},
Figure \ref{fig:hagmann-66-998}). The metric backbone of the functional activity correlation among brain regions of the observed human cohort is even smaller: $\tau^m=5.5\%$.
We also found a wide variation of semi-metric distortion in the human connectome (Figure \ref{fig:hagmann-66-998}).
In forthcoming work we expand this analysis to study the development of structural connectivity in longitudinal human cohorts \cite{neurobackbone_working}.
Indeed, the semi-metric analysis of the human connectome and functional activity has already proven useful to distinguish autistic and depressive cohorts from healthy subjects \cite{simas2015semi,simas2016commentary}.

%Computational complexity

Another promising use of distance backbones is in reducing the computational complexity of problems involving large data sets and networks, whereby computational complexity of shortest-path calculation can be decreased by removing semi-triangular edges. 
For instance, the Wikipedia knowledge graph used for fact-checking is composed of 14 million edges \cite{ciampaglia2015computational}. Yet, 98\% of these are semi-metric and can be removed in the calculation of shortest-paths used by the fact-checking algorithm. In other words, the metric backbone, which is sufficient to compute all shortest paths and thus all knowledge inferences, is only 2\% of the original network. Removal of of the semi-metric edges results in substantial storage and computation gains.

Likewise, related approaches to infer drug interactions from social-media users and electronic health records also rely on shortest-path calculations on distance graphs built from co-occurrence of drug names and biomedical terminology \cite{correia2020mining}. For instance, the metric backbone of a network built from an \textit{Instagram} cohort of $\approx 7,000$ users associated with depression \cite{correia2016monitoring} contains only 
%$\approx 14\%$ 
$\tau^m \approx 8\%$ 
of the edges of the original network (Table \ref{tableSM}).
Interestingly, a previous analysis of this network showed that $\approx 86\%$ of the users in the cohort contributed to the edges in the metric backbone, thus we can remove $\approx 14\%$ of the user timelines from the data set since those only contribute to semimetric edges not on the backbone ($\approx 980$ users) \cite{correia2019thesis}.  In other words, a substantial proportion of users in the data set does not contribute to shortest paths of the resulting network.
Similar proportions of redundant users are observed in analysis of related data sets on \textit{Twitter}. 
Due to the large-scale nature of these graphs, removing semi-metric redundancy directly from graphs or from underlying data sets can significantly improve all network inference algorithms for link-prediction and recommendation that are based on shortest-path calculation \cite{simas_rocha_2014_MWS}.

%Visualization

%Is this in FCT proposal?

A related potential use of distance backbones is in simplifying network visualization.  The distance backbone is the subgraph that exists on a specific geometry defined by the length function chosen and its associated generalized triangle inequality 
%(\S Results). 
(\S \ref{section_results}). 
In contrast, semi-triangular edges exist off this geometry. Therefore, distance backbones computed for geometries that are more amenable to 2D or 3D visualization---for example, the metric and Euclidean backbones---allow for a sequential visualization procedure whereby backbone edges are first rendered in direct proportion to their actual distance weights, and semi-triangular edges are subsequently rendered in increasing order of semi-triangular distortion or simply omitted. This would result in a network visualization anchored to a desired, natural geometry.

%Furthermore, because the distance backbone is typically very parsimonious, the paths it contains serve as ideal ``lines of argumentation'' to explain and visualize important paths in a network, for example, why a certain knowledge inference or recommendation based on shortest paths is made. 
%
%Also, because deleting edges on the backbone, especially the ultra-metric backbone, is likely to result in a measurable impact on average shortest-paths in the network 
%(\S \ref{section_results}),
%(\S Results), it would be useful to highlight them as preferred disease spreading pathways in actionable epidemiological models. 

\subsection{A complete characterization of triangular organization}
\label{section_complete}

%%%Axiomatic separation of two types of edges in distance networks

The general distance closure and backbone methodology %(\S Results) 
(\S \ref{section_results}) 
provides a comprehensive understanding of the triangular geometry of complex networks. It includes both the edges that obey generalized triangle inequalities and the semi-triangular edges that break those inequalities, which are not directly characterized by existing complex network measures.
The methodology grounds weighted graphs in well-known geometric axioms of generalized metric spaces and provides both a principled reduction technique for weighted graphs that preserves all shortest paths (the distance backbone) and a novel characterization of the two types of edges and their relationship to network robustness and evolution.
Specifically, in our approach the amount of (shortest-path) redundancy in a distance graph $D(X)$ is given by $\sigma^g$ (eq. \ref{eq_percentage_red}), or its dual measure of the relative size of the distance backbone $\tau^g=1-\sigma^g$ (eq. \ref{eq_percentage_backbone}), for any path length operation $g$ (eq. \ref{eq_length}). In addition, the semi-triangular distortion $s_{ij}^g$ (eq. \ref{eq_edge_distortion}) measures how much a generalized triangle inequality is broken for each edge $d_{ij}$ of $D$, while the distribution of this measure characterizes the overall semi-triangular geometry of graph $D$.
This way the approach characterizes both the edges that contribute to (generalized) shortest-paths as well as those that do not but are involved in other network phenomena such as modularity and diffusion.

%%% Semi-triangular edges are pervasive  and give robustness. 

\begin{figure}[ht!]
    \centering
    \includegraphics[width=\textwidth]{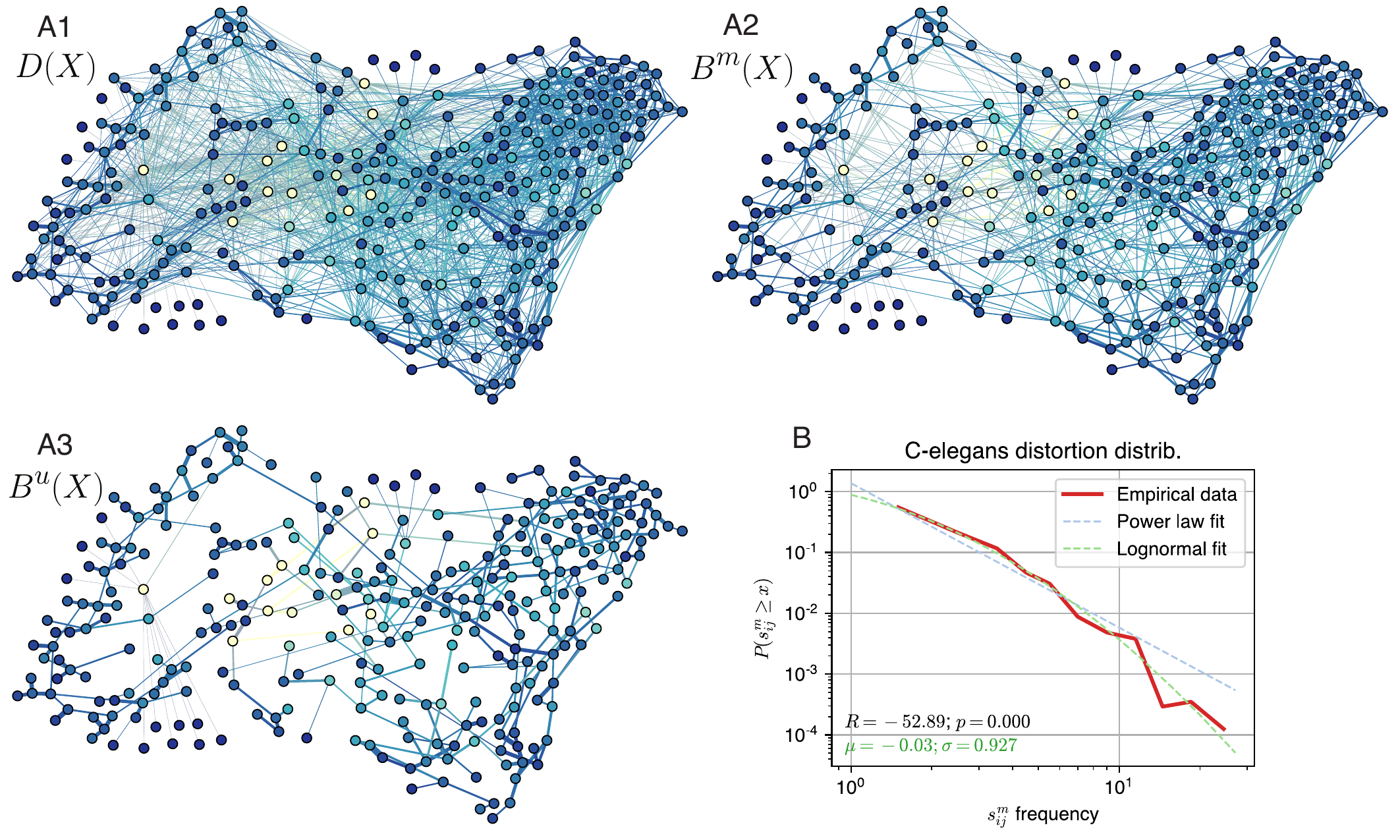}
    \caption{
        \textbf{C-elegans network and backbones.}
        \textbf{A1}. Original C-elegans distance network.
        \textbf{A2}. Metric backbone $D^m(X)$ with $\tau^m = 46.97\%$ of original edges.
        \textbf{A3}. Ultrametric backbone $D^u(X)$ with $\tau^u = 13.97\%$ of original edges.
        \textbf{C}. Log-binned distribution of semi-metric distortion $s^m_{ij}$ values for the $\sigma^m = 53.03\%$ of semi-metric edges in network.
        A log-normal ($ \langle s_{ij}^m \rangle = -0.03$; SD $=0.927$) and a Powerlaw fit are shown; comparison favors the former as a better representation of the data. 
    }
    \label{fig:c-elegans}
\end{figure}

%%% discuss C-elegans network here, contrast with HCN and HCN-Physical and NetScience. 

Our analysis of real-world networks demonstrates that semi-triangular edges are pervasive in networks across domains from biology to technology (Table \ref{tableSM}).
The results strongly suggest that the proportion of such edges ($\sigma^g$) plays an important role in complex networks, especially to increase the robustness of shortest-paths to random attacks.  
For instance, information processing in the human brain that depends on shortest-paths that structurally link brain regions seems to be very robust given that most edges are redundant for shortest-path calculation (Figure \ref{fig:hagmann-66-998}). 
Similarly, air traffic across U.S. airports (Figure \ref{fig:US-Airports-2006}) and automatic inference on knowledge networks are robust to most removals of connections between cities and concepts/keywords, respectively. 
%
%%% semi-triangular organization is not universal
%
However, that type of robustness is not universal. For instance, the \textit{C-elegans} neuron network shown in Figure \ref{fig:c-elegans} has, proportionally, a much larger metric backbone ($\tau^m \approx 47\%$) than the \textit{HCN-Fine} network ($\tau^m \approx 18\%$) and is thus a comparatively more triangular network. Therefore, random removal of synaptic connections between neurons in the \textit{C-elegans} is more likely to result in the increase of shortest paths between neurons: roughly a 1 in 2 chance of increase, versus a 1 in 6 chance when axonal pathways between brain regions in the human connectome are removed.
This is coherent with the fact that the neural architectures of c-elegans and humans are quite different. While the former is made up of fewer neurons precisely developed to implement specific functions, neurons in the human brain are thought to be more general-purpose with much greater redundancy in structural connectivity.

%below maybe for FCT

The co-authorship networks have shown an even stronger triangular organization with a majority of edges that obey the strongest ultra-metric triangle inequality, as seen in Figure \ref{fig:netscience}, 
%(\S Results). 
(\S \ref{section_results}). 
In contrast, all the other social networks studied are semi-triangular with a majority of edges that break the metric (and ultra-metric) triangle inequality. 
From an information transmission viewpoint, the most likely transmission between any two people in the social contact networks is via an indirect path, whereas in the co-authorship networks it is via a direct link.
The reason for the difference is likely because of how the social ties are formed, and specifically the social investment inherent in each case. Social contact networks are built from observations of the time people spend in the vicinity of one another in a particular context (e.g., school). Therefore, there are many fortuitous interactions with small social investment. In contrast, co-authoring scientific articles constitutes a substantial interaction with less chance of fortuitous interactions. Indeed, the density $\delta(D)$ of the co-authorship networks is much smaller.
This suggests that the metric backbone of social contact networks is a more accurate representation of a true, invested social structure of the people involved, as we explore in forthcoming work \cite{epibackbone_working}.

Similar reasoning also suggests that neural connections in C-elegans, given the small number of neurons, constitute a substantial functional investment that leads to a more triangular network organization characterized by a large metric backbone as shown in Figure \ref{fig:c-elegans}. In contrast, given the large number of neurons in the human brain, structural connections may represent a smaller functional investment leading to an overwhelmingly semi-triangular organization characterized by a very small metric backbone as shown in Figure \ref{fig:hagmann-66-998}. In a sense, the latter can afford connection redundancy for greater robustness but the former cannot. 
In forthcoming work we study the relationship between triangular organization and structural cost in the human brain connectome \cite{neurobackbone_working}.

%

%%%semi-metric distortion can establish strong wormholes

While edges off the metric backbone do not contribute to shortest paths and have null betweeness centrality, the distribution of their semi-metric distortion ($s_{ij}^m$) is meaningful.
When a distance graph is almost metric (e.g., \textit{HCN-Physical} in Table \ref{tableSM}), there is very little semi-metric distortion and most edges either obey the triangle inequality or barely break it as seen in Figures \ref{fig:hagmann-66-998}D and \ref{fig:hagmann-66-998}E. However, most real-world networks we have observed display a wide variation of semi-metric distortion as seen in Figures \ref{fig:US-Airports-2006}B and \ref{fig:hagmann-66-998}C. This means that the shortest path length between many pairs of nodes can be substantially smaller via indirect paths in the metric backbone than via a direct connection---over 100 times smaller for the human brain regions in the HCN, and over 10,000 times smaller for some airport pairs in the U.S. Airport network, as seen in Figures \ref{fig:US-Airports-2006}B and \ref{fig:hagmann-66-998}C.
In this sense, the distance backbone functions metaphorically as a wormhole, greatly and indirectly shortening the (semi-triangular) distance between many pairs of nodes.
%
%
%%%semi-metric distortion distribution picks relevant pairs
%
Interestingly, edges that substantially break the triangle inequality have been shown to be relevant to associations in network inference, link prediction, recommender systems, and even segmentation of healthy and diseased human-brain phenotypes \cite{simas_rocha_2014_MWS,correia2016monitoring,simas2015semi}.  
This suggests that such edges denote associations that are more likely to be stronger in the future or with more data. Given the very strong indirect connection (via the backbone), in time, diffusion of information is likely to more strongly, directly associate nodes that are presently related by very semi-metric edges, evolving the networks toward a more triangular organization overall \cite{Rocha2002}. 
%

%%%Comparison to betweeness centrality

The analysis of the triangular geometry thus provides a more nuanced characterization of edges in weighted graphs than does betweeness centrality. In the case of the metric closure (which uses the most traditional path length operation $g \equiv +$), all semi-metric edges have null betweeness centrality, but their semi-metric distortion varies widely and meaningfully, as we have shown.
Furthermore, our approach in effect generalizes betweeness centrality to include the (triangular) edges that do contribute to shortest paths. By considering more or less stringent length operations $g$, the resulting distance backbone reveals the most important edges for shortest paths. For instance, edges on the ultra-metric backbone contribute more to shortest paths than do those that are only on backbones computed with less stringent triangular constraints.
Therefore, the concept of triangular edges is more general than betweeness centrality, since it considers all possible shortest-path length measures. 
Moreover, it does not require an algorithm to compute all and how many shortest paths pass through a given edge, but simply the length of the shortest path between the nodes of each edge (the APSP.)
All together, the distance backbone analysis provides a principled graph reduction technique and also allows a finer characterization of how the triangular geometry of real-world networks affects shortest paths and thus more generally, information transmission in complex multivariate systems.

%discuss null models of triangular behavior, for example, based on the HCN-Phsyical. But to contrast with natural behavior of very large multivariate complex systems

% it generalizes distribution of shortest paths, thus robustness, as well as efficiency, etc.
%

%LMR: below from metric backbone section, summarize here and link to future work on modularity with Alain

%directed graphs and also wikipedia work
%From wikipedia fact-checking work, "the metric closure on the undirected graph gives the most accurate results. Therefore, we continue to use this combination in our semantic proximity computations when performing the validation tasks described below." We tested several closures, but in this massive knowledge graph the metric backbone (via metric closure) yielded the best inferences and automatic fact-cheking results. "we evaluated alternative definitions of  and found Eq 1 to perform " best.

\section{Materials, methods, and data}
\label{section_methods}

\subsection{Network Data}
\label{section_network_data}

To allow for easier comparisons, we used only the largest connected component of each distance graph described below. In every case, the set of nodes of $D(X)$ in Table \ref{tableSM}, $|X|$, denotes the number of nodes in the largest connected component. Therefore, for every network, $D(X)$ and $B^g(X)$ are connected graphs, and $D^{T,g}$ is a complete graph. 
We used the following networks, which we converted to distance graphs from proximity/strength versions via eq. \ref{eq_iso_map}, except where noted.

%
% Tech
%

%\subsubsection{Technological}

\vspace{2mm} %5mm vertical space

\noindent \textbf{Technological}
\begin{itemize}
\item The \textit{U.S.-airports-500} network is a distance graph of the 500 busiest commercial airports in the United States. An edge exists between two airport-nodes if a flight was scheduled between them in 2002.
Edge weights are a normalized measure of traffic (available airplane seats) between airports \cite{colizza2007reaction}.
\item \textit{U.S.-airports-2006} is the domestic nonstop segment of the U.S. airport transportation system for the year 2006, retrieved from \url{http://www.transtats.bts.gov}. Edge weights are 
%
%inversely proportional
the normalized average number of passengers traveling between two airport-nodes. Airports in the American Samoa, Guam, Northern Marianas, and Trust Territories of the Pacific Islands have been removed. 
%
%In addition, only the largest connected component was used for computation. 
%
This network is a reconstruction of the one used by Serrano et al.\cite{serrano}.
\end{itemize}

%

%\subsubsection{Biological}

%
% Biological
%

\noindent \textbf{Biological}
\begin{itemize}

\item \textit{C-elegans} is the \textit{Caenorhabditis elegans} worm (\textit{C.elegans}) neural network. Each node is a neuron and edge weight is a normalized measure of the number of synapses or gap junctions between two neuron-nodes \cite{WATTS1998}.

\item Four distinct  \textit{Human Connectome} networks (HCN) are obtained from the group averages of five human participants in a study to map the human brain \cite{hagmann2008mapping}. The edge weights of the %\textit{HCN-ROI}
\textit{HCN-Coarse}
and 
%\textit{HCN-Density}
\textit{HCN-Fine}
networks are a normalized measure of the volume of cortico-cortical axonal pathways between human brain regions (nodes) obtained via diffusion spectrum imaging (DSI).
The edge weights correspond to the number of tractography streamlines linking two region-nodes, divided by the combined volume of the two regions.
The two networks are built from the same data but \textit{HCN-Coarse} is a coarse-grained representation of \textit{HCN-Fine} that results in larger network density (Table \ref{tableSM}.
\textit{HCN-Coarse} is based on a parcellation of 66 regions of interest (ROI), whereas \textit{HCN-Fine} is based on a finer structural connectivity matrix of 998 nodes for which we kept the 989 nodes that form its largest connected component.
The edge weights of the \textit{HCN-Physical} network denote the physical length of each connection in the 998 node structural connectivity matrix---in general a little longer than the euclidean distance since connections are curved in the interior of the human brain. In this case, the distance graph uses these edge weights directly, rather than via eq.\ref{eq_iso_map}.
Finally, \textit{HCN-fMRI} is the resting-state functional magnetic resonance imaging matrix of the 998 regions. The original data is in a correlation matrix which for edge length we convert to a proximity measure by considering the absolute value of positive and negative correlations.

\item \textit{Enterocyte GRN} is a gene interaction network retrieved from STRING \cite{StringDB} for genes expressed in insect intestinal cells. Edge-weights between gene-nodes denote a confidence score that the genes are known to interact based on experimentally and computationally derived evidence sources \cite{enterocyte_working}.

\end{itemize}

%\subsubsection{Social}
\noindent \textbf{Social}
\begin{itemize}
%
% Social Networks
%
\item \textit{High-school} is a face-to-face contact network of a typical high school day in the United States (unspecified city/state) gathered in 2010. Data were collected from students, teachers, and staff via wireless ``sensor network motes'' (TelosB; Crossbow Technologies Inc.), with data covering 94\% of the entire school population \cite{Salathe2010highresolution}. Temporal contacts between pairs of individuals/nodes have been summed and normalized by their total interactions to calculate edge-weights and reduce the temporal network to a weighted graph.

\item \textit{Primary-school} is also a contact network. Data are between 242 individuals in a primary school in Lyon, France---232 students and 10 teachers. The school comprises 5 grades, each grade with two classes, each class with an assigned room and an assigned teacher. Lunches are served in a common canteen, and a shared playground is located outside the main building. As the playground and the canteen do not have enough capacity to host all the students at the same time, only two or three classes have concurrent breaks, and lunches are taken in two consecutive turns. Contacts were recorded using active RFID devices embedded in unobtrusive, wearable badges. The badges exchanged multi-channel, bi-directional radio communication \cite{sociopatterns2011primaryschool}. Similar to the high-school data set, temporal contacts were reduced to a weighted graph of interactions.

\item \textit{Freeman} is a network built from a data set collected in 1978 that contains the frequency of message exchange among 32 researchers working on social network analysis via an electronic communication tool \cite{Freeman:1979}.

\item \textit{cond-mat} is a weighted co-authorship network among scientists who had posted pre-prints on the condensed matter e-print archive between January 1\textsuperscript{st}, 1995 and December 31\textsuperscript{st}, 1999 \cite{Newman2001}. Edge weights are a normalized measure of volume of co-authored articles between a pair of scientist-nodes.

\item \textit{cond-mat-2003} is an updated version of the cond-mat network that includes all preprints posted between Jan 1\textsuperscript{st}, 1995 and June 30\textsuperscript{th}, 2003 \cite{Newman2001}.

\item \textit{net-science} is a co-authorship network of scientists working on network theory and experiment as originally compiled by M. Newman \cite{newman2006community}. Edge weights are a normalized measure of volume of co-authored articles between a pair of scientist-nodes.

\end{itemize}

%\subsubsection{Knowledge}
\noindent \textbf{Social}
\begin{itemize}
%
% Knowledge
%

\item \emph{Wikipedia-Fact} is a semantic proximity network of Wikipedia concepts used for large-scale automatic fact-checking \cite{ciampaglia2015computational}. 

\item The \textit{MyLib} networks are from scientific articles and user profiles of the \textit{MyLibrary} Recommender system at the Los Alamos National Laboratory's digital library\cite{Rocha2005}. 
\emph{MyLib-Keywords} is a semantic proximity network of the 500 most frequent keywords in scientific articles accessed and edge weights are co-occurrences in user profiles. \emph{MyLib-Users} is a network of the scientists who utilized the MyLibrary Recommender system and built from co-access patterns to academic journals.  
\emph{MyLib-Journals} is a network of ISSN (academic journals) and built from co-occurrence in user profiles.

\item \textit{Instagram depression} is a co-mention network built from complete Instagram timelines of users who had mentioned at least one drug known to treat depression. Nodes denote terms (i.e., drugs or medical terms) present in the timelines and edges connect terms that were mentioned within a seven day window \cite{correia2016monitoring, correia2019thesis}.
\end{itemize}

\subsection{Computational Methods and Tools}
\label{section_computational}

Network backbones have been computed using the `distanceclosure' python package developed by the authors and freely available at \url{https://github.com/rionbr/distanceclosure}.
Network plots have been rendered with Gephi \cite{Gephi2009}.
Distortion distributions have been fitted using the `powerlaw' python package \cite{Clauset2009powerlaw,Alstott2014powerlaw}.

%----------------------------------------------------------------------
%\bibliographystyle{unsrt}
%\bibliographystyle{ScienceAdvances}
\bibliographystyle{main}
\bibliography{gbib}

\begin{thebibliography}{00}

\bibitem{abihaidar_GB08}
Abi-Haidar, A., Kaur, J., Maguitman, A., Radivojac, P., Rechtsteiner, A.,
  Verspoor, K., Wang, Z. {\&} Rocha, L.~M. (2008)  Uncovering protein
  interaction in abstracts and text using a novel linear model and word
  proximity networks. {\em Genome Biol}, \textbf{9 Suppl 2}, S11.

\bibitem{aho1972transitive}
Aho, A.~V., Garey, M.~R. {\&} Ullman, J.~D. (1972)  The transitive reduction of
  a directed graph. {\em SIAM Journal on Computing}, \textbf{1}(2), 131--137.

\bibitem{albert-2002-74}
Albert, R. {\&} Barabasi, A.-L. (2002)  Statistical mechanics of complex
  networks. {\em Reviews of Modern Physics}, \textbf{74}, 47.

\bibitem{Alstott2014powerlaw}
Alstott, J., Bullmore, E. {\&} Plenz, D. (2014)  powerlaw: A Python Package for
  Analysis of Heavy-Tailed Distributions. {\em PLOS ONE}, \textbf{9}(1), 1--11.

\bibitem{barabasi2016network}
Barab{\'a}si, A.-L.  et~al. (2016) {\em Network science}.
Cambridge university press.

\bibitem{barrat-2004-101}
Barrat, A., Barthelemy, M., Pastor-Satorras, R. {\&} Vespignani, A. (2004a)
  The architecture of complex weighted networks. {\em PROC.NATL.ACAD.SCI.USA},
  \textbf{101}, 3747.

\bibitem{Barrat2004PRL}
Barrat, A., Barth{\'e}lemy, M. {\&} Vespignani, A. (2004b)  Weighted evolving
  networks: coupling topology and weight dynamics.. {\em Phys Rev Lett},
  \textbf{92}, 228701.

\bibitem{barrat_berth_vesp_book_2008}
Barrat, A., Barthelemy, M. {\&} Vespignani, A. (2008) {\em Dynamical processes
  on complex networks}.
Cambridge University Press.

\bibitem{Gephi2009}
Bastian, M., Heymann, S. {\&} Jacomy, M. (2009)  Gephi: An Open Source Software
  for Exploring and Manipulating Networks. In {\em International AAAI
  Conference on Weblogs and Social Media}.

\bibitem{bozzo2013resistance}
Bozzo, E. {\&} Franceschet, M. (2013)  Resistance distance, closeness, and
  betweenness. {\em Social Networks}, \textbf{35}(3), 460--469.

\bibitem{Brandes}
Brandes, U. {\&} Erlebach, T. (2005) {\em Network Analysis Methodological
  Foundations}.
Springer.

\bibitem{burda2004network}
Burda, Z., Jurkiewicz, J. {\&} Krzywicki, A. (2004)  Network transitivity and
  matrix models. {\em Physical Review E}, \textbf{69}(2), 026106.

\bibitem{Camerini197810}
Camerini, P. (1978)  The min-max spanning tree problem and some extensions.
  {\em Information Processing Letters}, \textbf{7}(1), 10 -- 14.

\bibitem{chakraborty1985reduction}
Chakraborty, M. {\&} Das, M. (1985)  Reduction of fuzzy strict order relations.
  {\em Fuzzy sets and systems}, \textbf{15}(1), 33--44.

\bibitem{ciampaglia2015computational}
Ciampaglia, G.~L., Shiralkar, P., Rocha, L.~M., Bollen, J., Menczer, F. {\&}
  Flammini, A. (2015)  Computational fact checking from knowledge networks.
  {\em PloS one}, \textbf{10}(6), e0128193.

\bibitem{Clauset2009powerlaw}
Clauset, A., Shalizi, C.~R. {\&} Newman, M. E.~J. (2009)  Power-Law
  Distributions in Empirical Data. {\em SIAM Review}, \textbf{51}(4), 661--703.

\bibitem{coifman2005geometric}
Coifman, R.~R., Lafon, S., Lee, A.~B., Maggioni, M., Nadler, B., Warner, F.
  {\&} Zucker, S.~W. (2005)  Geometric diffusions as a tool for harmonic
  analysis and structure definition of data: Diffusion maps. {\em Proceedings
  of the national academy of sciences}, \textbf{102}(21), 7426--7431.

\bibitem{colizza2007reaction}
Colizza, V., Pastor-Satorras, R. {\&} Vespignani, A. (2007)
  Reaction--diffusion processes and metapopulation models in heterogeneous
  networks. {\em Nature Physics}, \textbf{3}(4), 276--282.

\bibitem{conrad1990geometry}
Conrad, M. (1990)  The geometry of evolution. {\em BioSystems}, \textbf{24}(1),
  61--81.

\bibitem{epibackbone_working}
Correia, R., Barrat, A. {\&} Rocha, L. (2018)  The Metric Backbone of Contact
  Networks in Epidemic Spread Models. In {\em Network Science 2018}, volume
  Working Paper.

\bibitem{enterocyte_working}
Correia, R., Navarro~Costa, P. {\&} Rocha, L. (2021)  Extraction of overlapping
  modules in networks via spectral methods and information theory. In {\em
  Complex Networks \& Their Applications IX}, volume 943.

\bibitem{correia2019thesis}
Correia, R.~B. (2019) {\em Prediction of Drug Interaction and Adverse
  Reactions, with data from Electronic Health Records, Clinical Reporting,
  Scientific Literature, and Social Media, using Complexity Science Methods}.
PhD thesis, Indiana University. School of Informatics, Computing \&
  Engineering.

\bibitem{correia2019city}
Correia, R.~B., de~Ara{\'u}jo~Kohler, L.~P., Mattos, M.~M. {\&} Rocha, L.~M.
  (2019)  City-wide electronic health records reveal gender and age biases in
  administration of known drug--drug interactions. {\em NPJ Digital Medicine},
  \textbf{2}(1), 1--13.

\bibitem{correia2016monitoring}
Correia, R.~B., Li, L. {\&} Rocha, L.~M. (2016)  Monitoring potential drug
  interactions and reactions via network analysis of instagram user timelines.
  In {\em Biocomputing 2016: Proceedings of the Pacific Symposium}, pages
  492--503. World Scientific.

\bibitem{correia2020mining}
Correia, R.~B., Wood, I.~B., Bollen, J. {\&} Rocha, L.~M. (2020)  Mining social
  media data for biomedical signals and health-related behavior. {\em Annual
  Review of Biomedical Data Science}, \textbf{3}.

\bibitem{dijkstra}
Dijkstra, E.~W. (1959)  A Note on Two Problems in Connexion with Graphs. {\em
  Numerische Mathematik}, \textbf{1}, 269--271.

\bibitem{dombi1982general}
Dombi, J. (1982)  A general class of fuzzy operators, the DeMorgan class of
  fuzzy operators and fuzziness measures induced by fuzzy operators. {\em Fuzzy
  sets and systems}, \textbf{8}(2), 149--163.

\bibitem{Dorogovtsev2003}
Dorogovtsev, S.~N. {\&} Mendes, J. (2003) {\em Evolution of Networks}.
Oxford University Press.

\bibitem{estrada2012complex}
Estrada, E. (2012)  Complex networks in the Euclidean space of communicability
  distances. {\em Physical Review E}, \textbf{85}(6), 066122.

\bibitem{estrada2010resistance}
Estrada, E. {\&} Hatano, N. (2010)  Resistance distance, information
  centrality, node vulnerability and vibrations in complex networks. In {\em
  Network science}, pages 13--29. Springer.

\bibitem{fortunato_review}
Fortunato, S. (2010)  Community detection in graphs. {\em Physics reports},
  \textbf{486}(3-5), 75--174.

\bibitem{Freeman:1979}
Freeman, S.~C. {\&} Freeman, L.~C. (1979) {\em The networkers network: A study
  of the impact of a new communications medium on sociometric structure}.
Social sciences research reports, 46. School of Social Sciences, University of
  California, Irvine.

\bibitem{galvin_shore91}
Galvin, F. {\&} Shore, S. (1991)  Distance Functions and Topologies. {\em
  American Mathematical Monthly}, \textbf{98}, 620--623.

\bibitem{gates2021EfGraph}
Gates, A.~J., Correia, R.~B., Wang, X. {\&} Rocha, L.~M. (2021)  The effective
  graph reveals redundancy, canalization, and control pathways in biochemical
  regulation and signaling. {\em Proceedings of the National Academy of
  Sciences}, \textbf{118}(e2022598118), In Press.

\bibitem{gates2016control}
Gates, A.~J. {\&} Rocha, L.~M. (2016)  Control of complex networks requires
  both structure and dynamics. {\em Scientific reports}, \textbf{6}(1), 1--11.

\bibitem{Girvan2002}
Girvan, M. {\&} Newman, M. E.~J. (2002)  Community structure in social and
  biological networks. {\em Proc. Natl. Acad. Sci. USA}, \textbf{99},
  7821--7826.

\bibitem{Goh2005}
Goh, K.-I., Kahng, B. {\&} Kim, D. (2005)  Nonlocal evolution of weighted
  scale-free networks.. {\em Phys Rev E Stat Nonlin Soft Matter Phys},
  \textbf{72}, 017103.

\bibitem{Gondran2007}
Gondran, M. {\&} Minoux, M. (2007)  Dioids and semirings: Links to fuzzy sets
  and other applications. {\em Fuzzy Sets and Systems}, \textbf{158},
  1273--1294.

\bibitem{hagmann2008mapping}
Hagmann, P., Cammoun, L., Gigandet, X., Meuli, R., Honey, C.~J., Wedeen, V.~J.
  {\&} Sporns, O. (2008)  Mapping the structural core of human cerebral cortex.
  {\em PLoS Biol}, \textbf{6}(7), e159.

\bibitem{hamacher78}
Hamacher, H. (1978)  Uber logische Verknupfungen unscharfer Aussagen un deren
  Zugehorige Bewertungsfunktionen. In Trappl, R., Klir, G. {\&} Ricciardi, L.,
  editors, {\em Progress in Cybernetics and Systems Research}, volume~3, pages
  276--288. Hemisphere.

\bibitem{helikar2012cell}
Helikar, T., Kowal, B., McClenathan, S., Bruckner, M., Rowley, T., Madrahimov,
  A., Wicks, B., Shrestha, M., Limbu, K. {\&} Rogers, J.~A. (2012)  The cell
  collective: toward an open and collaborative approach to systems biology.
  {\em BMC systems biology}, \textbf{6}(1), 96.

\bibitem{klamt2010transwesd}
Klamt, S., Flassig, R.~J. {\&} Sundmacher, K. (2010)  TRANSWESD: inferring
  cellular networks with transitive reduction. {\em Bioinformatics},
  \textbf{26}(17), 2160--2168.

\bibitem{Klamt:2009kx}
Klamt, S., Haus, U. {\&} Theis, F. (2009)  Hypergraphs and Cellular Networks.
  {\em PLoS Comput Biol}.

\bibitem{Klement2004}
Klement, E., Mesiar, R. {\&} Pap, E. (2004)  Triangular norms. Position paper
  II: general constructions and parameterized families. {\em Fuzzy Sets and
  Systems}, \textbf{145}, 411--438.

\bibitem{Klir1995}
Klir, G. {\&} Yuan, B. (1995) {\em Fuzzy sets and fuzzy logic, theory and
  applications}.
Prentice Hall PTR.

\bibitem{kraemer2020effect}
Kraemer, M.~U., Yang, C.-H., Gutierrez, B., Wu, C.-H., Klein, B., Pigott,
  D.~M., Du~Plessis, L., Faria, N.~R., Li, R., Hanage, W.~P.  et~al. (2020)
  The effect of human mobility and control measures on the COVID-19 epidemic in
  China. {\em Science}, \textbf{368}(6490), 493--497.

\bibitem{louch2000personal}
Louch, H. (2000)  Personal network integration: transitivity and homophily in
  strong-tie relations. {\em Social networks}, \textbf{22}(1), 45--64.

\bibitem{martinez2016survey}
Mart{\'\i}nez, V., Berzal, F. {\&} Cubero, J.-C. (2016)  A survey of link
  prediction in complex networks. {\em ACM computing surveys (CSUR)},
  \textbf{49}(4), 1--33.

\bibitem{Menger}
Menger, K. (1942)  Statistical metrics. {\em PNAS}, \textbf{28}.

\bibitem{Mordeson2000}
Mordeson, J. {\&} Nair, P. (2000) {\em Fuzzy Graphs and Fuzzy Hypergraphs}.
Physica-Verlag.

\bibitem{Newman2001_PREII}
Newman, M.~E. (2001a)  Scientific collaboration networks. II. Shortest paths,
  weighted networks, and centrality.. {\em Phys Rev E Stat Nonlin Soft Matter
  Phys}, \textbf{64}, 016132.

\bibitem{Newman2001}
Newman, M.~E. (2001b)  The structure of scientific collaboration networks..
  {\em Proc Natl Acad Sci U S A}, \textbf{98}, 404--409.

\bibitem{newman2011complex}
Newman, M.~E. (2011)  Complex systems: A survey. {\em American Journal of
  Physics}, \textbf{79}(8), 800--810.

\bibitem{newman2006community}
Newman, M. E.~J. (2006)  Finding community structure in networks using the
  eigenvectors of matrices. {\em Phys. Rev. E}, \textbf{74}, 036104.

\bibitem{Oltvai2002}
Oltvai, Z. {\&} Barabasi, A. (2002)  Systems biology. Life's complexity
  pyramid. {\em Science}, \textbf{298}(5594), 763--764.

\bibitem{pastor_vesp}
Pastor-Satorras, R. {\&} Vespignani, A. (2004) {\em Evolution and structure of
  the Internet a statistical physics approach}.
Cambridge University Press, Cambridge, UK.

\bibitem{pescosolido2016linking}
Pescosolido, B.~A., Olafsdottir, S., Sporns, O., Perry, B.~L., Meslin, E.~M.,
  Grubesic, T.~H., Martin, J.~K., Koehly, L.~M., Pridemore, W., Vespignani, A.
  et~al. (2016)  Linking Genes-to-Global Cultures in Public Health Using
  Network Science. {\em Handbook of applied system science}, page~25.

\bibitem{Rammal_ultrametricity}
Rammal, R., Toulouse, G. {\&} Virasoro, M.~A. (1986)  Ultrametricity for
  physicists. {\em Rev. Mod. Phys.}, \textbf{58}, 765--788.

\bibitem{Rocha2005}
Rocha, L., Simas, T., Rechtsteiner, A., DiGiacomo, M. {\&} Luce, R. (2005)
  MyLibrary@LANL: Proximity and Semi-metric Networks for a Collaborative and
  Recommender Web Service. In Press, I., editor, {\em Proc. 2005 IEEE/WIC/ACM
  International Conference on Web Intelligence (WI'05)}, pages 565--571.

\bibitem{Rocha2002}
Rocha, L.~M. (2002)  Semi-metric Behavior in Document Networks and its
  Application to Recommendation Systems. In (Ed.), V.~L., editor, {\em Soft
  Computing Agents: A New Perspective for Dynamic Information Systems},
  International Series Frontiers in Artificial Intelligence and Applications,
  pages 137--163. IOS Press.

\bibitem{Salathe2010highresolution}
Salathé, M., Kazandjieva, M., Lee, J.~W., Levis, P., Feldman, M.~W. {\&}
  Jones, J.~H. (2010)  A high-resolution human contact network for infectious
  disease transmission. {\em Proceedings of the National Academy of Sciences},
  \textbf{107}(51), 22020--22025.

\bibitem{salnikov2018simplicial}
Salnikov, V., Cassese, D. {\&} Lambiotte, R. (2018)  Simplicial complexes and
  complex systems. {\em European Journal of Physics}, \textbf{40}(1), 014001.

\bibitem{schvaneveldt1990pathfinder}
Schvaneveldt, R.~W. (1990) {\em Pathfinder associative networks: Studies in
  knowledge organization.}
Ablex Publishing.

\bibitem{serrano}
Serrano, M.~{\'A}., Bogu{\~n}{\'a}, M. {\&} Vespignani, A. (2009)  Extracting
  the multiscale backbone of complex weighted networks. {\em Proceedings of the
  National Academy of Sciences}, \textbf{106}(16), 6483--6488.

\bibitem{silver2018tuned}
Silver, G., Akbarzadeh, M. {\&} Estrada, E. (2018)  Tuned communicability
  metrics in networks. The case of alternative routes for urban traffic. {\em
  Chaos, Solitons \& Fractals}, \textbf{116}, 402--413.

\bibitem{simas2015semi}
Simas, T., Chattopadhyay, S., Hagan, C., Kundu, P., Patel, A., Holt, R.,
  Floris, D., Graham, J., Ooi, C., Tait, R.  et~al. (2015)  Semi-metric
  topology of the human connectome: sensitivity and specificity to autism and
  major depressive disorder. {\em PloS one}, \textbf{10}(8), e0136388.

\bibitem{Simas2012}
Simas, T. {\&} Rocha, L.~M. (2012)  {Semi-metric networks for recommender
  systems}. In {\em 2012 IEEE/WIC/ACM International Conferences on Web
  Intelligence and Intelligent Agent Technology}, pages 175--179, Macau.

\bibitem{simas_rocha_2014_MWS}
Simas, T. {\&} Rocha, L.~M. (2015)  Distance closures on complex networks. {\em
  Network Science}, \textbf{3}(2), 227--268.

\bibitem{simas2016commentary}
Simas, T. {\&} Suckling, J. (2016)  Commentary: Semi-metric topology of the
  human connectome: Sensitivity and specificity to autism and major depressive
  disorder. {\em Frontiers in Neuroscience}, \textbf{10}, 353.

\bibitem{sporns2011networks}
Sporns, O. (2011) {\em Networks of the Brain}.
MIT press.

\bibitem{sociopatterns2011primaryschool}
Stehlé, J., Voirin, N., Barrat, A., Cattuto, C., Isella, L., Pinton, J.-F.,
  Quaggiotto, M., Van~den Broeck, W., Régis, C., Lina, B. {\&} et~al. (2011)
  High-Resolution Measurements of Face-to-Face Contact Patterns in a Primary
  School. {\em PLoS ONE}, \textbf{6}(8), e23176.

\bibitem{strogatz2014nonlinear}
Strogatz, S.~H. (2014) {\em Nonlinear dynamics and chaos: with applications to
  physics, biology, chemistry, and engineering}.
Westview press.

\bibitem{StringDB}
Szklarczyk, D., Gable, A., Lyon, D., Junge, A., Wyder, S., Huerta-Cepas, J.,
  Simonovic, M., Doncheva, N., Morris, J., Bork, P., Jensen, L. {\&} Mering, C.
  (2019)  STRING v11: protein-protein association networks with increased
  coverage, supporting functional discovery in genome-wide experimental
  datasets. {\em Nucleic Acids Res}, \textbf{47(D1)}, D607--D613.

\bibitem{neurobackbone_working}
Teixeira, A.~S., Faskowitz, J., Sporns, O. {\&} Rocha, L. (2020)  The Metric
  Backbone in the Human Connectome. In {\em Complex Networks 2020}.

\bibitem{vespignani2018twenty}
Vespignani, A. (2018)  Twenty years of network science. {\em Nature},
  \textbf{558}(7711), 528--529.

\bibitem{wang2020response}
Wang, C.~J., Ng, C.~Y. {\&} Brook, R.~H. (2020)  Response to COVID-19 in
  Taiwan: big data analytics, new technology, and proactive testing. {\em
  Jama}, \textbf{323}(14), 1341--1342.

\bibitem{Wang2005a}
Wang, W.-X., Wang, B.-H., Hu, B., Yan, G. {\&} Ou, Q. (2005)  General dynamics
  of topology and traffic on weighted technological networks.. {\em Phys Rev
  Lett}, \textbf{94}, 188702.

\bibitem{WATTS1998}
Watts, D.~J. {\&} Strogatz, S.~H. (1998)  Collective dynamics of 'small-world'
  networks. {\em Nature(London)}, \textbf{393}, 440--442.

\bibitem{y2018charting}
y~Piontti, A.~P., Perra, N., Rossi, L., Samay, N. {\&} Vespignani, A. (2018)
  {\em Charting the Next Pandemic: Modeling Infectious Disease Spreading in the
  Data Science Age}.
Springer.

\bibitem{Zadeh1965}
Zadeh, L. (1965)  Fuzzy sets and systems. In Fox, J., e., editor, {\em System
  Theory}, pages 29--37. Polytechnic Press, Brooklyn, NY.

\bibitem{zuo2012network}
Zuo, X.-N., Ehmke, R., Mennes, M., Imperati, D., Castellanos, F.~X., Sporns, O.
  {\&} Milham, M.~P. (2012)  Network centrality in the human functional
  connectome. {\em Cerebral cortex}, \textbf{22}(8), 1862--1875.

\bibitem{Zwick}
Zwick, U. (May 2002)  All Pairs Shortest Paths using Bridging Sets Retangular
  Matrix Multiplication. {\em Journal of the ACM}, \textbf{49}(3), 289--317.

\end{thebibliography}
%----------------------------------------------------------------------

%\newpage
%\mbox{}
%\newpage

%\section{Supplementary material}

\noindent \textbf{Acknowledgements:}

% Acknowledgments should be gathered into a paragraph after the final numbered reference. This section should also include
% * complete funding information,
% * a description of each authors contribution to the paper,
% * a listing of any competing interests of any of the authors (all authors must also fill out the Conflict of Interest form), and,
% * a section on data and materials availability, information about the location of the data if not included in the paper, including **accession numbers** to any data relating to the paper and deposited in a public database.
%
The authors thank Olaf Sporns and Giovanni Ciampaglia for sharing data used in the analysis and for many useful conversations to understand the results. We are also very thankful to Deborah Rocha for very thorough line editing.

\noindent \textbf{Funding:} RBC was funded by the CAPES Foundation (grant 18668127) and Fundação para a Ciência e a Tecnologia (grant PTDC/MEC-AND/30221/2017).
LMR was partially funded by the National Institutes of Health, National Library of Medicine Program, grant 1R01LM012832-01, by a Fulbright Commission fellowship, and by NSF-NRT grant 1735095 ``Interdisciplinary Training in Complex Networks and Systems.''
The funders had no role in study design, data collection and analysis, decision to publish, or preparation of the manuscript.
\\
\noindent \textbf{Author Contributions} TS and LMR conceived the research. LMR and RBC designed and conducted the analysis. RBC developed the Python code, performed all computations and visualizations. LMR wrote the manuscript with contributions from TS and RBC.\\
\noindent \textbf{Competing Interests} The authors declare that they have no competing financial interests.\\
\noindent \textbf{Data and materials availability:} Additional data and materials are available online.

\newpage

\section*{Supplementary information}

\subsection*{Theorems and Proofs}
\label{section_SI_Proofs}

\begin{theoremSI}[Backbone of non-weighted graphs]
    %\label{th_crisp_backbone}
    If $D(X)$ is a standard, non-weighted graph, then its distance backbone for any td-norm $g$ is the entire graph: $B^g(X) \equiv D(X).$
\end{theoremSI}

\begin{proof}
    The proof is straightforward from the definition of td-norm $g$ via its isomorphic t-norm $\wedge$
    %(\S Introduction).
(\S \ref{section_back}).
When graph $D(X)$ is non-weighted, it means that there is no distinguishing characteristic for the weights. In our framework this means that when an edge exists between two nodes $x_i$ and $x_j$, they are considered to be maximally associated: $d_{ij}=0$. Conversely, if there is no edge, the two variables are minimally associated: $d_{ij}=+\infty$\footnote{In the isomorphic space, all connected nodes have maximum proximity, $p_{ij}=1$, and when there is no edge between $x_i$ and $x_j$ we have minimum proximity, $p_{ij}=0$.}.
Because $0$ is the identity element of any td-norm $g$
(\S \ref{section_back}),
%(\S Introduction),
the generalized triangle inequality (eq. \ref{eq_gen_triang}) cannot be broken for any edge with $d_{ij}=0$.
\end{proof}

\begin{theoremSI}[Backbone Sufficiency]
    %\label{th_backbone_is_sufficient}
    Given a distance graph $D(X)$ defined on (node) variable set $X$, its shortest-path distance closure defined by any td-norm $g$ is equivalent to the same closure of its distance backbone subgraph:
    $D^{T,g}(X) \equiv B^{T,g} (X)$.
\end{theoremSI}

\begin{proof}
    The proof of this theorem is rather trivial. If an edge $d_{ij}$ of $D(X)$ breaks the generalized triangle inequality (eq. \ref{eq_gen_triang}), we have $d_{ij} > g (d^T_{ik}, d^T_{kj})$.
    Therefore, $d_{ij}$ is not an edge of the distance backbone ($b_{ij} = +\infty$), and there must exist at least one indirect path between $x_i$ and $x_j$ via a set other nodes $x_k \in K \subset X$ such that $\ell_{ij} < d_{ij}$, where $\ell_{ij}$ is the length of the indirect path given by eq. \ref{eq_length}.
    Since the closure computation (eq. \ref{eq_TC1} via isomorphism of eq. \ref{eq_isomorphism_formulae}) selects the shortest path between any pair of nodes ($f \equiv \min$), it cannot select the direct edge $d_{ij}$ for the shortest distance between $x_i$ and $x_j$, but rather the indirect path with smallest length: $d^T_{ij} = \min_K (\ell_{ij})$. Therefore,  $d_{ij}$ is not used to compute $d^T_{ij}$, nor the length of any shortest path that goes through $x_i$ and $x_j$---which must use $d^T_{ij}$ rather than $d_{ij}$, since $d^T_{ij} < d_{ij}$.
    Finally, if $d_{ij}$ does not break the generalized triangle inequality (eq. \ref{eq_gen_triang}), then it is an edge of the distance backbone and is sufficient to compute the smallest path length between $x_i$ and $x_j$, as there cannot be a shorter indirect path between them per eq. \ref{eq_gen_triang}: $b_{ij} = d_{ij} = d^T_{ij}$.
\end{proof}

\begin{corollarySI} [Backbone connectivity]
    %\label{th_backbone_is_connected}
    Given a connected distance graph $D(X)$, its distance backbone graph $B^g(X)$ is also a connected graph for any td-norm $g$.
\end{corollarySI}

\begin{proof}
    If graph $D(X)$ is connected, then there is a path between every pair of nodes $x_i$ and $x_j$ in graph, and its shortest-path distance closure with any td-norm $g$,  $D^{T,g}(X)$, is a complete (fully connected) graph.
    Since, per Theorem \ref{th_backbone_is_sufficient}, the closure of the backbone graph $B^g$ is sufficient to compute the same (complete) $D^{T,g}(X)$, the backbone graph must be connected as well.
\end{proof}

\begin{corollarySI} [Backbone Contains All Bridges]
    %\label{th_backbone_bridges}
    Given a distance graph $D(X)$, all its bridge edges are included in its distance backbone graph $B^g(X)$ for any td-norm $g$.
\end{corollarySI}

\begin{proof}
    A bridge is an edge whose deletion increases the graph's number of connected components.
    Therefore, if a bridge were not present on the backbone graph $B^g(X)$,  Theorem \ref{th_backbone_is_sufficient} and Corollary \ref{th_backbone_is_connected} would be false.
\end{proof}

\newpage

\subsection*{Additional Tables}
\label{section_SI_Tables}

\begin{table*}[!th]
    \centering
    \scalebox{0.8}{
    \begin{tabular}{c|llll}
    \toprule
    {} &  {U.S.-airports-500} &  {HCN-Coarse} & {C-Elegans} & {MyLib-Keywords}  \\
    \midrule
    $D(X)$    & 0.6175  & 0.7165  &  0.2924  & 0.9301  \\
    $D^m (X)$ & 0.1671  & 0.1318  &  0.0745  & 0.1919  \\
    $D^u (X)$ & 0.0     & 0.0     &  0.0     & 0.0     \\
    \bottomrule
    \end{tabular}
  }
  \caption{Watts and Strogatz (average) clustering coefficient for four distance graphs and their metric and ultra-metric backbones. Note that the clustering coefficient treats graphs as unweighted; that is, edge weights are assumed to be 1.}
  \label{tableClustering}
\end{table*}

\newpage

\subsection*{Additional Figures}
\label{section_SI_Figures}

\begin{figure*}[ht!]
    \centering
    \includegraphics[width=\textwidth]{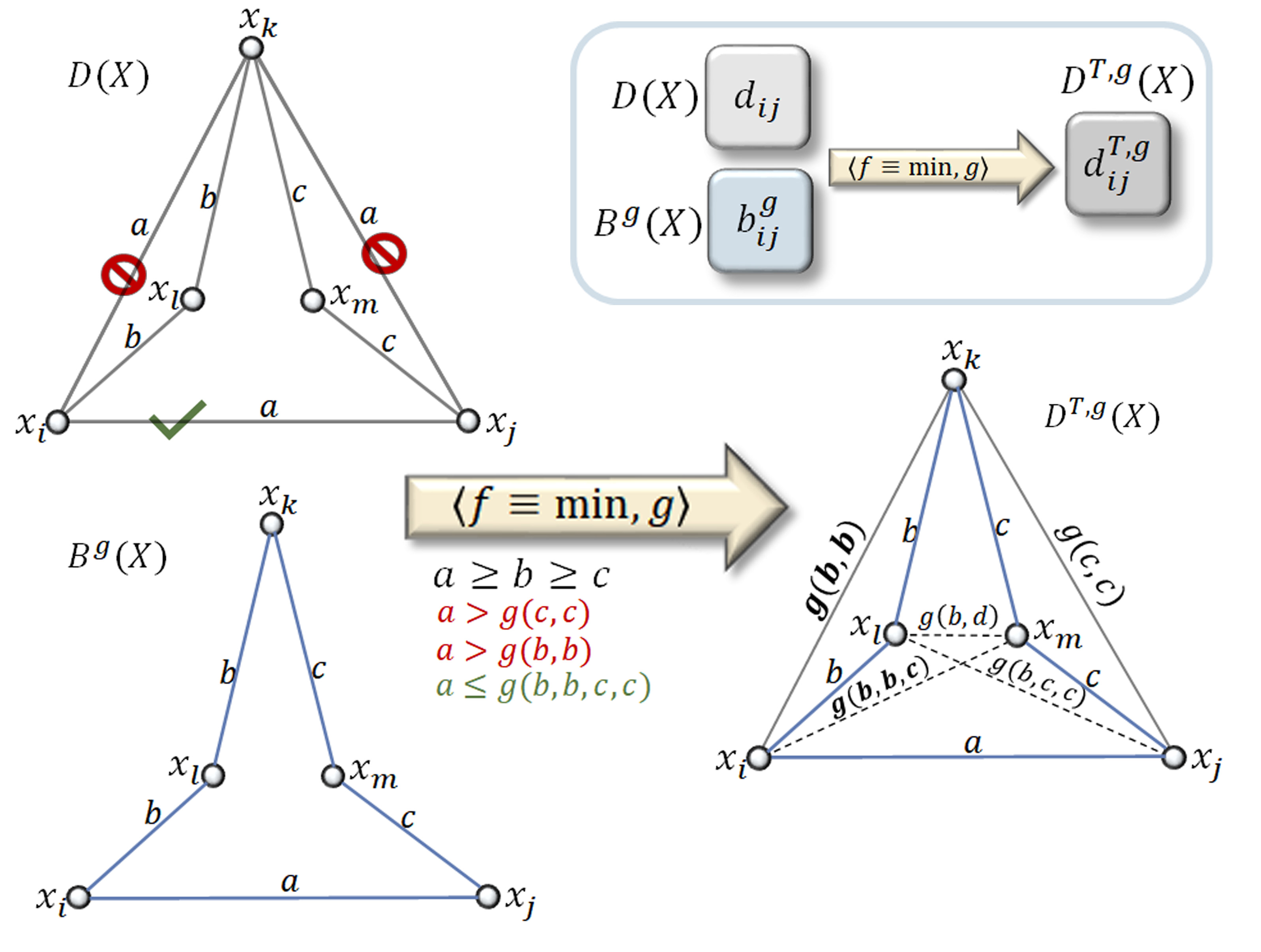}
    \caption{
        \textbf{Distance Backbone Example 2.}
        \textbf{Top, right}. Schematic of the distance closure $D^{T,g}(X)$ obtained from either the original distance graph $D(X)$ or its distance backbone $B^g(X)$.
        \textbf{Top, left}. Example distance graph of 5 nodes with edge distance weights constrained by $a\geq b \geq c$ and by $a > g(c,c)$ and $a > g(b,b)$ which break the generalized triangle inequality (eq. \ref{eq_gen_triang}) for nodes $\{x_j, x_k, x_m\}$ and $\{x_i, x_k, x_l\}$, as well as constraint $a \leq g(b,b,c,c)$ which enforces the inequality for nodes $\{x_i, x_k, x_j\}$, in indirect paths via nodes $\{x_l, x_m\}$ after closure.
        \textbf{Bottom}. The distance backbone graph $B^g(X)$ (left) and the distance closure graph $D^{T,g}(X)$ (right) for any td-norm $g$ given the edge weight constraints considered and td-norm properties; backbone (triangular) edges in blue, semi-triangular edges in gray, and (indirect) edges that do not exist in $D(X)$ appear in dashed gray in $D^{T,g}(X)$; bold edge weights denote changes from the example in Figure \ref{fig:distance_backbone}.
    }
    \label{fig:distance_backbone_2}
\end{figure*}

\begin{figure*}[ht!]
    \centering
    \includegraphics[width=\textwidth]{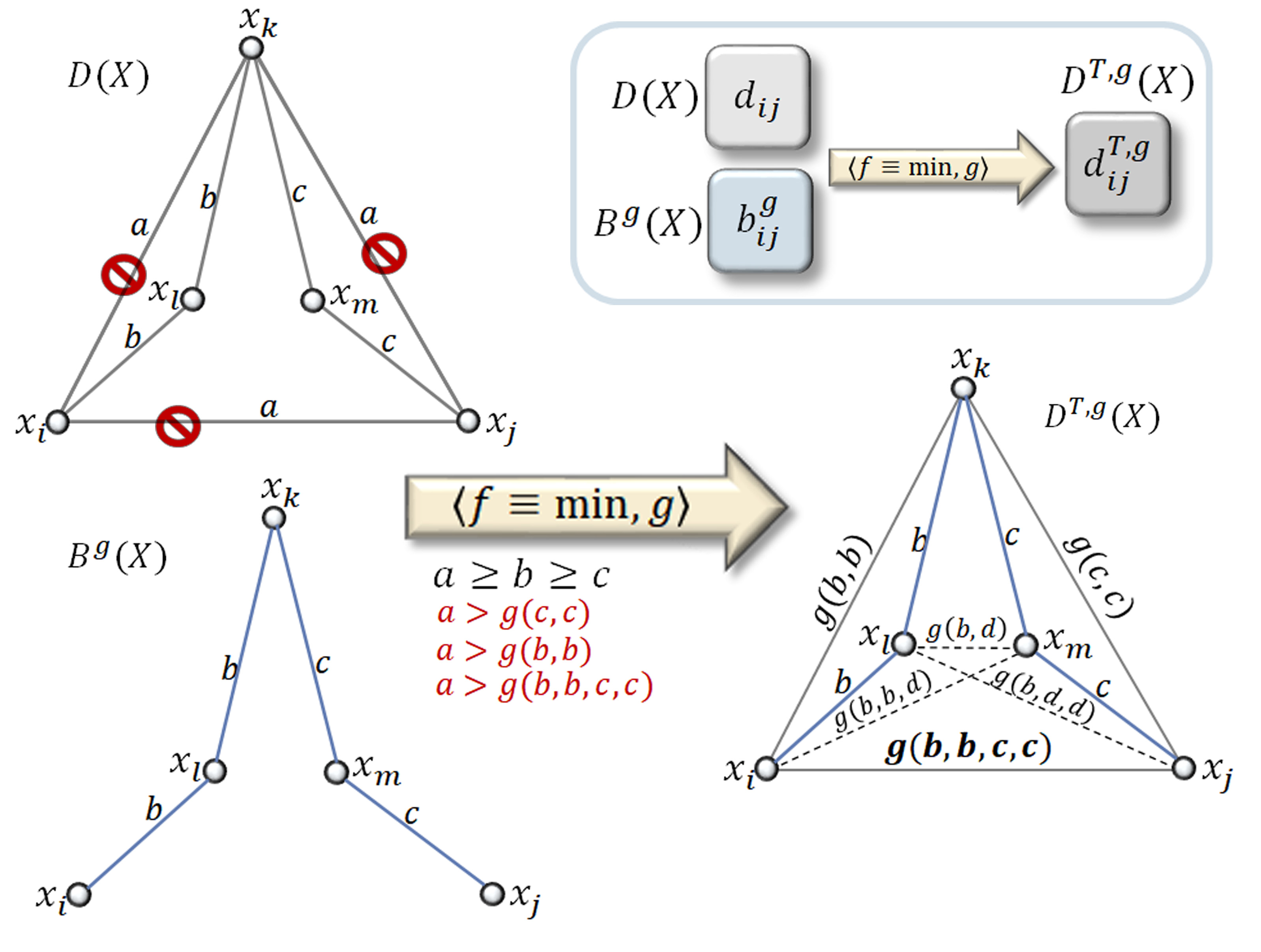}
    \caption{
        \textbf{Distance Backbone Example 3.}
        \textbf{Top, right}. Schematic of the distance closure $D^{T,g}(X)$ obtained from either the original distance graph $D(X)$ or its distance backbone $B^g(X)$.
        \textbf{Top, left}. Example distance graph of 5 nodes with edge distance weights constrained by $a\geq b \geq c$ and by $a > g(c,c)$, $a > g(b,b)$, and  $a \leq g(b,b,c,c)$ which break the generalized triangle inequality (eq. \ref{eq_gen_triang}) for nodes $\{x_j, x_k, x_m\}$, $\{x_i, x_k, x_l\}$, as well as $\{x_i, x_k, x_j\}$, the latter in indirect paths via nodes $\{x_l, x_m\}$ after closure.
        \textbf{Bottom}. The distance backbone graph $B^g(X)$ (left) and the distance closure graph $D^{T,g}(X)$ (right) for any td-norm $g$ given the edge weight constraints considered and td-norm properties; backbone (triangular) edges in blue, semi-triangular edges in gray, and (indirect) edges that do not exist in $D(X)$ appear in dashed gray in $D^{T,g}(X)$; bold edge weights denote changes from the example in Figure \ref{fig:distance_backbone_2}.
    }
    \label{fig:distance_backbone_3}
\end{figure*}

\begin{figure*}[ht!]
    \centering
    \includegraphics[width=\textwidth]{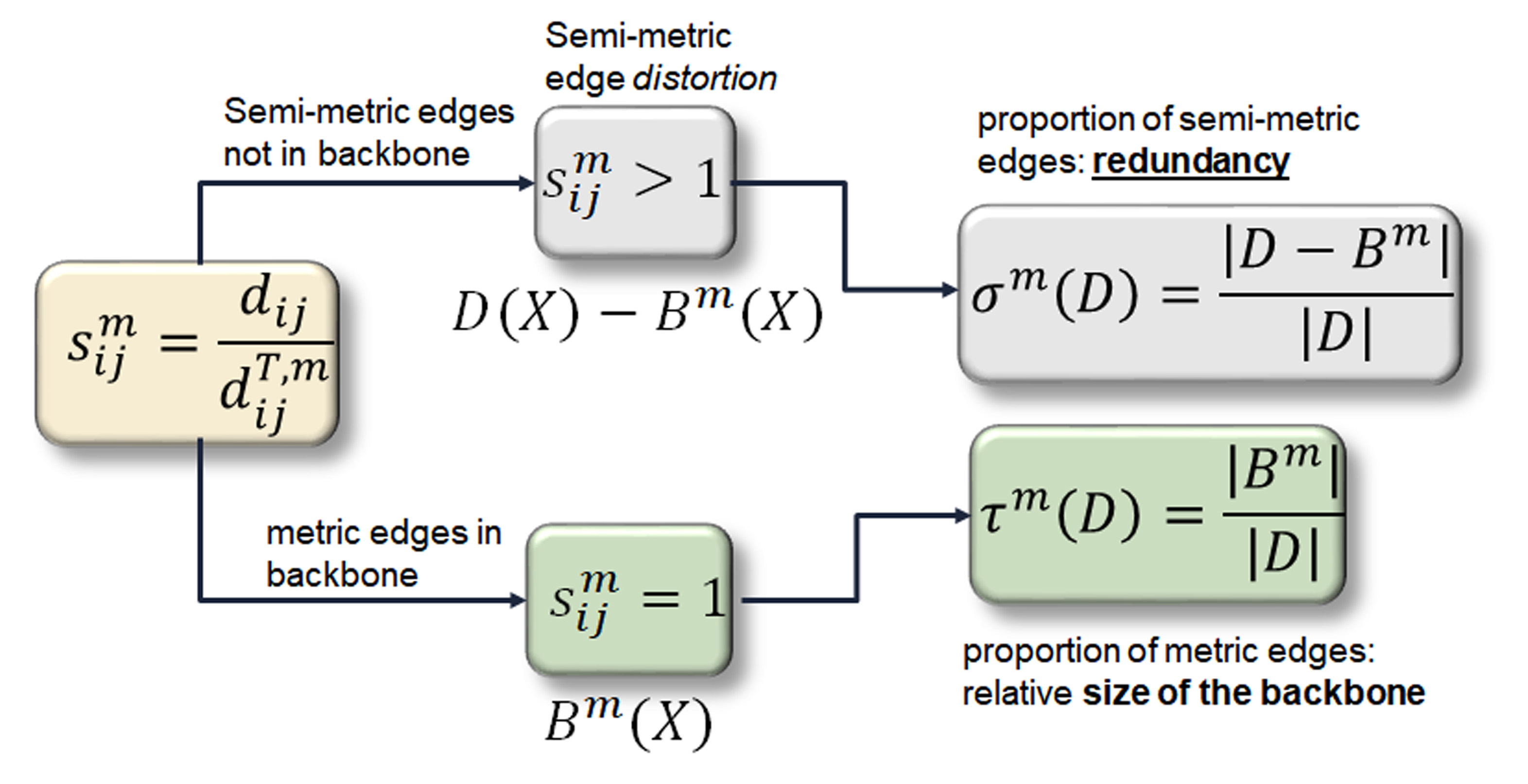}
    \caption{
        \textbf{Semi-metric measures.}
            Parsing of two types of distance graph edges, metric ($s^m_{ij}=1$) and semi-metric ($s^m_{ij} > 1$), from which graph-level measures of metric backbone size ($\tau^m(D)$) and (semi-metric) redundancy ($\sigma^m(D)$) derive, respectively.
            Measures apply only to the metric backbone that derives from the metric closure instantiated by td-norm $g \equiv  +$.
        }
    \label{fig:measures_metric}
\end{figure*}

\begin{figure*}[ht!]
    \centering
    \includegraphics[width=\textwidth]{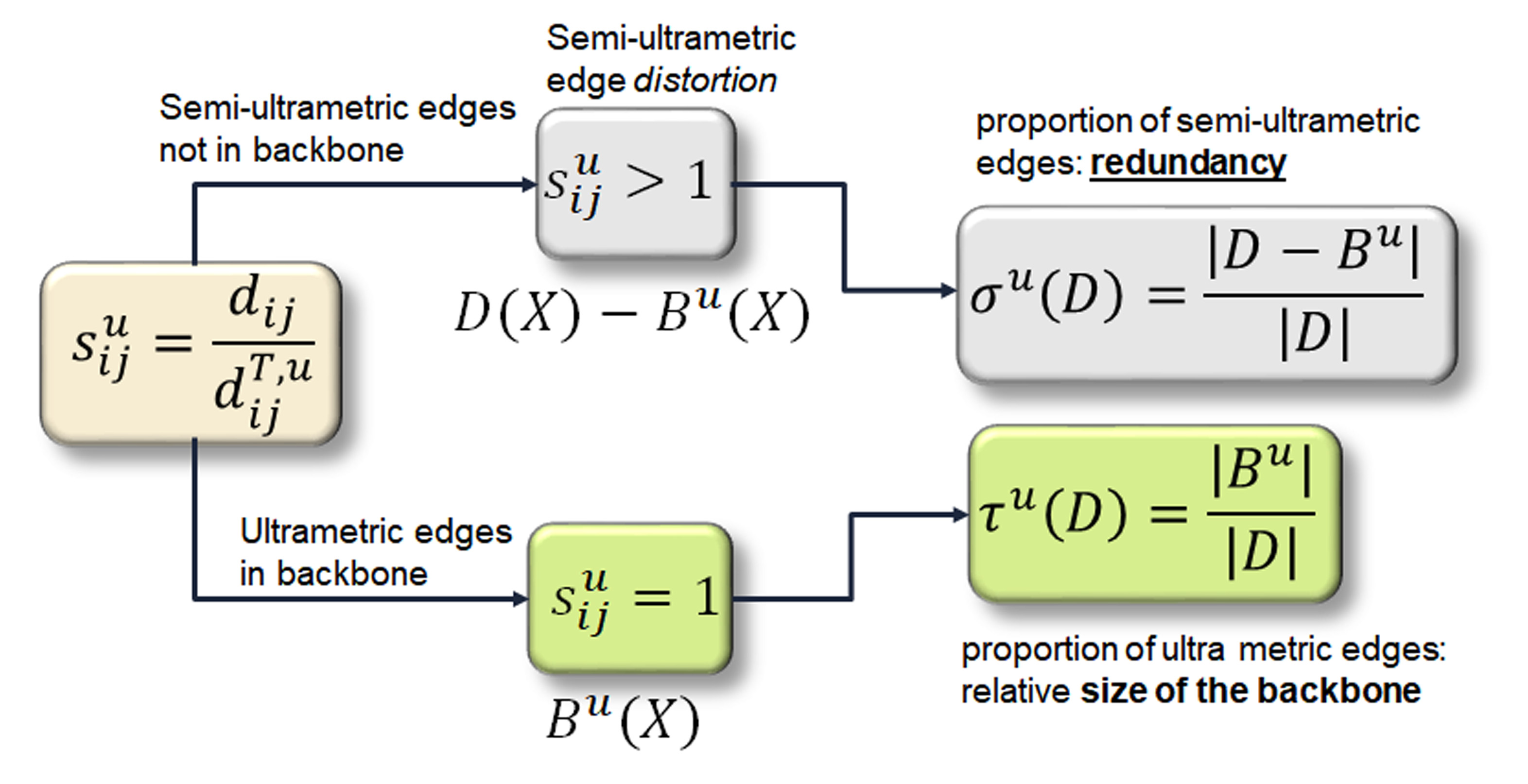}
    \caption{
        \textbf{Semi-ultrametric measures.}
        Parsing of two types of distance graph edges, ultrametric ($s^u_{ij}=1$) and semi-ultrametric ($s^u_{ij} > 1$), from which graph-level measures of ultrametric backbone size ($\tau^u(D)$) and (semi-ultrametric) redundancy ($\sigma^u(D)$) derive, respectively.
        Measures apply only to the ultrametric backbone that derives from the ultrametric closure instantiated by td-norm $g \equiv \max$.
    }
    \label{fig:measures_ultrametric}
\end{figure*}

\begin{figure*}[ht!]
    \centering
    \includegraphics[width=\textwidth]{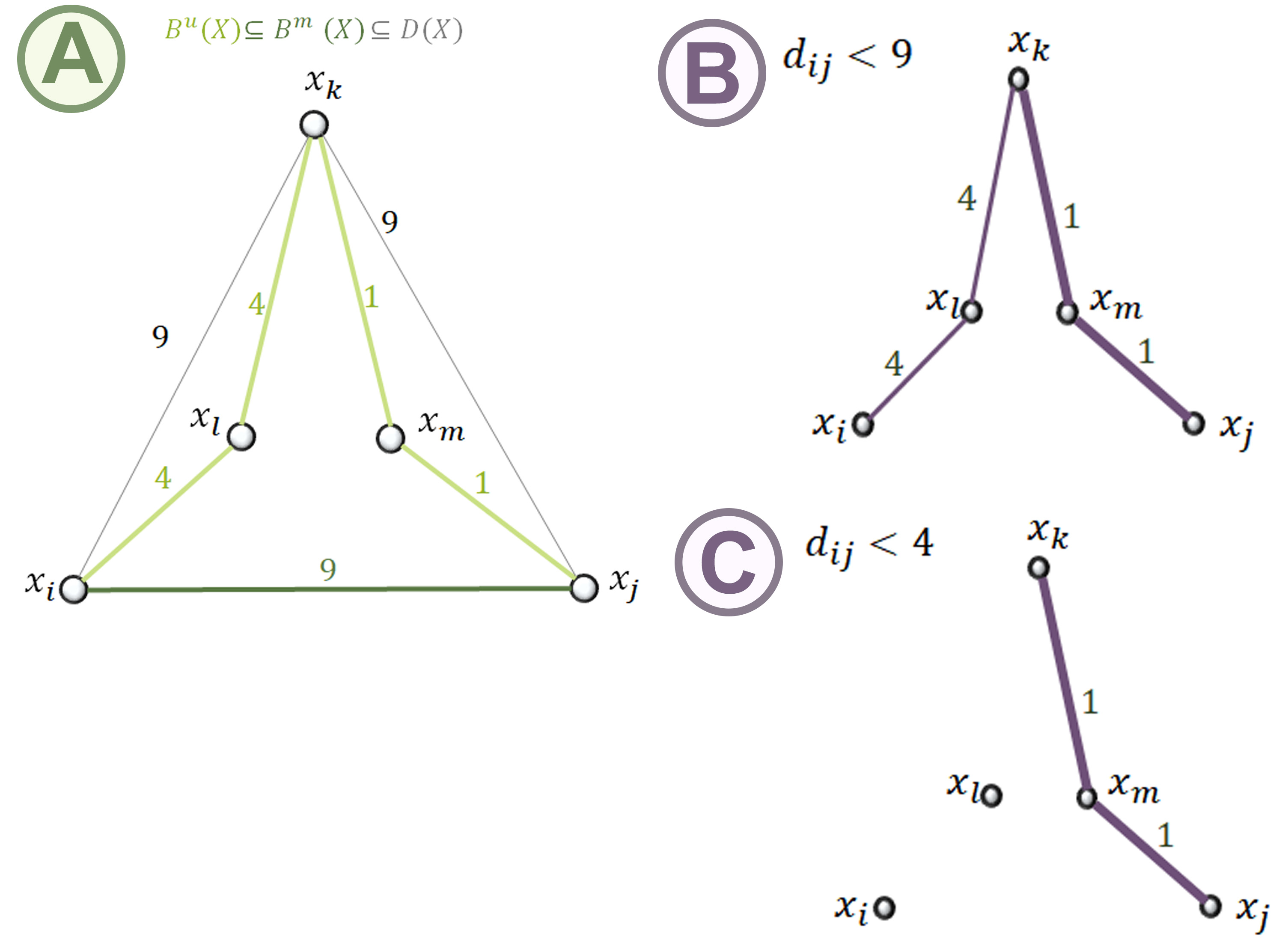}
    \caption{
        \textbf{Metric and ultra-metric backbones vs thresholding and Minimum Spanning Tree.}
        \textbf{A}. Example distance graph of 5 nodes with edge distance weights; edges $d_{ik}=9$ and $d_{jk}=9$ break the triangle inequality; edge $d_{ij}=9$ does not break the triangle inequality; metric backbone edges in green, ultra-metric backbone in lighter green, semi-metric edges in gray.
        \textbf{B}. Threshold graph for $d_{jk} < 9$ which is also the graph's MST (and in this case, also the ultra-metric backbone).
        \textbf{C}. Threshold graph for $d_{jk} < 4$.
        Notice that any threshold $d_{jk} >= 9$ returns the whole original distance graph, and threshold $d_{jk} < 9$ removes metric edge $d_{ij}$, thus affecting the shortest path between $x_i$ and $x_j$, $d^{T,m}_{ij}: 9 \rightarrow 10$; $d_{jk} < 4$ destroys original connectivity and shortest path distribution.
    }
    \label{fig:backbone_thresholds}
\end{figure*}

\end{document}